\documentclass[11pt]{article}
\usepackage{times}
\usepackage{xspace,comment,calc}
\usepackage{amssymb,amsmath,amsthm,amsfonts,amstext}
\usepackage[margin=1cm]{caption}
\usepackage[shortlabels]{enumitem}
\usepackage{multirow}
\usepackage{bbm}
\usepackage{comment}
\usepackage{url}
\usepackage{mdframed}
\usepackage{hyperref}
\hypersetup{colorlinks = true,linkcolor = blue, anchorcolor = blue, citecolor = blue, filecolor = red,urlcolor = red}
\usepackage[nameinlink]{cleveref}

\newtheorem{theorem}{Theorem}[section]
\newtheorem{lemma}[theorem]{Lemma}
\newtheorem{claim}[theorem]{Claim}

\newtheorem{corollary}[theorem]{Corollary}

{\theoremstyle{definition} \newtheorem{definition}[theorem]{Definition}
	}
{\theoremstyle{remark} \newtheorem{remark}[theorem]{Remark}}
\newtheorem{innercustomthm}{Theorem}
\newenvironment{customthm}[1]
{\renewcommand\theinnercustomthm{#1}\innercustomthm}
{\endinnercustomthm}

\newenvironment{proofof}[1]{\begin{proof}[Proof of #1]}{\end{proof}}

\newcommand{\bg}[1]{\medskip\noindent{\it #1}}

\newcommand{\R}{\ensuremath{\mathbb R}}
\newcommand{\Z}{\ensuremath{\mathbb Z}}

\newcommand{\A}{\ensuremath{\mathcal{A}}}
\newcommand{\C}{\ensuremath{\mathcal{C}}}

\newcommand{\F}{\ensuremath{\mathcal F}}

\newcommand{\D}{\ensuremath{\mathcal D}}

\newcommand{\W}{\ensuremath{\mathcal W}}

\newcommand{\Lc}{\ensuremath{\mathcal L}}

\newcommand{\OPT}{\ensuremath{\mathit{OPT}}}
\newcommand{\sm}{\ensuremath{\setminus}}
\newcommand{\es}{\ensuremath{\emptyset}}
\newcommand{\cost}{\ensuremath{\mathsf{cost}}}
\newcommand{\charge}{\ensuremath{\mathit{charge}}}
\newcommand{\argmin}{\ensuremath{\mathrm{argmin}}}

\newcommand{\frall}{\ensuremath{\text{ for all }}}

\newcommand{\ceil}[1]{\ensuremath{\left\lceil#1\right\rceil}}

\newcommand{\poly}{\operatorname{poly}}

\newcommand{\junk}[1]{}

\newcommand{\sse}{\subseteq}

\newcommand{\nbr}{\ensuremath{\mathsf{nbr}}}
\newcommand{\ctr}{\ensuremath{\mathsf{ctr}}}
\newcommand{\bq}{\ensuremath{\bar{q}}}
\newcommand{\bC}{\ensuremath{\bar{C}}}

\newcommand{\bF}{\ensuremath{\overline{F}}}
\newcommand{\brho}{\ensuremath{\overline{\rho}}}
\newcommand{\br}{\ensuremath{\overline{r}}}
\newcommand{\bell}{\ensuremath{\overline{\ell}}}
\newcommand{\bx}{\ensuremath{\overline{x}}}
\newcommand{\by}{\ensuremath{\overline{y}}}
\newcommand{\bz}{\ensuremath{\overline{z}}}

\newcommand{\assign}{\ensuremath{\leftarrow}}
\newcommand{\ttht}{\ensuremath{\tilde\theta}}
\newcommand{\tdt}{\ensuremath{t'}}
\newcommand{\tq}{\overset{\mathsf{\mbox{\tiny{int}}}}{q}}
\newcommand{\tw}{\ensuremath{\widetilde{w}}}
\newcommand{\tv}{\ensuremath{\widetilde{v}}}
\newcommand{\tz}{\overset{\mathsf{\mbox{\tiny{int}}}}{z}}
\newcommand{\tsg}{\ensuremath{\widetilde{\sigma}}}
\newcommand{\hb}{\ensuremath{\hat{b}}}

\newcommand{\hx}{\ensuremath{\hat{x}}}
\newcommand{\hq}{\ensuremath{\hat{q}}}
\newcommand{\hy}{\ensuremath{\hat{y}}}

\newcommand{\hw}{\ensuremath{\widehat{w}}}

\newcommand{\sg}{\ensuremath{\sigma}}
\newcommand{\dt}{\ensuremath{\delta}}
\newcommand{\Dt}{\ensuremath{\Delta}}
\newcommand{\e}{\ensuremath{\epsilon}}
\newcommand{\gm}{\ensuremath{\gamma}}
\newcommand{\Gm}{\ensuremath{\Gamma}}
\newcommand{\Om}{\ensuremath{\Omega}}
\newcommand{\ld}{\ensuremath{\lambda}}
\newcommand{\kp}{\ensuremath{\kappa}}
\newcommand{\al}{\ensuremath{\alpha}}
\newcommand{\tht}{\ensuremath{\theta}}

\newcommand{\lb}{\ensuremath{\mathsf{lb}}}

\newcommand{\high}{\ensuremath{\mathsf{hi}}}
\newcommand{\ub}{\ensuremath{\mathsf{ub}}}

\newcommand{\iopt}{\ensuremath{\mathit{opt}}}

\newcommand{\acost}{\ensuremath{\vec{c}}}
\newcommand{\down}{\ensuremath{\mskip2mu\downarrow}}

\newcommand{\obj}{\cost}
\newcommand{\vo}{\ensuremath{\vec{o}^{\down}}}
\newcommand{\vc}{\ensuremath{\vec{c}^{\down}}}

\newcommand{\pay}{P}
\newcommand{\thr}{\rho}
\newcommand{\thrlp}[1][{\thr}]{P\ensuremath{{}_{#1}}}
\newcommand{\thrdual}[1][{\thr}]{D\ensuremath{{}_{#1}}}

\newcommand{\load}{\ensuremath{\mathsf{load}}}
\newcommand{\lvec}{\ensuremath{\overrightarrow{\load}}}

\newcommand{\pset}{\ensuremath{\mathsf{POS}}}
\newcommand{\prox}{\ensuremath{\mathsf{prox}}}
\newcommand{\next}{\ensuremath{\mathsf{next}}}
\newcommand{\topl}{\ensuremath{\mathsf{Top}\text{-}\ell}\xspace}
\newcommand{\into}{\ensuremath{\mathrm{in}}}
\newcommand{\out}{\ensuremath{\mathrm{out}}}
\newcommand{\olblp}{\ensuremath{\mathsf{LP}}}
\newcommand{\nw}{\tw}

\newcommand{\ball}{\ensuremath{\mathbb B}}
\newcommand{\bopt}{\ensuremath{\mathit{Bopt}}}
\newcommand{\sgr}{\ensuremath{d}}
\newcommand{\hsgr}{\ensuremath{\widehat{\sgr}}}
\newcommand{\Scol}{\ensuremath{\C}}
\newcommand{\bon}{\ensuremath{\mathbbm{1}}}
\newcommand{\LP}{\mathsf{LP}}

\newcommand{\ny}{\hy}
\newcommand{\nq}{\hq}
\newcommand{\ocllp}[1][{\vec{t}}]{\ensuremath{\mathsf{CLP}_{#1}}}
\newcommand{\minmax}[1][m]{{#1}in-max\xspace}
\newcommand{\Exp}{\mathbf{Exp}}
\newcommand{\tm}{\ensuremath{\tau}}
\renewcommand{\bullet}{\cdot}

\newcommand{\cI}{\mathcal{I}}
\newcommand{\vv}{\vec{v}}
\newcommand{\vt}{\vec{t}}
\renewcommand{\epsilon}{\varepsilon}

\title{Approximation Algorithms for Minimum Norm 
	and \\
	Ordered Optimization Problems}
\author{{Deeparnab Chakrabarty\footnote{Dartmouth College, Email: \ttfamily{deeparnab@dartmouth.edu}}} \and {Chaitanya Swamy\footnote{University of Waterloo, Email: \ttfamily{cswamy@uwaterloo.edu}}}}
\date{}

\begin{document}
\maketitle	
	
	\begin{abstract}
\noindent
In many optimization problems, a feasible solution induces a multi-dimensional cost vector.
For example, in load-balancing a schedule induces a load vector across the machines.
In $k$-clustering, opening $k$ facilities induces an assignment cost vector across the clients.
Typically, one seeks a solution which either minimizes the sum- or the max- of this vector, and 
these problems (makespan minimization, $k$-median, and $k$-center) are classic NP-hard problems which have been extensively studied. \smallskip 

In this paper we consider the {\em minimum norm} optimization problem. Given an arbitrary monotone, symmetric norm,
the problem asks to find a solution which minimizes the norm of the induced cost-vector. 
These functions are versatile and model a wide range of problems under one umbrella.
We give a general framework to 
tackle the minimum norm problem, and illustrate its efficacy in the unrelated machine load balancing and $k$-clustering setting. 
Our concrete results are the following.
\begin{itemize}[noitemsep]
	\item We give constant factor approximation algorithms for the minimum norm load balancing problem in {\em unrelated} machines, and the minimum norm
	$k$-clustering problem. To our knowledge, our results constitute the {\em first} constant-factor approximations for such a general suite of objectives. 
	\item In load balancing with unrelated machines, we give a $2$-approximation for the problem of finding an assignment minimizing the sum of the largest $\ell$ loads, for any $\ell$.
	We give a $(2+\e)$-approximation for the so-called ordered load-balancing problem.
	\item For $k$-clustering, we give a $(5+\e)$-approximation for the ordered $k$-median problem significantly improving the constant factor approximations from
	Byrka, Sornat, and Spoerhase (STOC 2018) and Chakrabarty and Swamy (ICALP 2018).
	\item Our techniques also imply $O(1)$ approximations to the best {\em simultaneous optimization factor} 
	for any instance of the unrelated machine load-balancing and the $k$-clustering setting.
	To our knowledge, these are the first \emph{positive} simultaneous optimization results in these settings.
\end{itemize}
At a technical level, our main insight is connecting minimum-norm optimization to what we call \minmax ordered optimization. 
The main ingredient in solving the \minmax ordered optimization is {\em deterministic, oblivious rounding} of linear programming relaxations for load-balancing and clustering, and this technique may be of independent interest.
	\end{abstract}

\thispagestyle{empty}
\setcounter{page}{0}
\newpage

\section{Introduction}
In many optimization problems, a feasible solution induces a multi-dimensional cost vector.
For example, in the load balancing setting with machines and jobs, a solution is an assignment of jobs to machines, 
and this induces a {\em load} on every machine.
In a clustering setting with facilities and clients, a solution is to open $k$ facilities and connecting clients to the nearest open facilities, 
which induces an
{\em assignment cost} on every client. This multi-dimensional vector dictates the quality of the solution.
Depending on the application, oftentimes one minimizes either the sum of the entries of the cost vector,
or the largest entry of the cost vector. 
For example, in the load balancing setting, the largest entry of the load vector is the {\em makespan} of the assignment, and 
minimizing makespan has been extensively studied~\cite{LenstraST90,ShmoysT93,EbenlendrKS14,Svensson12,ChakrabartyKL15,JansenR17}. Similarly, in the clustering setting, the problem of minimizing the sum of assignment costs 
is the $k$-median problem, and the problem of minimizing the largest assignment cost is the $k$-center problem. Both of these are classic combinatorial optimization 
problems~\cite{HochbaumS85a,Gonzalez85,CharikarGST02,CharikarG99,JainV01,LiS16,ByrkaPRS14}. However, the techniques to study the sum-versions and max-versions are often different, and it is a natural and important to investigate what the complexity of these problems become if one
is interested in a different statistic of the cost vector.

In this paper, we study a far-reaching generalization of the above two objectives. We study the {\em minimum norm optimization} problem, where given an arbitrary monotone, symmetric norm 
$f$,  one needs to find a solution which minimizes the norm $f$ evaluated on the induced cost vector.
In particular, we study (a) the minimum norm load balancing problem which asks to find the assignment of jobs to (unrelated) machines which minimizes
$f(\lvec)$ where $\lvec$ is the induced load vector on the machines, and (b) the minimum norm $k$-clustering problem which asks to open $k$-facilities
minimizing $f(\acost)$ where $\acost$ is the induced assignment costs on the clients.

\emph{Our main contribution is a framework to study minimum norm optimization problems. Using this, we give constant factor approximation algorithms
	for the minimum norm unrelated machine load balancing and the minimum norm $k$-clustering problem} (\Cref{thm:normlb} and~\Cref{thm:normcl}). To our knowledge our results constitute the {\em first} constant-factor approximations for a general suite of objectives in these settings.
We remark that the above result is contingent on how $f$ is given. We need a ball-optimization oracle (see \eqref{balloracle} for more details), and for most norms
it suffices to have access to a {\em first-order} oracle which returns the (sub)-gradient of $f$ at any point.


Monotone, symmetric norms capture a versatile collection of objective functions. 
We list a few relevant examples below and point to the reader to~\cite{Bhatia13,BlasiokBCKY17,AndoniNNRW17} for a more comprehensive list of examples.

\begin{itemize}[leftmargin=10pt,parsep=0pt,listparindent=10pt,itemsep=1pt]
	\item {\bf $\ell_p$-norms.} Perhaps the most famous examples are $\ell_p$ norms where $f(\vv) := \left(\sum_{i=1}^n \vv_i^p\right)^{1/p}$ for $p\ge 1$. Of special interest are $p = \{1,2,\infty\}$. 
	For unrelated machines load-balancing, the $p = 1$ case is trivial while the $p = \infty$ case is makespan minimization. This has a $2$-approximation~\cite{LenstraST90,ShmoysT93} which has been notoriously difficult to beat. For the general $\ell_p$ norms, Azar and Epstein~\cite{AzarE05} give a $2$-approximation, with improvements given by~\cite{KumarMPS09,MakarychevS14}. For the $k$-clustering setting, the $p = \{1,2,\infty\}$ norms have been extensively studied over the years~\cite{Gonzalez85,HochbaumS85a,CharikarGST02,CharikarG99,JainV01,ByrkaPRS14,AhmadianNSW17}.
	One can also derive an $O(1)$-approximation for general $\ell_p$-norms using most of the algorithms\footnote{We could not find an explicit reference for this. The only work which we found that explicitly studies the $\ell_p$-norm minimization in the $k$-clustering  setting is by Gupta and Tangwongsan~\cite{GuptaT08}. They give a $O(p)$-approximation using local-search and prove that local-search can't do any better. However, $\ell_p^p$-``distances'' satisfy relaxed triangle inequality, in that, $d(u,v) \leq 2^p(d(u,w) + d(w,v))$. The algorithms of Charikar et al~\cite{CharikarGST02} and Jain-Vazirani~\cite{JainV01} need triangle inequality with only ``bounded hops'' and thus give $C^p$-approximations for the $\ell_p^p$ ``distances''. In turn this implies a constant factor approximation for the $\ell_p$-norm. } for the $k$-median problem.
	
	\item {\bf Top-$\ell$ norms and ordered norms.} Another important class of monotone, symmetric norms is the {\em Top-$\ell$-norm}, which given a vector $\vv$ returns the sum of the largest $\ell$ elements. These norms are another way to interpolate between the $\ell_1$ and the $\ell_\infty$ norm. 
	
	A generalization of the Top-$\ell$ norm optimization is what we call the {\em ordered norms}. 
	The norm is defined by a non-increasing, non-negative vector $w\in \R_+^n$ with $w_1\geq w_2 \ge \cdots \ge w_n \ge 0$. 
	Given these weights, the $w$-ordered, or simply, ordered norm of a vector $\vv\in \R_+^n$ is defined as
	$
	\obj(w;\vv) := \sum_{i=1}^n w_i \vv^{\down}_i
	$	
	where $\vv^{\down}$ is the entries of $\vv$ written in non-increasing order itself. It is not hard to see that the ordered norm is a non-negative linear combination of the Top-$\ell$ norms. 
	
	For load balancing in unrelated machines, we are not aware of any previous works studying these norms. \emph{We give a $2$-approximation for the Top-$\ell$-load balancing, and a $(2+\e)$-approximation for ordered load balancing} (\Cref{thm:top-l:loadbal} and~\Cref{thm:single-ordered:loadbal}). Note that the case of $\ell=1$ for Top-$\ell$-load balancing corresponds to makespan minimization for which beating factor $2$ is an open problem.
	
In $k$-clustering, the \topl optimization problem is called the $\ell$-centrum problem, and the ordered-norm minimization problem is called the ordered $k$-median problem. Only recently, 
a $38$-factor~\cite{ByrkaSS18} and $18+\e$-factor~\cite{ChakrabartyS18} approximation algorithm was given for the ordered $k$-median problem.
	\emph{We give a much improved $(5+\e)$-factor approximation algorithm for the ordered $k$-median problem} (\Cref{thm:ordered-median}).

	\item {\bf Min-max ordered norm.} Of particular interest to us is what we call the {\em min-max ordered optimization} problem. In this, we are given $N$ non-increasing, non-negative weight vectors $w^{(1)},\ldots,w^{(N)} \in \R^n_+$, and the goal is to find a solution $\vv$ which minimizes $\max_{r=1}^N \obj(w^{(r)};\vv)$. This is a monotone, symmetric norm since it is a maximum over a finite collection of monotone, symmetric norms. 	
	
	\emph{
		One of the main insights of this paper is that the minimum norm problem reduces to min-max ordered optimization} (\Cref{thm:minnorm}). In particular, we show that the value of any monotone, symmetric norm can be written as the maximum
		of a collection of (possibly infinite) ordered norms; this result may be of independent interest in other applications involving such norms~\cite{AndoniNNRW17,BlasiokBCKY17}.
	
	\item {\bf Operations.} One can construct monotone, symmetric norms using various operations such as (a) taking a nonnegative linear combination of monotone, symmetric norms; 
	(b) taking the maximum over any finite collection of monotone, symmetric norms; 
	(c) given a (not-necessarily symmetric) norm $g:\R^n\to\R_+$, setting
	$f(v):=g(v^{\down})$, or 
	$f(v):= \Exp_{\pi}[g\bigl(\{v_{\pi(i)}\}_{i\in[n]}\bigr)]$ where $\pi$ is a random permutation of $[n]$; 
	(d) given a monotone, symmetric norm $g:\R^k\to\R_+$, where $k\leq n$, 
	setting $f(v)=\sum_{S\sse[n]:|S|=k}g(\{v_i\}_{i\in S})$. 
	The richness of these norms makes the minimum-norm optimization problem a 
	versatile and appealing model which captures a variety of optimization problems  
	under one umbrella. 
	
	As an illustration, consider the following stochastic optimization problem  in the clustering setting (this is partly motivated by the stochastic fanout model described in ~\cite{KabiljoKPPS17} for a different setting). 
	We are given a universe of plausible clients, and a symmetric probability distribution over actual client instances. Concretely, say, each client materializes i.i.d with probability $p \in (0,1)$. The problem is to open a set of $k$ facilities such that the {\em expected maximum} distance of an instantiated client to an open facility is minimized.
	The expectation is indeed a norm (apply part (d) operation above) and thus we can get a constant factor approximation for it. In fact, the {\em expected maximum} for the i.i.d case is an ordered-norm, and so we can get a $(5+\e)$-approximation for this particular stochastic optimization problem.
	
	\item {\bf General Convex Functions.} One could ask to find a solution minimizing a general convex function of the cost vector. In general, such functions can be arbitrarily sharp and this precludes any non-trivial approximation. For instance in the clustering setting, consider the convex function $C(\acost)$ which takes the value $0$ if the sum of $\acost_j$'s (that is the $k$-median objective) is less than some threshold, and $\infty$ otherwise; for this function, it is NP-hard to get a finite solution. Motivated thus, Goel and Meyerson~\cite{GoelM06} call a solution $\vv$ an $\alpha$-approximate solution if $C(\vv/\alpha)\leq \iopt$ where $\vv$ is the induced cost vector, $\vv/\alpha$ is the coordinate-wise scaled vector, and $\iopt = \min_{\vec{w}}C(\vec{w})$. It is not hard to see\footnote{Consider the monotone, symmetric norm $f(x) := \min\{t: C(|x|/t)\le \iopt\}$.
		By definition $f(\vec{o})=1$, and so a $\al$-approximate min-norm solution $\vv$ satisfies $f(\vv)\le \al$, implying $C(\vv/\al)\le\iopt$.
		The definition requires knowing the value of $\iopt$ which can be guessed using binary search.}	
	that a constant factor approximation for monotone, symmetric norm-minimization  implies a constant-approximate solution for any monotone, symmetric convex function. In particular, for the load-balancing and clustering setting we achieve this.

\end{itemize}
\noindent
{\bf Connections and implications for simultaneous/fair optimization.} 
In the minimum-norm optimization problem, we are given a fixed norm function $f$ and we wish to find a solution minimizing $f(\vv)$ where $\vv$ is the cost-vector induced by the solution. In {\em simultaneous optimization}~\cite{KumarK06,GoelM06}, the goal is to find a solution $\vv$, which {\em simultaneously} approximates all norms/convex functions. 
Such solutions are desirable as they possess certain fairness properties.
More precisely, the goal is to find a solution inducing a cost vector $\vv$ which is simultaneous $\alpha$-approximate, that is, 
$g(\vv) \leq \alpha \cdot \iopt(g)$
for {\em all} monotone, symmetric norms $g:\R^n_+ \to \R_+$, where $\iopt(g) := \min_{\vec{w}} g(\vec{w})$.

Simultaneous optimization is clearly a much stronger goal than what we are shooting for, in that, if one can find a solution which 
is simultaneous $\alpha$-approximate, then this solution is clearly an $\alpha$-approximation for a fixed norm. 
It is rather remarkable that in the setting of load balancing with {\em identical jobs}, and even in the \emph{restricted assignment} setting where the jobs have fixed load but can be allocated only on a subset of machines, one can always achieve~\cite{AlonAWY98,AzarERW04,GoelM06} a simultaneous $2$-approximate solution. Unfortunately, for unrelated (even related) machines~\cite{AzarERW04} and $k$-clustering~\cite{KumarK06}, there are impossibility results ruling out the {\em existence} of any simultaneous $\alpha$-approximate solutions for constant $\alpha$. These impossibilities also show that the techniques used in~\cite{AlonAWY98,AzarERW04,GoelM06} are not particularly helpful when trying to optimize a {\em fixed} norm, which is the main focus in our paper.

Nevertheless, the techniques we develop give  {\em $O(1)$ approximations to the best simultaneous approximation factor possible in any instance of unrelated machines load-balancing and $k$-clustering} (\Cref{thm:simul}).
Fix an unrelated machines load balancing instance $\cI$. Let $\alpha^*_\cI$ be the smallest $\alpha$ for which there is a solution to $\cI$ which 
is simultaneous $\alpha$-approximate. Note that $\alpha^*_\cI$ could be a constant for a nice instance $\cI$; the impossibility result mentioned above states $\alpha^*_\cI$ can't be a constant for \emph{all} instances. It is natural, and important, to 
ask whether for such nice instances can one get constant factor simultaneous approximate solutions? We answer this in the affirmative.
We give an algorithm which, for any instance $\cI$, returns a solution 
inducing a load vector $\vv'$ such that $g(\vv') \leq O(\alpha^*_I)\cdot \iopt(g)$ for all monotone, symmetric norms simultaneously. 
We can also obtain a similar result for the $k$-clustering setting. 
These seem to be the first {\em positive} results on simultaenous optimization in these settings.
We remark that our algorithm is not a generic reduction to the minimum norm optimization, but is an artifact of our techniques developed to tackle the problem.
 \medskip

\noindent
{\bf Other related work.}
The ordered $k$-median and the $\ell$-centrum problem have been extensively studied in the Operations Research literature for more than two decades (see, e.g. the books~\cite{book1,book2}); we point the interested reader to these books, or the paper by Aouad and Segev~\cite{AouadS18}, and references within for more information on this perspective. From an approximation algorithms point of view, Tamir~\cite{Tamir01} gives
the first $O(\log n)$-approximation for the $\ell$-centrum problem, and  Aouad and Segev~\cite{AouadS18} give the first $O(\log n)$-approximations for the ordered $k$-median problem. 
Very recently, Byrka, Sornat, and Spoerhase~\cite{ByrkaSS18} and our earlier paper~\cite{ChakrabartyS18} give the first constant-factor approximations for the $\ell$-centrum and ordered $k$-median problems.
Another recent relevant work is of Alamdari and Shmoys~\cite{AlamdariS17} who consider the $k$-{\em centridian} problem where the objective is a weighted average of the $k$-center and the $k$-median objective (a special case of the ordered $k$-median problem);~\cite{AlamdariS17} give a constant-factor approximation algorithm for this problem.

In the load balancing setting, research has mostly focused on $\ell_p$ norms; we are not aware of any work studying the \topl optimization question in load balancing. For the $\ell_p$-norm Awerbuch et al.~\cite{AwerbuchAGKKV95} give a $\Theta(p)$-approximation for unrelated machines; their algorithm is in fact an {\em online} algorithm. Alon et al.~\cite{AlonAWY98} give a PTAS for the case of identical machines. This paper~\cite{AlonAWY98} also shows a polynomial time algorithm in the case of restricted assignment (jobs have fixed processing times but can't be assigned everywhere) with unit jobs which is optimal {\em simultaneously} in all $\ell_p$-norms. Azar et al.~\cite{AzarERW04} extend this result to get a $2$-approximation algorithm {\em simultaneously} in all $\ell_p$ norms in the restricted assignment case. This is 
generalized to a simultaneous $2$-approximation in all symmetric norms (again in the restricted assignment situation) by Goel and Meyerson~\cite{GoelM06}. As mentioned in the previous subsection, Azar et al.~\cite{AzarERW04} also note that even in the related machine setting, no constant factor approximation is possible simultaneously even with the $\ell_1$ and $\ell_\infty$ norm. For unrelated machines, for any fixed $\ell_p$ norm Azar and Epstein~\cite{AzarE05} give a $2$-approximation via convex programming. The same paper also gave a $\sqrt{2}$-approximation for the $p=2$ case. These factors have been improved (in fact for any constant $p$ the approximation factor is $< 2$) by Kumar et al.~\cite{KumarMPS09} and Makarychev and Sviridenko~\cite{MakarychevS14}. We should mention that the techniques in these papers are quite different from ours and in particular these strongly use the fact that the $\ell_p^p$ cost is separable.
Finally, in the clustering setting, Kumar and Kleinberg~\cite{KumarK06} and Golovin et al.~\cite{GolovinGKT08} give simultaneous constant factor approximations in all $\ell_p$ norms, but their results are \emph{bicriteria results} in that they open $O(k\log n)$ and $O(k\sqrt{\log n})$ facilities instead of $k$.

\section{Technical overview and organization}\label{sec:techideas}

We use this section to give an overview of the various technical ideas in this paper and point out the reader to where more details can be found. \medskip

\noindent
{\bf First approach and its failure.} Perhaps the first thing one may try for the minimum-norm optimization problem is to write a {\em convex program} 
$\min f(\vv)$ where $\vv$ ranges over {\em fractional} cost vectors, ideally, convex combinations of integral cost vectors. 
If there were a {\em deterministic} rounding algorithm
which given an optimal solution $\vv^*$ could return a solution $\vv$ such that for every coordinate $\vv_j \leq \rho\vv^*_j$, then by homogeneity of $f$, we would get a 
$\rho$-approximation. Indeed, for some optimization problems such a rounding is possible.
Unfortunately, for both unrelated load balancing and $k$-clustering, this strategy is a failure as 
there are simple instances for both problems, where even when $\vv^*$ is a convex combination of integer optimum solutions, no such rounding, with constant $\rho$, exists. In particular, the {\em integrality gaps} of these convex programs are unbounded.
\medskip

\noindent
{\bf Reduction to min-max ordered optimization (\Cref{minnorm}).} 
Given the above failure,
at first glance, it may seem hard to be able to reason about a general norm.
One of the main insights of this paper is that the monotone, symmetric norm minimization problem reduces to min-max ordered optimization.
This is a key conceptual step since it allows us a foothold in arguing about the rather general problem. Our result may also be of interest in other settings dealing with symmetric norms.
In particular, we show that given any monotone, symmetric norm $f$, the function value at any point $f(x)$ is equal to $\max_{w \in \Scol} \obj(w;x)$ (Lemma~\ref{normeqv}) where $\Scol$ is a potentially infinite family of non-increasing subgradients on the unit-norm ball. 
That is, $f(x)$ equals the maximum over a collection of ordered norms.
Thus, finding the $x$ minimizing $f(x)$ boils to the min-max ordered-optimization problem. The snag is that collection of weight vectors could be infinite. This is where the next simple, but extremely crucial, technical observation helps us. \medskip

\noindent
{\bf Sparsification idea (\Cref{sparsify}).} 
Given a non-increasing, non-negative weight vector $w\in \R_+^n$,
the ordered norm of a vector $\vv\in \R^n_+$ is $\obj(w;\vv) := \sum_{i=1}^n w_i \vv^{\down}_i$. 
The main insight is that although $w$ may have all its $n$-coordinates distinct, only a few {\em fixed} coordinates matter. 
More precisely, if we focus only on the coordinates $\pset := \{1,2,4,8,\cdots,\}$ and define a $\tw$-vector with $\tw_i = w_i$ if $i\in \pset$, and $\tw_i = w_\ell$ where
$\ell$ is the nearest power of $2$ larger than $i$, then it is not too hard to see $\obj(\tw;x)\leq \obj(w;x) \leq 2\obj(\tw;x)$. 
Indeed, one can increase the granularity of the coordinates to (ceilings of) powers of $(1+\dt)$ to get arbitrarily close approximations where the number of relevant coordinates is $O(\log n/\dt)$.

The above sparsification shows that for the ordered norms, one can just focus 
on weight vectors which have breakpoints in {\em fixed} locations {\em independent} of what the weight vectors are. Note that other kinds of sparsification which round every coordinate of a weight vector to the nearest power of $(1+\dt)$ don't have this {\em weight-independence} in the positions of breakpoints. This fixedness of the locations (and the fact that there are only logarithmically many of them) allows us to form a polynomial sized $\epsilon$-net of weight vectors. More precisely, for any weight vector $w$, there is another weight vector $w'$ in this net such that for any vector $\vv$, $\obj(w;\vv)$ and $\obj(w';\vv)$ are within multiplicative $(1\pm \dt)$. In particular, this helps us bypass the problem of ``infinitely many vectors'' in $\Scol$ described above. \medskip

\noindent
{\bf Ordered optimization and proxy costs (\Cref{sec:proxy}).} Now we focus on min-max ordered optimization. First let us consider just simple ordered optimization, and in particular, just Top-$\ell$ optimization. To illustrate the issues, let us fix the optimization problem to be load balancing on unrelated machines. 
One of the main technical issues in tackling the Top-$\ell$ optimization problem is that one needs to find an assignment such that sum of loads on a set of $\ell$ machines is minimized, but this set of machines itself depends on the assignment. Intuitively, the problem would be easier (indeed, trivial) if we could sum the loads over all machines. Or perhaps sum some {\em function} of the loads, but over {\em all} machines. Then perhaps one could write a linear/convex program to solve this problem fractionally, and the objective function would be clear.
This is where the idea of {\em proxy costs} comes handy. We mention that this idea was already present in the paper of Aouad and Segev~\cite{AouadS18}, and then in different forms in Byrka et al~\cite{ByrkaSS18} and our earlier paper~\cite{ChakrabartyS18}.

The idea of this proxy cost is also simple. Suppose we knew what the $\ell$th largest load would be in the optimal solution -- suppose it was $\rho$. Then the Top-$\ell$ load can be written as 
$\ell \cdot \rho + \sum_{i: ~\textrm{all machines}} (\load(i) - \rho)^+$, where we use $(z)^+ := \max(z,0)$. This is the {\em proxy-cost} of the Top-$\ell$ norm given parameter $\rho$.
Note that the summation is over {\em all} machines; however, the summand is not the load of the machine but a function $h_\rho(\load(i))$ of the load. Furthermore, we could assume by binary search that we have a good guess of $\rho$.

For ordered optimization, first we observe that $\obj(w;\vv)$ can be written as a non-negative linear combination of the Top-$\ell$ norms (see~\Cref{topltoord}). In particular, if we have the guesses of the $\ell$th largest loads for all $\ell$, then we could write the proxy cost of $\obj(w;\vv)$. However, guessing $n$ of the $\rho_\ell$'s would be infeasible. This is where the sparsification idea described above comes handy again. Since the only relevant positions of $w$ to define $\tw$ are the ones in $\pset$, one just needs to guess approximations for $\rho_\ell$s only in these positions to define the proxy function. And this again can be done in polynomial time. Once again, what is key is that positions are {\em independent} of the particular weight vector. This is {\em key for min-max ordered optimization}.
Even though there are $N$ different weight functions, their sparsified versions have the {\em same} break points, and their proxy functions are defined using these same, logarithmically many break points. \medskip

\noindent
{\bf LP relaxations and deterministic oblivious rounding (\Cref{sec:minmaxapproach,sec:loadbal,sec:clustering}).} One can use the proxy costs to write linear programming relaxations for the problems at hand (in our case, load balancing and $k$-clustering). Indeed, for $k$-clustering, this was the approach taken by Byrka et al.~\cite{ByrkaSS18} and our earlier work~\cite{ChakrabartyS18} for ordered $k$-median.
With proxy costs, the LP relaxation for ordered $k$-median is the usual LP but the objective has {\em non-metric} costs. Nevertheless, both the papers showed constant integrality gaps for these LPs (our proxy costs were subtly different but within $O(1)$-factors). For load-balancing, the usual LP has a bad gap, and one needs to add additional constraints. After this, however, we can indeed show the LP has an integrality gap of $\leq 2$ (this is established in~\Cref{sec:loadbal:splcase}).

However, it is not at all clear how to use this LP for {\em min-max ordered problems} with multiple weight functions. The algorithms of Byrka et al~\cite{ByrkaSS18} are randomized which bound the {\em expected} cost of the ordered $k$-median; with multiple weights, this won't help solve the min-max problem unless one can argue very sharp concentration properties of the algorithm.
The same is true for our load-balancing algorithm. These algorithms can be derandomized, but these derandomizations lead to algorithms which use the (single) weight function crucially, and it is not clear at all how to minimize the max of even two weight functions. The primal-dual algorithm in~\cite{ChakrabartyS18} suffers from the same problem. Our approach in this paper is to consider {\em deterministic} rounding of the LP solution which are {\em oblivious to the weight} vectors. We can achieve this for the LP relaxations we write for load balancing and $k$-clustering (although we need to strengthen the latter furthermore). We defer further technical overview to~\Cref{sec:minmaxapproach}, and then give details for load-balancing in~\Cref{sec:loadbal} and for $k$-clustering in~\Cref{sec:clustering}. After reading~\Cref{sec:minmaxapproach}, the sections on load balancing and clustering can be read in any order.\medskip

\noindent
{\bf Extensions: connections to simultaneous optimization (\Cref{fairness}).} We end the paper by showing the power of deterministic, weight-oblivious rounding to give constant factor approximations to instance-optimal algorithms for simultaneous optimization. The key idea stems from the result of Goel and Meyerson~\cite{GoelM06}, which itself stems from the majorization theory of Hardy, Littlewood, and Polya~\cite{HardyLP34}, that if we want to simultaneously optimize all monotone, symmetric, norms, then it suffices to simultaneously optimize all the Top-$\ell$ norms. If the best simultaneous optimization for a given instance is $\alpha^*_{\cI}$, then one can cast this as a multi-budgeted ordered optimization problem where we need to find a solution where the ordered-norm with respect to the $r$th weight vector is at most some budget $B_r$. Once again, if we have a good deterministic, weight-oblivious rounding algorithm for the LP relaxation, the multi-budgeted ordered optimization problem can also be easily solved. As a result, for any load-balancing and $k$-clustering instance, we get $O(1)$-approximations to the best simultaneous optimization factor possible for that instance.


\section{Preliminaries}\label{prelim}

Solutions to the optimization problems we deal with in this paper induce cost vectors.
We use $\vec{v}$ to denote them when talking about problems in the abstract. 
In load-balancing, 
the vector of the loads on machines is denoted by $\lvec$, or $\lvec_\sg$ if $\sg$ is the assignment of jobs.
In $k$-clustering, we the vector of assignment costs of clients is denoted as $\acost$. 
We always use $\vec{o}$ to denote the cost vector in the optimum solution.

For an integer $n$, we use $[n]$ to denote the set $\{1,\ldots,n\}$.
For a vector $\vec{v}\in\R^n$, we use $\vec{v}^{\down}$ to denote the vector $v$ with
coordinates sorted in non-increasing order. That is, we have $\vec{v}^{\down}_i=\vec{v}_{\pi(i)}$, where
$\pi$ is a permutation of $[n]$ such that $\vec{v}_{\pi(1)}\geq \vec{v}_{\pi(2)}\geq\ldots
\vec{v}_{\pi(n)}$. 

Throughout the paper, we use $w$ (with or without superscripts) to denote a non-increasing, non-negative weight vector. The dimension of this vector is the dimension of the cost vector.
In the abstract, we use $n$ to denote this dimension; so $w\in \R_+^n$ and $w_1\ge w_2\ge \cdots\ge w_n\ge 0$. 
We use $\tw$ to denote the ``sparsified'' version of the weight vector $w$ which is defined in \Cref{sparsify}.\smallskip

\noindent
{\bf Ordered and top-$\ell$ optimization.} Given a weight vector $w$ as above, 
the {\em ordered optimization} problem asks to find a solution with induced cost vector $\vec{v}$ which minimizes
$\obj(w;\vec{v}) := \sum_{i=1}^n w_i \vec{v}^{\down}_i$. This is the $w$-ordered norm, or simply ordered norm of $\vec{v}$.
We denote the special case of when $w$ is a $\{0,1\}$ vector {\em \topl optimization}.
That is, if $w_1 = \cdots w_\ell = 1$ and $w_i = 0$ otherwise, the problem asks to find a solution $\vec{v}$ minimizing the sum of the 
$\ell$ largest entries. We use the notation $\obj(\ell;\vec{v})$ to denote the cost of the \topl optimization problem. 
In the literature in the $k$-clustering setting, 
the \topl optimization problem is called the {\em $\ell$-centrum problem}, and the ordered optimization problem
is called the {\em ordered $k$-median} problem.
\smallskip
%
%

\noindent
{\bf \minmax[M] and multi-budgeted ordered optimization.} In a significant generalization of
ordered optimization, 
we are given 
multiple non-increasing weight vectors $w^{(1)},\ldots,w^{(N)}\in\R_+^n$, and \minmax ordered optimization asks to find a
solution with induced cost vector $\vec{v}$ which minimizes
$\max_{r\in[N]}\obj(w^{(r)};\vec{v})$. A related problem called {\em multi-budget ordered optimization} has the same setting
as \minmax ordered optimization, but one is also given $N$ budgets $B_1,\ldots, B_N \ge 0$.
The objective is to find a solution inducing cost vector $\vec{v}$ such that $\obj(w;\vec{v})\leq B_r$, for all $r$.
This problem leads to the connections 
with simultaneous optimization~\cite{KumarK06,GoelM06}; 
we discuss these connections more in~\Cref{fairness}. \smallskip

\noindent
{\bf Minimum norm optimization.} A function $f:\R^n \to \R$ is a norm
if 
(i) $f(x)=0$ iff $x=0$; 
(ii) $f(x+y)\leq f(x)+f(y)$ for all $x,y\in\R^n$ (triangle inequality); and 
(iii) $f(\ld x)=|\ld|f(x)$ for all $x\in\R^n,\ld\in\R$ (homogeneity). Properties (ii) and
(iii) imply that $f$ is convex. $f$ is {\em symmetric} if permuting the coordinates of $x$ does not affect its
value, i.e., $f(x)=f(x^{\down})$ for all $x\in\R^n$. 
$f$ is monotone if increasing its coordinate cannot decrease its value
\footnote{Symmetric norms mayn't be monotone. For instance, consider the set
	$C\sse\R^2$, which is the convex hull of the points
	$\{(1,1),(-1,-1),(0,0.5),(0.5,0),(0,-0.5),(-0.5,0)\}$, and define $f(x)$ to be the
	smallest $\ld$ such that $x/\ld\in C$. It is not hard to see that $f$ is a symmetric norm
	over $\R^2$, $f(0,0.5)=1$, but $f(0.5,0.5)\leq 0.5$.}. 
In {\em minimum norm optimization} problem we are given a monotone, symmetric norm $f$, 
and we have to find a solution inducing a cost vector $\vv$ which minimized $f(\vv)$.
Notice that \topl optimization, ordered optimization, and min-max ordered optimization 
are special cases of this problem.  \smallskip

\noindent
{\bf Load balancing and $k$-clustering problems.} 
In the load balancing setting, we have $m$ machines, $n$ jobs, and a processing time $p_{ij} \ge 0$ of job $j$ on machine $i$.
The solution to the problem is an assignment $\sigma$ of jobs to machines. This induces a load $\load_{\sg}(i) := \sum_{j : \sigma(j) = i} p_{ij}$ on each machine.
The vector $\lvec_\sg$ of these loads is the cost-vector associated with the solution $\sigma$. Thus, the min-norm load balancing problem asks to find $\sg$ minimizing
$f(\lvec_\sg)$.

In the $k$-clustering setting, we have a metric space $\bigl(\D,\{c_{ij}\}_{i,j\in\D}\bigr)$, and an integer
$k\geq 0$. The solution to the problem is a set $F\subset \D$, $|F|=k$ of $k$ {\em open} facilities This induces a cost-vector
$\acost$, where $\acost_j := \min_{i\in F} c_{ij}$ is the assignment cost of $j$ to the nearest open facility.
The min-norm $k$-clustering problem asks to find the set $F$ of facilities which minimizes $f(\acost)$.

\section{Sparsifying weights} \label{sparsify}
Let $\dt>0$ be a parameter.
We show how to sparsify $w\in\R^n$ to a weight vector $\tw\in\R^n$ (with non-increasing
coordinates) having $O(\log n/\dt)$ distinct weight values, 
such that for any vector $\vv$, we have $\obj(\tw;\vv) \le \obj(w;\vv) \leq (1+\dt)\obj(\tw;\vv)$.
Moreover, an important property we ensure is that
the breakpoints of $\tw$---i.e., the indices where $\tw_i>\tw_{i+1}$---lie in a set
that depends only by $n$ and $\dt$ and is {\em independent} of $w$. 
As explained in~\Cref{sec:techideas}, sparsification in two distinct places; one, to give a polynomial 
time reduction from min-norm optimization to \minmax ordered optimization (\Cref{minnorm}), and
two, to specify proxy costs which allow us to tackle \minmax ordered optimization.
%

For simplicity, we first describe a sparsification that leads to a factor-$2$ loss (instead of $1+\dt$), and
then refine this.
For every index $i\in[n]$, we set $\tw_i=w_{i}$ if $i=\min\{2^s,n\}$ for some integer
$s\geq 0$; 
otherwise, if $s\geq 1$ is such that $2^{s-1}<i<\min\{2^s,n\}$, set
$\tw_i=w_{\min\{2^s,n\}}=\tw_{\min\{2^s,n\}}$.  
Note that $\tw\leq w$ coordinate wise, and $\tw_1\geq\tw_2\geq\ldots\tw_n$.

Observe that, unlike a different sparsification based on, say,  geometric bucketing of the $w_i$s, the
sparsified vector $\tw$ is {\em not} component-wise close to $w$; in fact $\tw_i$ could be
substantially smaller than $w_i$ for an index $i$. Despite this,~\Cref{simpsparse} shows
that $\obj(\tw;\vv)$ and $\obj(w;\vv)$ are close to each other. 

\begin{claim} \label{simpsparse}
	For any $\vv\in\R_+^n$, we have $\obj(\tw;\vv)\leq\obj(w;\vv)\leq 2\obj(\tw;\vv)$.
\end{claim}

\begin{proof}
	Since $\tw\leq w$, it is immediate that $\obj(\tw;\vv)\leq\obj(w;\vv)$. The other
	inequality follows from a charging argument. Note that for any $s\geq 2$, we have
	$\bigl(\min\{2^s,n\}-2^{s-1}\bigr)\leq 2\bigl(\min\{2^{s-1},n\}-2^{s-2}\bigr)$;  
	hence, the cost contribution 
	$\sum_{i=2^{s-1}+1}^{\min\{2^s,n\}}w_i\vv^{\down}_i$ is at most twice the
	cost contribution in $\obj(\hw;v)$ from the indices 
	$i\in\bigl\{2^{s-2}+1,\ldots,\min\{2^{s-1},n\}\bigr\}$. 
	The remaining cost $w_1\vv^{\down}_1+w_2\vv^{\down}_2$ is at most $2\tw_1\vv^{\down}_1$.
\end{proof}


For the refined sparsification that only loses a $(1+\dt)$-factor, we consider positions
that are powers of $(1+\dt)$.
Let $\pset_{n,\dt}:=\bigl\{\min\{\ceil{(1+\dt)^s},n\}: s\geq 0\bigr\}$. (Note that
$\{1,n\}\sse\pset_{n,\dt}$.) 
Observe that $\pset_{n,\dt}$ depends {\em only} on $n,\dt$ and is oblivious of the weight
vector. We abbreviate $\pset_{n,\dt}$ to $\pset$ in the remainder of this section, and
whenever $n,\dt$ are clear from the context. 
For $\ell\in\pset$, $\ell<n$, define $\next(\ell)$ to be the smallest index in $\pset$ larger
than $\ell$. 
For every index $i\in[n]$, we set $\tw_i=w_{i}$ if $i\in\pset$;
otherwise, if $\ell\in\pset$ is such that $\ell<i<\next(\ell)$ (note that $\ell<n$), set
$\tw_i=w_{\next(\ell)}=\tw_{\next(\ell)}$.  
The following is a generalization of~\Cref{simpsparse}.
\begin{lemma} \label{wsparse}
	For any $\vv\in\R_+^n$, we have $\obj(\tw;\vv)\leq\obj(w;\vv)\leq(1+\dt)\obj(\tw;\vv)$.
\end{lemma}

Not to detract the reader, we defer the proof of~\Cref{wsparse} to~\Cref{append-wsparse}.
We once again stress that the, perhaps more natural, way of 
geometric bucketing (which is indeed used
by~\cite{AouadS18,ByrkaSS18,ChakrabartyS18}) where one ignores small $w_i$s and rounds down each remaining $w_i$ to the
nearest power of $2$ (or $(1+\e)$), doesn't work for our purposes. 
With geometric bucketing, the resulting sparsified
vector $w'$ is component-wise close to $w$ (and so $\obj(w';\vv)$ is close to $\obj(w;\vv)$).
But the breakpoints of $w'$ 
depend heavily on $w$, whereas 
the breakpoints of $\tw$ all lie in $\pset$. 
As noted earlier, this non-dependence on $w$ is extremely crucial for us.
\section{Reducing minimum norm optimization to \minmax ordered optimization} 
\label{minnorm}\label{sec:minnorm}

In this section we show our reduction of the minimum norm optimization problem to \minmax ordered optimization.
We are given a monotone, symmetric norm $f:\R^n\to\R_+$, and we want to find a solution to the underlying optimization problem
which minimizes the $f$ evaluated on the induced cost vector. Let $\vec{o}$ denote the optimal cost vector and let $\iopt = f(\vec{o})$.

We assume the following {\em (approximate) ball-optimization oracle}. Given any cost vector $c\in \R^n$, we can (approximately) optimize $c^\top x$ over
the ball $\mathbb{B}_+(f) := \{x \in \R^n_+ ~:~ f(x)\leq 1\}$.
\begin{equation}
\begin{split}
\textrm{Oracle $\mathcal{A}$ takes input $c\in \R_+^n$ returns a $\kappa$-approximation to }~ 
Bopt(c) := \max \{c^\top x: x\in \mathbb{B}_+(f)\} \\
\textrm{That is, $\mathcal{A}$ returns $\hat{x} \in \mathbb{B}_+(f)$ such that  $c^\top \hat{x} \geq Bopt(c)/\kappa$} ~~~~~~~~~~~~~~~~~~~~~~~~~~~~~~~~~~~~~~~~~~~~~~~~~~~~~~~~~~~~~
\end{split} 
\tag{B-O} \label{balloracle}
\end{equation}
Note that under mild assumptions, the ball-optimization oracle can be obtained, via the ellipsoid method, using a {first-order oracle} for $f$ that returns the subgradient (or even approximate subgradient) of $f$. Recall, $\sgr\in\R^n$ is a 
{\em subgradient} of $f$ at $x\in\R^n$ if we have $f(y)-f(x)\geq\sgr^T(y-x)$ for all
$y\in\R^n$. It is well known that a convex function has a subgradient at every point in its domain.
We begin by stating some preliminary properties of norms,
monotone norms, and symmetric norms. The proof can be found in~\Cref{appen-norms}.

\begin{lemma} \label{normsubprop}
	Let $f:\R^n\to\R_+$ be a norm and $x\in\R_+^n$. 
	\begin{enumerate}[(i), nosep, labelwidth=\widthof{(iii)}, leftmargin=!]
		\item If $\sgr$ is a subgradient of $f$ at $x$, then $f(x)=\sgr^T x$ and 
		$f(y)\geq\sgr^T y$ for all $y\in\R^n$. Also, $\sgr$ is a subgradient of $f$ at any point
		$\ld x$, where $\ld\geq 0$. 
		\item If $f$ is monotone, there exists a subgradient $\sgr$ of $f$ at $x$ such that
		$\sgr\geq 0$. 
		\item Let $f$ be symmetric, and $\sgr$ be a subgradient of $f$ at $x$. Then, $\sgr$ and
		$x$ are similarly ordered, i.e., if $\sgr_i<\sgr_j$ then $x_i\leq x_j$, and
		$f(x)=\obj(\sgr^{\down};x)$. Moreover, for any permutation $\pi:[n]\to[n]$, 
		the vector $\sgr^{(\pi)}:=\bigl\{\sgr_{\pi(i)}\bigr\}_{i\in[n]}$is a subgradient of $f$ at
		$x^{(\pi)}$. 
	\end{enumerate}
\end{lemma}
%
Motivated by the above lemma, we define the following set of non-increasing subgradients over points on the unit norm-ball.
This set is possibly infinite.
$$
\Scol=\Bigl\{\sgr\in\R_+^n: \quad \sgr_1\geq\sgr_2\geq\ldots\geq\sgr_n, \quad 
\sgr\text{ is a subgradient of $f$ at some $x\in\ball_+(f)$}\Bigr\}.
$$
As a warm up,~\Cref{normeqv} shows that min-norm optimization is {\em equivalent} to 
\minmax ordered optimization with an {\em infinite} collection of weight vectors. This establishes the reduction, however it is inefficient.

\begin{lemma} \label{normeqv}
	Let $x\in\R_+^n$. We have $f(x)=\max_{w\in\Scol}\obj(w;x)$.
\end{lemma}

\begin{proof}
	We first argue that $f(x)\leq\max_{w\in\Scol}\obj(w;x)$. 
	By part (ii) (of~\Cref{normsubprop}), there is a subgradient $\sgr\geq 0$ of $f$ at 
	$x$. By part (iii), there is 
	a common permutation $\pi$ that defines $\sgr^{\down}$ and $x^{\down}$, and 
	$\hsgr=\sgr^{\down}$ is
	a subgradient of $f$ at $x^{\down}$. By part (i), $\hsgr$ is also a subgradient of $f$
	at $x^{\down}/f(x^{\down})\in\ball_+(f)$. So $\hsgr\in\Scol$.
	Also, $f(x)=\obj(\hsgr;x)$ (by part (iii)), and so $f(x)\leq\max_{w\in\Scol}\obj(w;x)$.
	
	Conversely, consider any $w\in\Scol$, and let it be a subgradient of $f$ at
	$z\in\ball_+(f)$. We have $f(x)=f(x^{\down})\geq w^Tx^{\down}$ (by part (i) of~\Cref{normsubprop}), and so $f(x)\geq\obj(w;x)$. Therefore,
	$f(x)\geq\max_{w\in\Scol}\obj(w;x)$. 
\end{proof}

To reduce to min-max ordered optimization, we need to find a polynomial-sized collection of weight vectors.
Next, we show how to leverage the 
{\em weight sparsification} idea in~\Cref{sparsify} and achieve this taking a slight hit in the approximation factor.
Let $0<\e\leq 0.5$ be a parameter. 
The sparsification procedure (\Cref{wsparse}) shows that, with an $(1+\e)$-loss, we can
focus on a set of $O(\log n/\e)$ coordinates and describe the weight vectors by
their values at these coordinates. For the ordered-optimization objective
$\obj(w;x)$, moving to the sparsified weight incurs only a $(1+\e)$-loss. Furthermore, again taking a loss of $(1+\e)$, we can assume these
coordinates are set to powers of $(1+\e)$.
Our goal (roughly speaking) is then only to consider the collection
consisting of the sparsified, rounded versions of vectors in $\Scol$. 
\Cref{nondec} implies that
we can enumerate all sparsified, rounded weight vectors in
polynomial time. 

But we also need to be able to determine if such a vector $\tw$ 
is ``close'' to a subgradient in $\Scol$, and this is where \eqref{balloracle} is used.
First note that $\sgr \in \Scol$ iff\footnote{If $\sgr\in\Scol$ is the subgradient of $f$ at $y\in\ball_+(f)$,
	$\sgr^Tx\leq f(x)\leq 1\ \ \forall x\in\ball_+(f)$, and $\sgr^Ty/f(y)=1$, so
	$\max_{x\in\ball_+(f)}\sgr^Tx=1$. Alternately, if $Bopt(\sgr)=1$, then we have $d^\top z = 1$ for 
	some $f(z)\leq 1$ implying $f(z) + d^\top(y-z) \leq d^\top y$ for any $y$. If the LHS is $> f(y)$, then we would get
	$d^\top (y/f(y)) > 1$ contradicting $Bopt(\sgr)=1$. } 
$Bopt(\sgr)=1$. Thus to check if $\tw$ is ``close'' to a subgradient in $\Scol$, it suffices to (approximately) solve for $Bopt(\tw)$
and check if the answer is within $(1\pm \e)$ (or scaled by $\kappa$ if we only have an approximate oracle). We give the details next.

%

To make the enumeration go through we need to make the following mild assumptions. These assumptions need to be 
checked for the problems at hand, and are often easy to establish.
%
\begin{enumerate}[label=(A\arabic*), topsep=0.5ex, itemsep=0ex,
	labelwidth=\widthof{(A3)}, leftmargin=!] 
	\item \label{assum1}
	We can determine in polytime if $\vo_1=0$. If $\vo_1>0$ (so $\iopt>0$), then $\vo_1\geq 1$ (assuming integer data),
	and we can compute an estimate $\high$ such that $\vo_1\leq\high$. 
	In the sequel, assume that $\vo_1\geq 1$.
	
	\item \label{assum2}
	We have bounds $\lb,\ub>0$ such that $\lb\leq\iopt\leq\ub$. 
	Then~\ref{assum1} and~\Cref{normsubprop} (i) imply that $\sgr_1\leq\ub$ for
	all $\sgr\in\Scol$.
	
	%
\end{enumerate}

We take $\dt=\e$ in the sparsification procedure in~\Cref{sparsify}. 
Let $\pset=\pset_{n,\e}:=\{\min\{\ceil{(1+\e)^s},n\}: s\geq 0\}$. Recall that $\next(\ell)$ is the smallest index in $\pset$ larger
than $\ell$. 
The sparsified version of $w\in\R^n$ is the vector $\tw\in\R^n$ 
given by $\tw_i=w_{i}$ if $i\in\pset$; and $\tw_i=w_{\next(\ell)}$ otherwise, 
where $\ell\in\pset$ is such that $\ell<i<\next(\ell)$.    
Since $\tw$ is completely specified by specifying the positions in $\pset$, we define
the $|\pset|$-dimensional vector
$u :=\{\tw_\ell\}_{\ell\in\pset}$. We identify $\tw$ with $u \in\R_+^\pset$
and say that $\tw$ is the {\em expansion of $u$}.   
%
\begin{equation*}
\begin{split}
\text{Define} \quad \W'\sse\R_+^n\ :=\ \Bigl\{
&\text{expansion of }u \in\R_+^\pset: \quad \exists\ell^*\in\pset\ \text{s.t.}\ 
u_\ell=0\ \forall\ell\in\pset\text{ with }\ell>\ell^*, \\
& u_1,u_2,\ldots,u_{\ell^*}\text{ are powers of $(1+\e)$ (possibly smaller than $1$)} \\
& u_1\in\bigl[\tfrac{\lb}{n\cdot\high},\,\,\ub(1+\e)\bigr), 
\quad u_1\geq u_2\geq\ldots\geq u_{\ell^*}\geq\tfrac{\e u_1}{n(1+\e)}\Bigr\}.
\end{split}
\end{equation*}
Let $\bon^n$ denote the all $1$s vector in $\R^n$. 
Now define  
$$
\W\ :=\ 
\Bigl\{w\in\W': \ \ \text{oracle $\A$ run on $w$ returns $\hx\in\ball_+(f)$ s.t. }%
w^T\hx\in\bigl[(1-\e)/\kp,1+\e\bigr]\Bigr\}\,\cup\,
\Bigl\{\tfrac{\lb}{n\cdot\high}\cdot\bon^n\Bigr\}.
$$ 
The extra scaled all ones vector is added for a technical reason. We use the following enumeration claim.

\begin{claim} \label{nondec}
	There are at most   $(2e)^{\max\{N,k\}}$ non-increasing sequences of $k$ integers chosen from $\{0,\ldots,N\}$.
\end{claim}

The following theorem establishes the reduction from the minimum norm problem to \minmax ordered optimization.
The proof idea is as sketched above; we defer the details of the proof to~\Cref{appen-norms}.

\begin{theorem} \label{normredn}\label{thm:minnorm}
	For any $\vv\in\R_+^n$, the following hold.
	
	\vspace*{0.2ex}
	\noindent
	(i) $\max_{w\in\W}\obj(w;\vv)\leq
	\max\bigl\{\kp(1+\e)f(\vv),\frac{\lb}{n\cdot\high}\sum_{i\in[n]}\vv_i\bigr\}$, \quad 
	%
	(ii) $f(\vv)\leq(1-\e)^{-1}\max_{w\in\W}\obj(w;\vv)$.
	
	\vspace*{0.2ex}
	\noindent 
	Hence, a $\gm$-approximate solution $\vv$ for the \minmax ordered-optimization
	problem with objective \linebreak $\max_{w\in\W}\obj(w;\vv)$ (where $\gm\geq 1$) satisfies
	$f(\vv)\leq\gm\kp(1+3\e)\iopt$. 
	
	Constructing $\W$ requires 
	$O\bigl(\frac{\log n}{\e^2}\log(\frac{n\cdot\ub\cdot\high}{\lb})(\frac{n}{\e})^{O(1/\e)}\bigr)$ 
	calls to $\A$, which is also a bound on $|\W|$. 
\end{theorem}

\section{Proxy costs}\label{ordproxy}\label{sec:proxy}
As mentioned in~\Cref{sec:techideas}, the key to tackling ordered optimization is to view
the problem of minimizing the sum of a suitably devised proxy-cost function over all coordinates.
We describe this proxy in this section.
We first so so for \topl optimization. This
will serve to motivate and illuminate the proxy-cost function that we use for
(general) ordered optimization. As usual, we use $\vec{o}$ to denote the cost vector
corresponding to an optimal solution, and $\iopt$ to denote the optimal cost. 
Recall, $\cost(\ell;\vv)$ is the cost of the \topl optimization.

Define $z^+:=\max\{0,z\}$ for $z\in\R$. 
For any scalar $\rho > 0$, define $h_\rho(z) := (z - \rho)^+$.
The main insight is that for any $\vv\in\R^n$,  we have
$\cost(\ell;\vv)=\min_{\thr\in\R}\ell\cdot\thr+\sum_{i=1}^nh_\thr(\vv_i)$. 

\begin{claim} \label{topllb}
	For any $\ell\in[n]$, any $\vv\in\R^n$, and any $\thr\in\R$, we have
	$\obj(\ell;\vv)\leq\ell\cdot\thr+\sum_{i=1}^nh_\thr(\vv_i)$.
\end{claim}

\begin{proof}
	We have 
	$\obj(\ell;\vv)=\sum_{i=1}^\ell \vv^{\down}_i
	\leq\ell\cdot\thr+\sum_{i=1}^\ell(\vv^{\down}_i-\thr)^+
	\leq\ell\cdot\thr+\sum_{i=1}^n(\vv^{\down}_i-\thr)^+$. 
\end{proof}

\begin{claim} \label{toplopt}
	Let $\ell\in[n]$, and $\thr$ be such that 
	$\vo_\ell\leq\thr\leq(1+\e)\vo_\ell$. Then 
	$\ell\cdot\thr+\sum_{i=1}^n h_\thr(\vec{o}_i)\leq(1+\e)\obj(\ell;\vec{o})$. 
\end{claim}

\begin{proof}
	We have $\sum_{i=1}^n(\vec{o}_i-\thr)^+\leq\sum_{i=1}^\ell(\vo_i-\vo_\ell)$.
	Since $\thr\leq(1+\e)\vo_\ell$, we have 
	$\ell\cdot\thr+\sum_{i=1}^n(\vec{o}_i-\thr)^+
	\leq(1+\e)\sum_{i=1}^\ell\bigl(\vo_\ell+(\vo_i-\vo_\ell)\bigr)=(1+\e)\obj(\ell;\vec{o})$. 
\end{proof}

The above claims indicate that if we obtain a good estimate $\thr$ of $\vo_\ell$, then 
$\ell\cdot\thr+\sum_{i=1}^nh_\thr(\vv_i)$ can serve as a good proxy for $\cost(\ell;\vv)$, 
and we can focus on 
the problem of finding $v$ minimizing $\sum_{i=1}^nh_\thr(\vv_i)$. 
The following properties will be used many times.

\begin{claim} \label{triangle}
	We have: (i) $h_\thr(x)\leq h_\thr(y)$ for any $\thr$, $x\leq y$;
	(ii) $h_{\thr_1}(x)\leq h_{\thr_2}(x)$ for any $\thr_1\geq\thr_2$, and any $x$;
	(iii) $h_{\thr_1+\thr_2}(x+y)\leq h_{\thr_1}(x)+h_{\thr_2}(y)$ for any $\thr_1,\thr_2,x,y$.
\end{claim}

\begin{proof}
	Part (iii) is the only part that is not obvious. If $h_{\thr_1+\thr_2}(x+y)=0$, then the
	inequality clearly holds; otherwise, 
	$h_{\thr_1+\thr_2}(x+y)=x-\thr_1+y-\thr_2\leq (x-\thr_1)^++(y-\thr_2)^+$.
\end{proof}

We remark that 
our proxy function for \topl optimization is similar to, but subtly
stronger than, the proxy function utilized in recent prior works on the $\ell$-centrum and
ordered $k$-median clustering problems~\cite{ByrkaSS18,ChakrabartyS18}. 
This strengthening (and its extension to ordered optimization) forms the basis of our
significantly improved approximation guarantees of $(5+\e)$ for 
ordered $k$-median (\Cref{improv}), which improves upon the prior-best guarantees  
for {\em both} $\ell$-centrum and ordered $k$-median~\cite{ChakrabartyS18}. 
Furthermore, this proxy function also leads to (essentially) a 
$2$-approximation for \topl load balancing 
and ordered load balancing (\Cref{sec:loadbal:splcase}). 

\vspace*{-1ex}
\paragraph{Ordered optimization.}
We now build upon our insights for \topl optimization.
Let $w\in\R^n$ be the weight vector
(with non-increasing coordinates) underlying the ordered-optimization problem.
So, $\iopt=\cost(w;\vec{o})$ is the optimal cost. 
The intuition underlying our proxy function comes from the observation 
that we can write 
$\obj(w;\vv)=\sum_{i=1}^n(w_i-w_{i+1})\obj(i;\vv)$, where we define $w_{n+1}:=0$.
Plugging in the proxy functions for $\obj(i;\vv)$ in this expansion 
immediately leads to a
proxy function for $\obj(w;\vv)$. The $\obj(i;\vv)$ terms that appear with positive
coefficients in the above linear combination 
are those where $w_i>w_{i+1}$, i.e., corresponding to the breakpoints
of $w$. Thus, 
the proxy function that we obtain for ordered optimization 
will involve multiple $\thr$-thresholds, which are intended to be the estimates of the 
$\vo_i$ values corresponding to breakpoints.
However, we cannot afford to ``guess'' so many of these thresholds. 
An important step to make this work is to 
first {\em sparsify} the weight vector $w$
to control the number of breakpoints, 
and then 
utilize the above expansion. 
As mentioned in~\Cref{sparsify}, while geometric
bucketing of weights would reduce the number of breakpoints for a single weight function, for our applications to
\minmax ordered optimization, we need the uniform way of sparsifying multiple weight
vectors, and we therefore use the sparsification procedure in~\Cref{sparsify}.


Let $\dt,\e>0$ be parameters.
Let $\pset=\pset_{n,\dt}:=\{\min\{\ceil{(1+\dt)^s},n\}: s\geq 0\}$. 
Recall that $\next(\ell)$ is the smallest index in $\pset$
larger than $\ell$. For notational convenience, we define $\next(n):=n+1$, and for $\vv\in\R^n$,
define $\vv_{n+1}:=0$.
We sparsify $w$ to $\tw\in\R^n$ by setting
$\tw_i=w_{i}$ if $i\in\pset$, and $\tw_i=w_{\next(\ell)}$ otherwise, 
where $\ell\in\pset$ is such that $\ell<i<\next(\ell)$.    
%
%

Our proxy function is obtained by guessing (roughly speaking) the thresholds
$\vo_\ell$ for all $\ell\in\pset$ within a multiplicative $(1+\e)$ factor, and rewriting 
$\obj(\tw;\vv)$ in terms of these thresholds. 
Let $\vec{t}:=\{t_\ell\}_{\ell\in\pset}$ be a {\em threshold vector}. 
Define $\vt_{n+1}:=0$. 
We say that $\vec{t}$ is {\em valid} if $t_\ell\geq t_{\next(\ell)}$ for all
$\ell\in\pset$. 
(So this implies that $\vec{t}\geq 0$.)
A valid threshold vector $\vec{t}$, defines the proxy function.
\begin{alignat}{1}
\hspace*{-2.5ex}
\prox_{\vec{t}}(\tw;\vv) := 
\sum_{\ell\in\pset}\bigl(&\tw_\ell-\tw_{\next(\ell)}\bigr)
\Bigl[\ell\cdot t_\ell+\sum_{i=1}^n h_{t_\ell}(\vv_i)\Bigr] = 
\sum_{\ell\in\pset}\bigl(\tw_\ell-\tw_{\next(\ell)}\bigr)\ell\cdot t_\ell
+\sum_{i=1}^n h_{\vec{t}}~(\tw;\vv_i), \label{proxexpr} \\
\text{where,} \quad 
h_{\vec{t}}~(\tw;a) &:= \sum_{\ell\in\pset}\bigl(\tw_\ell-\tw_{\next(\ell)}\bigr)h_{t_\ell}(a) 
\label{hfun}
\end{alignat}
Note that the above proxy functions are strict generalizations of the case of the Top-$\ell$ optimization 
in which case $\pset = \{\ell\}$, and the weights are $1$ till $\ell$ and $0$ afterwards.

Throughout the rest of this section, we work with the sparsified weight vector $\tw$. 
Observe that $h_{\vec{t}}~(\tw;x)$ is a continuous, piecewise-linear, non-decreasing
function of $x$. 
Our proxy for $\obj(\tw;\vv)$ will be the function $\prox_{\vec{t}}(\tw;\vv)$
for a suitably chosen threshold vector $\vec{t}$.
%
To explain the above definition, notice that \eqref{proxexpr} is 
the expression obtained by plugging in the proxy functions
($\ell\cdot\thr+\sum_{i=1}^n(v_i-\thr)^+$) defined for the
$\cost(\ell;\bullet)$-objectives in the expansion of $\obj(\tw;v)$ as a linear 
combination of $\obj(\ell;v)$ terms. 

\begin{claim} \label{topltoord}
	For any $\vv\in\R^n$, we have 
	$\obj(\tw;\vv)=\sum_{\ell\in\pset}\bigl(\tw_\ell-\tw_{\next(\ell)}\bigr)\obj(\ell;\vv)$.
\end{claim}

\begin{proof}
	We have
	\begin{equation*}
	\obj(\tw;\vv)=\sum_{i=1}^n\tw_i\vv^{\down}_i
	=\sum_{i=1}^n\sum_{\ell=i}^n(\tw_\ell-\tw_{\ell+1})\vv^{\down}_i
	=\sum_{\ell=1}^n(\tw_\ell-\tw_{\ell+1})\sum_{i=1}^\ell \vv^{\down}_i \\
	= \sum_{\ell\in\pset}\bigl(\tw_\ell-\tw_{\next(\ell)}\bigr)\obj(\ell;\vv).
	\end{equation*}
	The last equality follows since $\tw_\ell=\tw_{\ell+1}$ for all $\ell\in[n]\sm\pset$, and 
	$\tw_{\ell+1}=\tw_{\next(\ell)}$ for $\ell\in\pset$.
\end{proof}

\begin{claim} \label{proxylb}
	For any valid threshold vector $\vec{t}\in\R^{\pset}$, 
	and any $\vv\in\R^n$, we have 
	$\obj(\tw;\vv)\leq\prox_{\vec{t}}~(\tw;\vv)$.
\end{claim}

\begin{proof}
	We have $\prox_{\vec{t}}~(\tw;\vv)= 
	\sum_{\ell\in\pset}\bigl(\tw_\ell-\tw_{\next(\ell)}\bigr)\bigl(\ell\cdot t_\ell+\sum_{i=1}^n h_{t_\ell}(\vv_i)\bigr)$.
	The statement now follows by combining~\Cref{topltoord} and~\Cref{topllb}, taking $t=t_\ell$ for each $\ell\in\pset$.
\end{proof}

\begin{claim} \label{proxysimp}
	Let $\vec{t}\in\R^{\pset}$ be a valid threshold vector such that 
	$\vo_\ell\leq t_\ell\leq(1+\e)\vo_\ell$ for all $\ell\in\pset$.
	Then, $\prox_{\vec{t}}~(\tw;\vec{o})\leq(1+\e)\obj(\tw;\vec{o})$.
\end{claim}

\begin{proof}
	We have $\prox_{\vec{t}}~(\tw;\vec{o})=\sum_{\ell\in\pset}\bigl(\tw_\ell-\tw_{\next(\ell)}\bigr)
	\bigl(\ell\cdot t_\ell+\sum_{i=1}^n h_{t_\ell}(\vo_i)\bigr)$.
	The statement now follows by combining~\Cref{toplopt}, where we take $t=t_\ell$ for
	each $\ell\in\pset$, and~\Cref{topltoord}.
\end{proof}

\Cref{proxylb} and~\Cref{proxysimp} imply that: (1) if we can obtain in polytime a
valid threshold vector $\vec{t}\in\R^\pset$ satisfying the conditions of~\Cref{proxysimp}, and (2) obtain a cost vector $v$ that approximately minimizes
$\sum_{i=1}^nh_{\vec{t}}~(v_i)$, then we would obtain an approximation guarantee for the
ordered-optimization problem. We will not quite be able to satisfy (1). Instead, we
will obtain thresholds that will satisfy a somewhat weaker condition (see~\Cref{proxyopt}), which we show is still sufficient.
The following claim, whose proof is in~\Cref{append-ordproxy}, will be useful.

\begin{claim} \label{proxynear}
	Let $\vec{t},\vec{\tdt}\in\R^{\pset}$ be two valid threshold vectors with
	$\vec{t}\leq\vec{\tdt}$ and
	$\|\vec{t}-\vec{\tdt}\|_\infty\leq\Dt$. Then, for any $\vv\in\R^n$, we have 
	$\bigl|\prox_{\vec{t}}~(\tw;\vv)-\prox_{\vec{\tdt}}~(\tw;\vv)\bigr|\leq n\tw_1\Dt$.
\end{claim}

\begin{lemma} \label{proxyopt}
	Let $\vec{t}\in\R^{\pset}$ be a valid threshold vector satisfying the following for all
	$\ell\in\pset$: $\vo_\ell\leq t_\ell\leq(1+\e)\vo_\ell$ if
	$\vo_\ell\geq\frac{\e\vo_1}{n}$, and $t_\ell=0$ otherwise.
	Then, 
	$$
	\prox_{\vec{t}}~(\tw;\vec{o}) = 
	\sum_{\ell\in\pset}\bigl(\tw_\ell-\tw_{\next(\ell)}\bigr)\ell\cdot t_\ell
	+\sum_{i=1}^nh_{\vec{t}}~(\tw;\vec{o}_i)\leq (1+2\e)\obj(\tw;\vec{o}).
	$$
\end{lemma}

\begin{proof}
	For $\ell\in\pset$, define $\tdt_\ell=t_\ell$ if $\vo_\ell\geq\frac{\e\vo_1}{n}$, and
	$\tdt_\ell=\vo_\ell$ otherwise. 
	Clearly, $\vec{t}\leq\vec{\tdt}$ and
	$\|\vec{t}-\vec{\tdt}\|_\infty\leq\frac{\e\vo_1}{n}$, so by~\Cref{proxynear},
	we have $\prox_{\vec{t}}~(\tw;\vo)\leq\prox_{\vec{\tdt}}~(\tw;\vo)+\e\tw_1\vo_1$.
	The threshold vector $\vec{\tdt}$ satisfies the conditions of~\Cref{proxysimp}, so
	$\prox_{\vec{\tdt}}~(\tw;\vec{o})\leq(1+\e)\obj(\tw;\vec{o})$. 
	So $\prox_{\vec{t}}~(\tw;\vec{o})\leq(1+2\e)\obj(\tw;\vec{o})$. 
\end{proof}

\begin{lemma}[Polytime enumeration of threshold vectors] \label{polyguess}
	Suppose that we can obtain in polynomial time a (polynomial-size) set $S\sse\R$ containing
	a value $\rho$ satisfying $\vo_1\leq \rho\leq(1+\e)\vo_1$. 
	Then, in time $O\bigl(|S|\cdot|\pset|\cdot\max\{(\frac{n}{\e})^{O(1/\e)},n^{1/\dt}\}\bigr)
	=O\bigl(|S|\max\{(\frac{n}{\e})^{O(1/\e)},n^{O(1/\dt)}\}\bigr)$, we can obtain a set 
	$A\sse\R_+^\pset$ that contains a valid threshold vector $\vec{t}$ satisfying the conditions 
	of~\Cref{proxyopt}. 
	
	If $\vec{o}$ is integral, $\vo_1>0$, and $\rho$ is a power of $(1+\e)$, then this $\vec{t}$
	satisfies: for every $\ell\in\pset$, either $t_\ell=0$ or $t_\ell\geq 1$ and is a power of
	$(1+\e)$. 
\end{lemma}

\begin{proof}
	We first guess the largest index $\ell^*\in\pset$ such that
	$\vo_\ell\geq\frac{\e\vo_1}{n}$. For each such $\ell^*$, and each $t_1\in S$, we do the
	following. We guess $t_\ell$ for $\ell\in\pset, 2\leq\ell\leq\ell^*$, where all the
	$t_\ell$s are of the form $t_1/(1+\e)^j$ for some integer $j\geq 0$ and are at least 
	$\frac{\e t_1}{n(1+\e)}$, and the $j$-exponents are non-decreasing with $\ell$. 
	For $\ell\in\pset$ with $\ell>\ell^*$, we set 
	$t_\ell=0$, 
	and add the resulting threshold vector $\vec{t}$ to $A$. Note that there are at most 
	$1+\log_{1+\e}\bigl(\frac{n}{\e}\bigr)=O\bigl(\frac{1}{\e}\log\frac{n}{\e}\bigr)$ choices
	for the exponent $j$. So since we need to guess a non-decreasing sequence of at most 
	$|\pset|=O(\log n/\dt)$ exponents from a range of size 
	$O\bigl(\frac{1}{\e}\log\frac{n}{\e}\bigr)$, there are only
	$\exp\bigl(\max\{O(\frac{1}{\e}\log(\frac{n}{\e})),|\pset|\}\bigr)=
	O\bigl(\max\{(\frac{n}{\e})^{O(1/\e)},n^{1/\dt}\}\bigr)$ choices (by~\Cref{nondec}). 
	So the enumeration takes time 
	$O\bigl(|S|\cdot|\pset|\max\{(\frac{n}{\e})^{O(1/\e)},n^{1/\dt}\}\bigr)$, which is
	also an upper bound on $|A|$. 
	
	We now argue that $A$ contains a desired valid threshold vector. First, note that by
	construction $A$ only contains valid threshold vectors. Consider the iteration when we
	consider $t_1=\rho$, and have guessed $\ell^*$ correctly. 
	For $\ell\in\pset$ with $2\leq\ell\leq\ell^*$, we know that 
	$\vo_\ell\geq\frac{\e\vo_1}{n}\geq\frac{\e t_1}{n(1+\e)}$ and 
	$\vo_\ell\leq\vo_1\leq t_1$. So we will enumerate non-increasing values
	$t_2,\ldots,t_{\ell^*}$ such that $\vo_\ell\leq t_\ell\leq(1+\e)\vo_\ell$ for each such
	$\ell$. 
	The remaining $t_\ell$s are set to $0$, so $\vec{t}$ satisfies the conditions of~\Cref{proxyopt}. 
	
	Finally, suppose $\vec{o}\in\Z_+^n$ and $\rho$ is a power of $(1+\e)$. If $t_\ell<1$, then
	$\ell\geq\ell^*$, but $\vo_\ell\leq t_\ell<1$, which means that $\vo_\ell=0$ contradicting
	that $\vo_\ell\geq\frac{\e\vo_1}{n}$. Also, $t_\ell=\rho/(1+\e)^j$, so it is a power of
	$(1+\e)$.  
\end{proof}

The upshot of the above discussion is that it suffices to focus on the algorithmic problem
of minimizing $\sum_{i=1}^nh_{\vec{t}}~(v_i)$ for a given valid threshold vector. This is
formalized by the following lemma whose proof is in~\Cref{append-ordproxy}.

\begin{lemma} \label{proxyub}
	Let $\vec{t}\in\R^{\pset}$ be a valid threshold vector satisfying the conditions of~\Cref{proxyopt}. 
	Let $\vv\in\R_+^n$ be such that 
	$\sum_{i=1}^n h_{\tht\vec{t}}~(\tw;\vv_i)\leq\gm\cdot\sum_{i=1}^n h_{\vec{t}}~(\tw;\vec{o}_i)+M$, where 
	$\gm,\tht\geq 1$, $M\geq 0$. 
	Then, $\obj(\tw;\vv)\leq\max\{\tht,\gm\}(1+2\e)\obj(\tw;\vec{o})+M$, and hence 
	$\obj(w;\vv)\leq (1+\dt)\max\{\tht,\gm\}(1+2\e)\iopt+(1+\dt)M$.
\end{lemma}

\section{Approach towards \minmax ordered optimization} \label{simord}\label{sec:minmaxapproach}
Given the reduction~\Cref{normredn} in~\Cref{minnorm}, we now discuss our approach for solving \minmax-ordered load balancing and clustering.
Eventually, we will need to take a problem-dependent approach, but at a high
level, there are some common elements to our approaches for the two problems as we now  
elucidate.

As a stepping stone, we first consider 
ordered optimization (i.e., where we have one weight vector $w$), and formulate a suitable 
LP-relaxation (see~\Cref{sec:loadbal:lp} and~\Cref{sec:clus:lp}) for the problem of minimizing $\sum_{i=1}^n h_{\vec{t}}(\tw;\vv_i)$, i.e.,
the $\vv$-dependent part of our proxy function for $\obj(\tw;\vv)$ (see~\eqref{proxexpr} and~\eqref{hfun}), where $\tw$ is the sparsified version of $w$.
Our LP-relaxation 
will have the property that 
only its objective depends on $\tw$ and not its constraints.
The LP for \minmax ordered optimization is obtained by
modifying the objective in the natural way.  

{\em The technical core of our approach involves devising a 
	deterministic, weight-oblivious rounding procedure} for this LP (see~\Cref{sec:loadbal:deto} and~\Cref{sec:clus:deto}).
To elaborate, we design a 
procedure that given an arbitrary feasible solution, say $\bx$, to this LP, rounds it 
\emph{deterministically}, without any knowledge of $w$, 
to produce a solution to the
underlying optimization problem 
whose induced cost vector $\vv$ satisfies the following: 
for {\em every} sparsified weight vector $\tw$, we have (loosely speaking)
$\obj(\tw,\vv)=O(1)\cdot\text{(LP-objective-value of $\bx$ under $\nw$)}$.
We call this a 
{\em deterministic, weight-oblivious} rounding procedure.
To achieve this, we need to introduce some novel constraints in our LP, beyond the standard
ones for load balancing and $k$-clustering.
The benefit of such an oblivious guarantee is clear: if $\bx$ is an optimal
solution to the LP-relaxation for \minmax ordered optimization, then the above guarantee yields 
$O(1)$-approximation for the \minmax ordered-optimization problem. Indeed, this also will solve the multi-budgeted ordered optimization problem.

We point out that it is important that the oblivious rounding procedures we design are
{\em deterministic}, which is also what makes them noteworthy, and we need to develop
various new ideas to obtain such guarantees. Using a randomized
$O(1)$-approximation oblivious rounding procedure in \minmax ordered optimization would
yield that the {\em maximum expected cost} $\obj(w^{(i)};\tv)$ under weight vectors $w^{(i)}$ in 
our collection is $O(\iopt)$; but what we need is a bound on the 
{\em expected maximum cost}. 
Therefore, without a sharp concentration result, a
randomized oblivious guarantee is insufficient for the purposes of utilizing it for
\minmax ordered optimization. Also, note that derandomizing an
oblivious randomized-rounding procedure would typically cause it to lose its obliviousness 
guarantee. (We also remark that if we allow randomization, then it is well-known
that {\em any} LP-relative approximation algorithm can be used to obtain a randomized
oblivious rounding procedure (see~\cite{CarrV02}.) 

To obtain our deterministic oblivious rounding procedure, we first 
observe that
$\sum_{i=1}^n h_{\vec{t}}~(\tw;\vv_i)$ can be equivalently written as
$\sum_{\ell\in\pset}\tw_{\next(\ell)}\sum_{i=1}^n\bigl(\min\{\vv_i,t_\ell\}-t_{\next(\ell)}\bigr)^+$.
In our LP-relaxation, we introduce fractional variables to specify the quantities
$\sum_{i=1}^n\bigl(\min\{\vv_i,t_\ell\}-t_{\next(\ell)}\bigr)^+$.  
If we can round the fractional solution while roughly preserving these
quantities (up to constant factors), then we can get the desired oblivious
guarantee. This is what we achieve (allowing for an $O(1)$ violation of the
thresholds) by,
among other things, leveraging our new valid constraints that we add to the LP. 
For instance, in load balancing, $\vv_i$ denotes the load on machine $i$ and the above
quantity represents the portion of the total load on a machine between thresholds $t_{\next(\ell)}$ and  
$t_\ell$, and we seek to be preserve this in the rounding. 

Preserving the aforementioned quantities amounts to having multiple knapsack constraints, and rounding them so as to satisfy them with as little violation as possible.
We utilize the following technical tool to achieve this.
We emphasize that the objective $c^T q$ below is {\em not} related to $\tw$, but
encodes quantities that arise in our rounding procedure. 
\Cref{iterrnd} is proved using {\em iterative rounding}, 
by combining ideas from~\cite{BansalKN09}, which considered directed network design,
and the ideas involved in 
an iterative-rounding based $2$-approximation algorithm for the generalized assignment
problem (see Section 3.2 of~\cite{LauRS11}). Similar results are known in the literature,
but we could not quite find a result that exactly fits our needs; 
we include a proof in~\Cref{append-iterrnd} for completeness.

\begin{theorem} \label{iterrnd}
	Let $\hq$ be a feasible solution to the following LP:
	\begin{equation}
	\min \quad c^T q \quad A_1q\leq b_1, \ \ A_2q\geq b_2, \quad Bq\leq d, \quad
	q\in\R_+^M. \tag{Q} \label{iterlp}
	\end{equation}
	Suppose that:
	(i) $A_1,A_2,B,b_1,b_2,d\geq 0$;
	(ii) $A_1, A_2$ are $\{0,1\}$-matrices, and the supports of the rows of 
	$\bigl(\begin{smallmatrix} A_1 \\ A_2 \end{smallmatrix}\bigr)$ form a laminar family;
	(iii) $b_1,b_2$ are integral; and
	(iv) $q_j\leq 1$ is an implicit constraint implied by $A_1q\leq b_1,\ A_2q\geq b_2$.
	Let $k$ be the maximum number of constraints of $Bq\leq d$ that a variable appears in.
	
	We can round $\hq$ to an integral (hence $\{0,1\}$)
	solution $\tq$ satisfying: (a) $c^T\tq\leq c^T\hq$; 
	(b) the support of $\tq$ is contained in the support of $\hq$; 
	(c) $A_1\tq\leq b_1$, $A_2\tq\geq b_2$; and 
	(d) $(B\tq)_i\leq d_i+k(\max_{j:\hq_{j}>0}B_{ij})$ for all $i$ ranging over the rows of
	$B$. 
\end{theorem}  

\section{Load balancing}\label{loadbal}\label{sec:loadbal}

In this section, we use our framework to design constant factor approximation algorithms for the minimum-norm load balancing problem.
Let us recall the problem.
We are given a set $J$ of $n$ jobs, a set of $m$ machines, and for each job $j$ and machine
$i$, the processing times $p_{ij}\geq 0$ required to process $j$ on machine $i$.
We have to output an assignment $\sigma: J\to[m]$ of jobs to machines.
The load on machine $i$ due to  $\sigma$ is $\load_\sigma(i):=\sum_{j : \sigma(j) = i} p_{ij}$. 
Let $\lvec_\sg:=\{\load_\sg(i)\}_{i\in[m]}$ denote the load-vector induced by $\sg$.

In the minimum-norm load-balancing problem, one seeks to minimize the norm of the
load vector $\lvec_\sg$ for a given monotone, symmetric norm.
In the special case of {\em ordered load-balancing} problem, given a non-negative, non-increasing vector $w \in \R^m_+$ (that is, $w_1 \geq w_2 \geq \cdots \geq w_m \geq 0$), one seeks to minimize
$\obj\bigl(w;\lvec_\sg\bigr):= w^T\lvec_\sg^{\down} =\sum_{i=1}^m w_i\lvec_\sg^{\down}(i)$.
In the {\em Top-$\ell$ load balancing problem}, one seeks to minimize the sum of the $\ell$ largest loads in $\lvec_\sg$.

\begin{theorem}
	\label{thm:normlb}
	Given any monotone, symmetric norm $f$ on $\R^m$ with a $\kp$-approximate ball-optimization oracle for $f$ (see~\eqref{balloracle}), 
	and for any $\e>0$, there is a 
	$38\kp(1+5\e)$-approximation algorithm for the problem of finding an assignment $\sg:J\to [m]$ which minimizes
	$f\bigl(\lvec_\sg\bigr)$. The running time of the algorithm is 
	$\poly\bigl(\text{input size},(\frac{m}{\e})^{O(1/\e)}\bigr)$.
\end{theorem}

We have not optimized the constants in the above theorem. For the special cases of Top-$\ell$ and ordered load balancing, 
we can get much better results. 

\begin{theorem}
	\label{thm:top-l:loadbal}
	There is a polynomial time $2$-approximation for the Top-$\ell$-load balancing problem.
\end{theorem}
\begin{theorem}
	\label{thm:single-ordered:loadbal}
	There is a polynomial time $(2+\e)$-approximation for the ordered load balancing problem, for any constant $\e>0$.
\end{theorem}

As shown by the reduction in~\Cref{minnorm}, the key
component needed to tackle 
the norm-minimization problem 
is an algorithm for the {\em \minmax multi-ordered load-balancing problem},
wherein we are given {\em multiple} non-increasing weight vectors
$w^{(1)},\ldots,w^{(N)}\in\R_+^m$, and our goal is to find an assignment $\sg:J\to[m]$
to minimize $\max_{r\in[N]}\obj(w^{(r)};\lvec_\sg)$. 

\begin{theorem} \label{thm:simordlb}\label{simordlb} [Min-max ordered load balancing]~
	
	Given any non-increasing weight vectors
	$w^{(1)},\ldots,w^{(N)}\in\R_+^m$, we can find $38(1+\dt)$-approximation algorithm  
	to the \minmax ordered load balancing problem of 
	finding an assignment $\sg:J\to[m]$ minimizing $\max_{r\in[N]}\obj(w^{(r)};\lvec_\sg)$. 
	The algorithm runs in time $\poly\bigl(\text{input size},m^{O(1/\dt)}\bigr)$.
\end{theorem}

As per the framework described in~\Cref{sec:minmaxapproach}, in~\Cref{sec:loadbal:lp} we 
write an LP-relaxation for the (single) ordered optimization problem. 
Then in~\Cref{sec:loadbal:deto} we describe a {\em deterministic, weight-oblivious
	rounding} scheme which implies~\Cref{thm:simordlb}. 
Finally, in~\Cref{sec:loadbal:splcase}, we describe simpler and better 
rounding algorithms proving~\Cref{thm:top-l:loadbal} and~\Cref{thm:single-ordered:loadbal}.
These rounding algorithms are randomized (and oblivious), but their derandomizations are not.
Nevertheless, we encourage the reader to first read~\Cref{sec:loadbal:splcase} as a warm up to the deterministic, weight-oblivious rounding.
%
%

\subsection{Linear programming relaxation}\label{sec:loadbal:lp}

We begin by restating some definitions from~\Cref{sec:proxy} in the load balancing setting.
As usual, $\vec{o}$ will denote the 
load-vector induced by an optimal assignment for the problem under
consideration. Recall that $\pset=\pset_{m,\dt}:=\{\min\{\ceil{(1+\dt)^s},m\}: s\geq 0\}$ is the sparse set of $O(\log m/\dt)$ indices. 
For $\ell\in\pset$,
$\next(\ell)$ is the smallest index in $\pset$ larger than $\ell$ if $\ell<m$, and is
$m+1$ otherwise. Given $\pset$, recall the sparsified weight vector $\tw$ of any weight vector $w$;
every $i\in[m]$, we set $\tw_i=w_{i}$ if $i\in\pset$;
otherwise, if $\ell\in\pset$ is such that $\ell<i<\next(\ell)$, we set
$\tw_i=w_{\next(\ell)}$.

Given a valid threshold vector $\vec{t}\in\R^\pset$ (i.e., $t_\ell$ is non-increasing in $\ell$)
we move from $\obj(\tw;\lvec_\sg)$ to the proxy
\begin{alignat}{1}
\prox_{\vec{t}}(\tw;\lvec_\sg)\ & :=\ 
\sum_{\ell\in\pset}\bigl(\tw_\ell-\tw_{\next(\ell)}\bigr)\ell\cdot t_\ell
+\sum_{i=1}^m h_{\vec{t}}\bigl(\tw;\load_\sg(i)\bigr), \qquad \text{where} \notag \\
h_{\vec{t}}~(\tw;a)\ & :=\ \sum_{\ell\in\pset}\bigl(\tw_\ell-\tw_{\next(\ell)}\bigr)(a-t_\ell)^+.
\label{hloadbal} \tag{Prox-LB}
\end{alignat}

Again, from~\Cref{ordproxy}, we know that for the right choice of $\vec{t}$, this
change of objective does not incur much loss, 
and so our goal is to find $\sg:J\to[m]$ that 
approximately minimizes $\sum_{i=1}^m h_{\vec{t}}\bigl(\tw;\load_\sg(i)\bigr)$ 
(see~\Cref{proxyub}). We now describe the LP relaxation to minimize the proxy-cost. Our LP is parametrized by the vector $\vec{t}$.

Before describing the LP for the ordered load balancing, let us describe the LP for the special case of \topl load balancing.
In this case, not that $\pset=\{\ell\}$, and we have a guess $t = t_\ell$ of the $\ell$th largest load.
Also recall $h_t(\load_{\sg}(i))$ is simply $(\load_{\sg}(i)-t)^+$. The LP, 
as is usual, has variables $x_{ij}$ to denote if $j$ is assigned machine $i$. This $x_{ij}$ is {\em split} into $y_{ij} + z_{ij}$
where $z_{ij}$ denotes the fraction job $j$ contributes to the load of machine $i$ in the interval $[0,t)$. In the objective, for any machine $i$, we only consider the load ``above the threshold'', that is, only the $p_{ij}y_{ij}$ portion.
\begin{alignat}{3}
\min & \quad & \olblp_{t}(x,y,z)\ :=\ \sum_i\sum_j & 
p_{ij}y_{ij} \tag{$\textrm{Top-$\ell$-LB}_{\vec{t}}$} \label{smp:ordlblp} \\
\text{s.t.} && \sum_i x_{ij} & =1 \qquad && \forall j \tag{T1} \label{smp:jasgn} \\
&& x_{ij} & = z_{ij} + y_{ij} \qquad && \forall i,j,\ \forall\ell\in\pset 
\label{smp:xyzconst} \tag{T2} \\
&& \sum_j p_{ij} z_{ij}  & \leq t \qquad && 
\forall i,\ \forall \ell\in\pset \label{smp:mctload} \tag{T3} \\
&& p_{ij}y_{ij} & \geq \bigl(p_{ij}-t\bigr)x_{ij} 
\qquad && \forall i,j,\ \forall \ell\in\pset \label{smp:jobtload} \tag{T4} \\
&& x_{ij}, z_{ij}, y_{ij} & \geq 0 \qquad && 
\forall i,j,\ \forall \ell\in\pset. \notag
\end{alignat}
The following lemma shows that the above LP is a valid relaxation.
\begin{lemma} \label{smp:newconst} \label{smp:lem:lpvalid}
	For any $t>0$ and any integral assignment $\sigma$,  
	the value of the LP is at most $\sum_{i=1}^m h_{t}~(\load_\sg(i))$.
\end{lemma}
\begin{proof}
	Given any assignment $\sigma$, set $x_{ij} = 1$ iff $\sigma(j) = i$.
	For each  $i,j$ with $x_{ij} = 1$,
	set 
	\[
	z_{ij} = \begin{cases}
	1 & \textrm{if $\load_{\sg}(i) < t$}
	\\
	\frac{t}{\load_{\sg}(i)} & \textrm{if $\load_{\sg}(i) \ge t$}
	\end{cases}
	\]
	Set $y_{ij} = x_{ij} - z_{ij}$. We claim this $(x,y,z)$ satisfies all constraints and has LP objective value equal to  $\sum_{i=1}^m h_{\vec{t}}(\tw;\load_\sg(i))$.
	
	Constraint~\eqref{smp:jasgn} is satisfied since all jobs are assigned.
	Constraint~\eqref{smp:xyzconst} is satisfied by definition. 
  	For any machine $i$, if $\load_\sg(i) < t$, then we get $\sum_{j} p_{ij} z_{ij} = \load_\sg(i) < t$.
	Otherwise, we get $\sum_{j} p_{ij} z_{ij} = t$. This implies \eqref{smp:mctload} is satisfied. 
	We also satisfy \eqref{smp:jobtload}. To see this, note the inequality is vacuous if $z_{ij} = 1$, and otherwise it is satisfied with equality.
	
	Finally note that $(\load_{\sg}(i) - t)^+ = \load_\sigma(i)y_{ij}$, for all $i$ and for all $j\in \sigma^{-1}(i)$. 
	If $\load_{\sg}(i) < t$, then both sides are $0$; otherwise, $y_{ij} = 1 - \frac{t}{\load_\sg(i)}$ for all $j$ assigned to $i$ by $\sigma$.
	Since the RHS is precisely $\sum_{i,j} p_{ij}y_{ij}$, the LP objective is precisely $\sum_{i\in [m]} h_{t}(\load_{\sg}(i))$.
	%
	%
\end{proof}
At this point, we invite the reader to skip to~\Cref{sec:loadbal:splcase} to see a rounding for just the \topl load balancing problem.
Next, we describe the LP for the general case by taking linear combinations of the above LP. \medskip

Now we write the LP for the ordered load balancing case.
Again, we use variables $x_{ij}$ to denote if job $j$
is assigned to machine $i$. Now for every $i,j$, {\em and every $\ell\in\pset$}, we have
variables $z_{ij}^{(\ell)}, y_{ij}^{(\ell)}$ to denote respectively the portions of job $j$ that lie
``below'' and ``above'' the $t_\ell$ threshold on machine $i$. 
More precisely, given an {\em integral} assignment $\sg$ and an ordering of the jobs in $\sg^{-1}(i)$,
$z_{ij}^{(\ell)}$ denotes the fraction of $j$ that contributes to the load in the interval 
$[0,t_\ell)$ on machine $i$, and $y_{ij}^{(\ell)}$ denotes the fraction of $j$ that contributes to 
the load interval $[t_\ell,\infty)$. Thus, for every $\ell$, we have 
$x_{ij}=z_{ij}^{(\ell)}+y_{ij}^{(\ell)}$, and $\sum_jp_{ij}y_{ij}^{(\ell)}$ represents
$\bigl(\load_\sg(i)-t_\ell\bigr)^+$. 
Throughout $i$ indexes the set $[m]$ of machines, and $j$ indexes the job-set $J$. 
To keep notation simple, define $z_{ij}^{(m+1)}=0$ for all $i,j$.
\begin{alignat}{3}
\min & \quad & \olblp_{\vec{t}}(\tw;x,y,z)\ :=\ \sum_i\sum_{\ell\in\pset}\sum_j & 
\bigl(\tw_\ell-\tw_{\next(\ell)}\bigr)p_{ij}y_{ij}^{(\ell)} \tag{$\textrm{OLB-P}_{\vec{t}}$} \label{ordlblp} \\
\text{s.t.} && \sum_i x_{ij} & =1 \qquad && \forall j \tag{OLB1} \label{jasgn} \\
&& x_{ij} & = z_{ij}^{(\ell)}+y_{ij}^{(\ell)} \qquad && \forall i,j,\ \forall\ell\in\pset 
\label{xyzconst} \tag{OLB2} \\
&& z_{ij}^{\next(\ell)} & \leq z_{ij}^{(\ell)} \qquad 
&& \forall i,j,\ \forall\ell\in\pset \label{zmon} \tag{OLB3} \\
&& \sum_j p_{ij}\bigl(z_{ij}^{(\ell)}-z_{ij}^{(\next(\ell))}\bigr) & \leq t_\ell-t_{\next(\ell)} \qquad && 
\forall i,\ \forall \ell\in\pset \label{mctload} \tag{OLB4} \\
&& p_{ij}y_{ij}^{(\ell)} & \geq \bigl(p_{ij}-t_\ell\bigr)x_{ij} 
\qquad && \forall i,j,\ \forall \ell\in\pset \label{jobtload} \tag{OLB5} \\
&& x_{ij}, z_{ij}^{(\ell)}, y_{ij}^{(\ell)} & \geq 0 \qquad && 
\forall i,j,\ \forall \ell\in\pset. \notag
\end{alignat}
\begin{lemma} \label{newconst} \label{lem:lpvalid}
	For any valid threshold vector $\vec{t}$ and any integral assignment $\sigma$,  
	the value of the LP is at most $\sum_{i=1}^m h_{\vec{t}}~(\tw;\load_\sg(i))$.
\end{lemma}
\begin{proof}
	Given any assignment $\sigma$, set $x_{ij} = 1$ iff $\sigma(j) = i$.
	For each $\ell\in\pset$ and for $i,j$ with $x_{ij} = 1$,
	set 
	\[
	z^{(\ell)}_{ij} = \begin{cases}
	1 & \textrm{if $\load_{\sg}(i) < t_\ell$}
	\\
	\frac{t_\ell}{\load_{\sg}(i)} & \textrm{if $\load_{\sg}(i) \ge t_\ell$}
	\end{cases}
	\]
	Set $y^{(\ell)}_{ij} = x_{ij} - z^{(\ell)}_{ij}$. We claim this $(x,y,z)$ satisfies all constraints and has LP objective value equal to  $\sum_{i=1}^m h_{\vec{t}}(\tw;\load_\sg(i))$.
	
	Constraint~\eqref{jasgn} is satisfied since all jobs are assigned.
	Constraint~\eqref{xyzconst} is satisfied by definition, and \eqref{zmon} is satisfied
	since $t_\ell \geq t_{\next(\ell)}$. For any machine $i$ and $\ell\in \pset$, if $\load_\sg(i) < t_\ell$, then we get $\sum_{j} p_{ij} z^{(\ell)}_{ij} = \load_\sg(i) < t_\ell$.
	Otherwise, we get $\sum_{j} p_{ij} z^{(\ell)}_{ij} = t_\ell$. This implies \eqref{mctload} is satisfied since $t_\ell \geq t_{\next(\ell)}$.
	We also satisfy \eqref{jobtload}. To see this, note the inequality is vacuous if $z^{(\ell)}_{ij} = 1$, and otherwise it is satisfied with equality.
	
	Finally note that $(\load_{\sg}(i) - t_\ell)^+ = \load_\sigma(i)y^{(\ell)}_{ij}$, for all $i$ and for all $j\in \sigma^{-1}(i)$. 
	If $\load_{\sg}(i) < t_\ell$, then both sides are $0$; otherwise, $y^{(\ell)}_{ij} = 1 - \frac{t_\ell}{\load_\sg(i)}$ for all $j$ assigned to $i$ by $\sigma$.
	Since the RHS is precisely $\sum_{i,j} p_{ij}y^{(\ell)}_{ij}$, the LP objective is precisely $\sum_{i\in [m]} h_{\vec{t}}(\tw;\load_{\sg}(i))$.
	%
	%
\end{proof}

In~\Cref{sec:loadbal:splcase} (which, as we encourage, can be read before moving further), we show a simple {\em randomized rounding} algorithm.
As discussed in~\Cref{sec:minmaxapproach}, we need a {\em deterministic}, {\em weight-oblivious}
rounding algorithm. The main technical contribution of this section is precisely such a rounding procedure.

\begin{theorem}\label{oblroundthm}\label{thm:oblro:lb}\label{obllbthm}
	{\em (Deterministic weight-oblivious rounding for load balancing.)}~
	
	\noindent
	Let $\vec{t}$ be a valid threshold vector such that
	every $t_\ell$ is either a power of $2$ or $0$.
	There is a {\em deterministic} algorithm which takes any solution $(x,y,z)$ satisfying constraints~\eqref{jasgn}-\eqref{jobtload}, 
	and produces an assignment $\tsg:J \to[m]$ such that, for any 
	sparsified weight vector $\nw$, we have that 
	\begin{equation}\label{eq:obr}\sum_{i=1}^mh_{10\vec{t}}~\bigl(\nw;\load_{\tsg}(i)\bigr)\leq
	2\cdot\olblp_{\vec{t}}~(\nw;x,y,z)+4\sum_{\ell\in\pset}\nw_\ell t_\ell \end{equation}
\end{theorem}

Note that the algorithm doesn't use the weights; rather the fixed output satisfies \eqref{eq:obr} for {\em all} weights simultaneously.
We prove this theorem in~\Cref{sec:loadbal:deto} which can be directly skipped to. 
In the remainder of this section, we use the theorem to prove~\Cref{thm:simordlb} and~\Cref{thm:normlb}.

\begin{proof}[{\bf Proof of~\Cref{thm:simordlb}}]~
	We sparsify $w^{(r)}$ to $\tw^{(r)}$ for all $r\in[N]$; recall $\pset = \pset_{m,\dt}$.		
	Let $\vec{o}$ be the load-vector induced by an optimal solution. 
	Let	$\iopt :=\max_{r\in[k]}\obj(w^{(r)};\vec{o})$. 
	
	Using the enumeration procedure in~\Cref{polyguess} with $\e = 1$ and finding a $\rho$ that is a power of $2$ such that
	$\vo_1\leq\rho\leq 2\vo_1$, we may assume that we have obtained  a valid threshold vector
	$\vec{t}$ where all $t_\ell$s are powers of $2$ or $0$, and which satisfies the conditions:
	$\vo_\ell\leq t_\ell\leq 2\vo_\ell$ if $\vo_\ell\geq\frac{\vo_1}{m}$, and $t_\ell=0$
	otherwise. 
	
	%
	We now solve an LP similar to \eqref{ordlblp} with the objective modified to encode the min-max-ness.
	\begin{alignat}{6}
	\min \qquad\lambda :&& (x,y,z) ~\textrm{satisfies \eqref{jasgn} - \eqref{jobtload}} \label{lp:comb} \\
	&&  \sum_{\ell\in \pset}(\tw^{(r)}_{\ell} - \tw^{(r)}_{\next(\ell)})\ell t_\ell ~~+~~ \olblp_{\vec{t}}(\nw^{(r)};x,y,z) ~~\leq \lambda && \qquad \forall r \in [N] \label{eq:multi}
	\end{alignat}
	
	\begin{claim} \label{clm:jujubes}
		Let $\lambda^*$ be the optimum solution to the LP above.
		Then, $\lambda^* \leq 3\iopt$.	
	\end{claim}
	\begin{proof}
		Let $\sigma^*$ be the optimal integral assignment attaining $\iopt$ and $(x,y,z)$ be the assignment described by this integral assignment as in the proof of~\Cref{lem:lpvalid}.
		For any $r\in [N]$, we therefore get $\LP_{\vec{t}}(\tw^{(r)}; x,y,z) = \sum_{i=1}^m h_{\vec{t}}~(\tw^{(r)};\vec{o})$.
		Thus, from the definition of $\prox$ and using \eqref{eq:multi}, we get $\lambda^* \leq \max_{r\in[N]}\prox_{\vec{t}}(\tw^{(r)}; \vo)$.
		Finally, \Cref{proxyopt} (with $\e=1$) gives that
		for any $r\in [N]$, $\prox_{\vec{t}}(\tw^{(r)}; \vo) \leq 3\cost(\tw^{(r)};\vec{o})$. 
	\end{proof}
	
	Given the optimal solution $(x,y,z)$ to the above LP, we use \Cref{oblroundthm} (since we have ensured that the $t_\ell$'s are powers of $2$ or $0$)
	to obtain an assignment $\tsg:J\to[m]$. We get for {\em any} $r\in [N]$,
	
	\begin{alignat}{1}
	\prox_{10\vec{t}}\bigl(\tw^{(r)};\lvec_{\tsg}\bigr) & =
	\sum_{\ell\in\pset}\bigl(\tw^{(r)}_\ell-\tw^{(r)}_{\next(\ell)}\bigr)\ell\cdot 10t_\ell
	+\sum_{i=1}^mh_{10\vec{t}}\bigl(\tw^{(r)};\load_{\tsg}(i)\bigr) \notag \\
	& \leq
	10 \sum_{\ell\in\pset}\bigl(\tw^{(r)}_\ell-\tw^{(r)}_{\next(\ell)}\bigr)\ell\cdot t_\ell
	+2\cdot\olblp_{\vec{t}}~\bigl(\tw^{(r)};x,y,z\bigr)
	+4\sum_{\ell\in\pset}\tw^{(r)}_\ell t_\ell \notag\\
	& \leq 10\lambda^* + 
	4\sum_{\ell\in\pset}\tw^{(r)}_\ell t_\ell \qquad \leq 30\iopt + 4\sum_{\ell\in\pset}\tw^{(r)}_\ell t_\ell  \label{simlbineq1}
	\end{alignat}
	where the first inequality follows from the {\em obliviousness property} of the rounding in \Cref{oblroundthm}.
	The same rounded assignment works for all the weights simultaneously. The second inequality follows from \eqref{eq:multi}.
	The last inequality invokes \Cref{clm:jujubes}.
	Now we use the fact that $t_\ell \leq 2\vec{o}_\ell$ to get 
	$\sum_{\ell\in\pset}\tw^{(r)}_\ell t_\ell\leq 
	2\obj\bigl(\tw^{(r)};\vec{o}\bigr)\leq 2\obj\bigl(w^{(r)};\vec{o}\bigr)\leq 2\cdot\iopt$. 
	The second-last inequality follows from the sparsification property (~\Cref{wsparse}). 
	Substituting in \eqref{simlbineq1}, we get that for any $r\in [N]$, $\prox_{10\vec{t}}\bigl(\tw^{(r)};\lvec_{\tsg}\bigr)  \leq 38\iopt$.
	From \Cref{proxylb} and~\Cref{wsparse}, we get for any $r\in [N]$, 
	\[\obj(w^{(r)};\lvec_{\tsg}) \leq (1+\dt)\obj(w^{(r)};\lvec_{\tsg}) \leq (1+\dt)\prox_{10\vec{t}}\bigl(\tw^{(r)};\lvec_{\tsg}\bigr) \leq 38(1+\dt)\iopt\]
\end{proof}

\begin{proof}[{\bf Proof of \Cref{thm:normlb}}]
	This follows by combining \Cref{normredn} and \Cref{simordlb} (taking
	$\dt=\e$). We only need to show that we can obtain the estimates $\high,\lb,\ub$
	in~\ref{assum1}, \ref{assum2}, 
	and they lead to the stated running time.
	The approximation guarantee obtained is $38(1+\e)\kp(1+3\e)\leq 38\kp(1+5\e)$.
	
	Let $\vo$ be the sorted cost vector induced by an optimal assignment. 
	Let $e_i\in\R^m$ denote the vector with $1$ in coordinate $i$, and $0$s
	everywhere else. 
	We can determine in polytime if $\vo_1=0$; if not, since the $p_{ij}$s are integers, we
	have $\vo_1\geq 1$, and $\iopt\geq f(\vo_1e_1)\geq\lb:=f(e_1)$ since $f$ is monotone. 
	Consider the assignment where $\sg(j):=\argmin_{i\in[m]}p_{ij}$ for each job $j$. We
	have 
	$$
	\iopt\leq f\bigl(\lvec_\sg\bigr)\leq\sum_j f(p_{\sg(j)j}e_{\sg(j)})
	=\ub:=\sum_j f(p_{\sg(j)j}e_i).
	$$
	The second inequality follows from the triangle inequality; the third equality follows
	from symmetry. This also means that 
	$\vo_1\leq\high:=\sum_jp_{\sg(j)j}=\sum_j\min_{i\in[m]}p_{ij}$, since by monotonicity, we
	have $\iopt=f(\vo)\geq f(\vo_1e_1)$. So $\ub/\lb=\high$ and
	$\log\bigl(\frac{n\cdot\ub\cdot\high}{\lb}\bigr)=\poly(\text{input size})$. So the running
	time of the reduction in \Cref{normredn}, and the size of the simultaneous
	ordered-load-balancing problem it creates, are 
	$\poly\bigl(\text{input size},(\frac{m}{\e})^{O(1/\e)}\bigr)$, and the entire algorithm
	runs in time 
	$\poly\bigl(\text{input size},(\frac{m}{\e})^{O(1/\e)}\bigr)$. 
\end{proof}

\subsection{Deterministic weight oblivious rounding :  proof of \Cref{thm:oblro:lb}}
\label{oblround}\label{sec:loadbal:deto}
We are given a solution $(x,y,z)$ which satisfy constraints \eqref{jasgn}-\eqref{jobtload}. 
It is convenient to do a change of variables. First, define $z_{ij}^{(m+1)}=0$ and $z_{ij}^{(0)}=x_{ij}$, and let
$y_{ij}^{(\ell)}=x_{ij}-z_{ij}^{(\ell)}$ for $\ell=0, m+1$. 
For all $\ell\in\{0\}\cup\pset$ and all $i,j$, define  
\[\bq_{ij}^{(\ell)} :=  z_{ij}^{(\ell)}-z_{ij}^{(\next(\ell))}=y_{ij}^{(\next(\ell))}-y_{ij}^{(\ell)}\]
which is nonnegative due to \eqref{zmon}. 
Now for any $\tw$, we can rewrite $\olblp_{\vec{t}}(\tw;x,y,z)$ as follows
\begin{equation}
\begin{split}
\sum_{i,j}\sum_{\ell\in\pset}\bigl(\tw_\ell-\tw_{\next(\ell)}\bigr)p_{ij}y_{ij}^{(\ell)}
& =\negthickspace
\sum_{i,j,\ell\in\pset}\negthickspace\bigl(\tw_\ell-\tw_{\next(\ell)}\bigr)%
\negthickspace\negthickspace\negthickspace
\sum_{\ell'\in\{0\}\cup\pset: \ell'<\ell}%
\negthickspace\negthickspace\negthickspace p_{ij}\bq_{ij}^{(\ell')} \\
& =\negthickspace\negthickspace\sum_{i,j,\ell'\in\{0\}\cup\pset}%
\negthickspace\negthickspace\negthickspace\negthickspace
p_{ij}q_{ij}^{(\ell')}\tw_{\next(\ell')}\ =:\ \olblp_{\vec{t}}~(\tw;\bq)
\end{split}
\label{newobj}
\end{equation}
We first give an overview of the rounding procedure. We begin by {\em filtering} $\bq$ to 
obtain $\hq\leq 2\bq$ with the property that $\hq_{ij}^{(\ell)}=0$ if $p_{ij}>2t_\ell$ for
all $i,j$ and all $\ell\in\pset$ (\Cref{filter}). 
This relies on the constraints \eqref{jobtload}.
Next, we set up an auxiliary LP \eqref{itersys} similar to \eqref{ordlblp} using
the same modified set of variables $q_{ij}^{(\ell)}$, 
which have the same intended meaning.
We include constraints \eqref{jasgn}, and \eqref{mctload} but with the RHS multiplied by   
$2$. We also include constraints 
$\sum_{i,j}p_{ij}q_{ij}^{(\ell)}\leq 2\sum_{i,j}p_{ij}\bq_{ij}^{(\ell)}$ for all
$\ell\in\pset$; the objective of \eqref{itersys} is to minimize 
$\sum_{i,j}p_{ij}q_{ij}^{(0)}$.
The latter budget constraints and the objective of \eqref{itersys} serve to ensure that
the objective values of $q$ and $\bq$ under \eqref{newobj} are comparable. Notice that
$\hq$ yields a feasible solution to this auxiliary LP.
We next use {\em iterative rounding} (that is, \Cref{iterrnd}) on this system to produce an
integral point $\tq$ that assigns every job, satisfies the other budget constraints
approximately, 
and whose objective value under \eqref{itersys} is at most that of $\hq$. 
We argue that the integral point $\tq$ yields the desired 
assignment $\tsg:J\to[m]$, where $\tsg(j)$ is set to the unique $i$ such that
$\sum_{\ell\in\{0\}\cup\pset}\tq_{ij}^{(\ell)}=1$. We now describe the algorithm in detail
and proceed to analyze it.

	\noindent
	{\bf Algorithm.}
\begin{enumerate}[label=L\arabic*., ref=L\arabic*, topsep=0.5ex, itemsep=0ex,
labelwidth=\widthof{L3.}, leftmargin=!] 
\item {\bf Filtering.} \label{filtstep}
For every job $j$ and machine $i$, we do the following. 
If $p_{ij}\leq 2t_\ell$ for all $\ell\in\pset$, then set
$\hq_{ij}^{(\ell)}=\bq_{ij}^{(\ell)}$ for all $\ell\in\{0\}\cup\pset$.
Otherwise, let $\bell\in\pset$ be the smallest index for which 
$p_{ij}>2t_\ell$. 
For every index $\ell\in\{0\}\cup\pset$, we set $\hq_{ij}^{(\ell)}=0$ if $p_{ij}>2t_\ell$,
and $\hq_{ij}^{(\ell)}=\bq_{ij}^{(\ell)}\cdot x_{ij}/y_{ij}^{(\bell)}$ otherwise
(where $0/0$ is defined as $0$).
\Cref{filter} shows that $\hq\leq 2\bq$, and $\hq$ satisfies \eqref{jasgn2}. We will
produce an integral point $\tq$ whose support is contained in that of $\hq$, so
$\tq_{ij}^{(\ell)}=1$ will imply that $p_{ij}\leq 2t_\ell$.

\item {\bf Iterative rounding.} \label{iterstep} 
Consider the following auxiliary LP.
\begin{alignat}{3}
\min & \quad & \sum_{i,j}p_{ij}q_{ij}^{(0)} & \tag{IR} \label{itersys} \\
\text{s.t.} 
&& \sum_i\sum_{\ell\in\{0\}\cup\pset}q_{ij}^{(\ell)} &= 1 \qquad && \forall j 
\label{jasgn2} \\
&& \sum_j p_{ij}q_{ij}^{(\ell)} & \leq 2\bigl(t_\ell-t_{\next(\ell)}\bigr) 
\qquad && \forall i,\ \forall \ell\in\pset \notag \\
&& \sum_{i,j} p_{ij}q_{ij}^{(\ell)} & \leq 2\sum_{i,j}p_{ij}\bq_{ij}^{(\ell)} 
\qquad && \forall \ell\in\pset \notag \\
&& q_{ij}^{(\ell)} & \geq 0 \qquad && \forall i,j,\ \forall \ell\in\{0\}\cup\pset. \notag 
\end{alignat}
%
We call the constraints, except for \eqref{jasgn2} and the non-negativity constraints,
budget constraints.  
By \Cref{filter}, $\hq$ is a feasible solution to \eqref{itersys}. 

We round $\hq$ to an integral point $\tq$ using \Cref{iterrnd}, taking $A_1=A_2$
to be the constraint matrix formed by constraints \eqref{jasgn2}, where each equality
constraint is written as a pair of $\leq$- and $\geq$- inequalities.
%
Define $\tsg:J\to[m]$ by setting $\tsg(j)$ to be the unique $i$ such that
$\sum_{\ell\in\{0\}\cup\pset}b_{ij}^{(\ell)}=1$. Return $\tsg$.
\end{enumerate}
\paragraph{Analysis.}

\begin{lemma} \label{filter}
The solution $\hq$ obtained after step~\ref{filtstep} satisfies $\hq\leq 2\bq$ and
constraints \eqref{jasgn2}. 
\end{lemma}

\begin{proof}
Fix a job $j$ and a machine $i$. 
If $x_{ij}=\sum_{\ell\in\{0\}\cup\pset}\bq_{ij}^{(\ell)}=0$, 
or $p_{ij}\leq 2t_\ell$ for all $\ell\in\pset$, 
then we have $\hq_{ij}^{(\ell)}=\bq_{ij}^{(\ell)}$ for all $\ell\in\pset$.
So suppose otherwise. 
Let $\bell\in\pset$ be the smallest index for which 
$p_{ij}^{(\ell)}>2t_\ell$. 
Constraint \eqref{jobtload} for $i,j,\bell$ 
implies that $y_{ij}^{(\bell)}\geq 0.5x_{ij}>0$.
It follows that $\hq_{ij}^{(\ell)}\leq 2\bq_{ij}^{(\ell)}$ for all $\ell\in\{0\}\cup\pset$.
Also, since $p_{ij}\leq 2t_\ell$ for $\ell\in\{0\}\cup\pset$ iff $\ell<\bell$, we have
$y_{ij}^{(\bell)}=\sum_{\ell\in\{0\}\cup\pset:p_{ij}\leq 2t_\ell}\bq_{ij}^{(\ell)}$.
Therefore, $\sum_{\ell\in\{0\}\cup\pset}\hq_{ij}^{(\ell)}=x_{ij}$, and so $\hq$
satisfies \eqref{jasgn2}.
\end{proof}

We summarize the properties satisfied by the integral point $\tq$ obtained 
by rounding $\hq$ using \Cref{iterrnd}.

\begin{lemma} \label{iter}
The $\{0,1\}$-point $\tq\geq 0$ obtained in step~\ref{iterstep} satisfies
$\sum_{i,j}p_{ij}\tq_{ij}^{(0)}\leq 2\sum_{i,j}p_{ij}\bq_{i,j}^{(0)}$, constraints
\eqref{jasgn2}, and
\begin{alignat}{2}
\sum_j p_{ij}\tq_{ij}^{(\ell)} & \leq 0 
\qquad && \forall i,\ \forall \ell\in\pset: t_\ell=t_{\next(\ell)} \label{mctload2a} \\
\sum_j p_{ij}\tq_{ij}^{(\ell)} & \leq 6t_\ell-2t_{\next(\ell)} 
\qquad && \forall i,\ \forall \ell\in\pset: t_\ell>t_{\next(\ell)} \label{mctload2b} \\
\sum_{i,j} p_{ij}\tq_{ij}^{(\ell)} & \leq 2\sum_{i,j}p_{ij}\bq_{ij}^{(\ell)}+4t_\ell
\qquad && \forall \ell\in\pset. \label{totload2}
\end{alignat}
\end{lemma}

\begin{proof}
These are all direct consequences of \Cref{iterrnd}.
Part (a) (of \Cref{iterrnd}) shows that
$\sum_{i,j}p_{ij}\tq_{ij}^{(0)}\leq\sum_{i,j}p_{ij}\hq_{ij}^{(0)}\leq 2\sum_{i,j}p_{ij}\bq_{ij}^{(0)}$. 
Since \eqref{jasgn2} is encoded via the constraints involving $A_1, A_2$ in the setup of
\Cref{iterrnd}, part (c) shows that \eqref{jasgn2} holds. 

Every $q_{ij}^{(\ell)}$ variable appears in at most 2 budget constraints of
\eqref{itersys}. 
If $t_\ell=t_{\next(\ell)}$, then
$\bq_{ij}^{(\ell)}=\hq_{ij}^{(\ell)}=0$ for all $i,j$. So since the support of $\tq$ is a
subset of the support of $\hq$ (part (b)), we have $\sum_j p_{ij}\tq_{ij}^{(\ell)}=0$. 
So suppose $t_\ell>t_{\next(\ell)}$, and consider the budget constraint 
$\sum_j p_{ij}q_{ij}^{(\ell)}\leq 2\bigl(t_\ell-t_{\next(\ell)}\bigr)$ 
for a given machine $i$ and $\ell\in\pset$. If $\hq_{ij}^{(\ell)}>0$, we know that
$p_{ij}\leq 2t_\ell$, so applying part (d), shows that \eqref{mctload2b} holds.
%
%
Part (d) then also shows that \eqref{totload2} holds.
\end{proof}

\begin{proof}[{\bf Finishing up the proof of \Cref{oblroundthm}}]
We first show that for any $i$ and any $\ell\in\pset$, we have that \linebreak
$\sum_{\ell'\in\pset:\ell'\geq\ell}\sum_j p_{ij}\tq_{ij}^{(\ell')}\leq 10t_\ell$. 
By \Cref{iter}, we have that
$\sum_{\ell'\in\pset:\ell'\geq\ell}\sum_j p_{ij}\tq_{ij}^{(\ell')}$ is at most \linebreak
$\sum_{\ell'\in\pset:\ell'\geq\ell, t_{\ell'}>t_{\next(\ell')}}\bigl(6t_{\ell'}-2t_{\next(\ell')}\bigr)$.
Suppose that $\ell_1<\ell_2<\ldots<\ell_a\in\pset$ are all the indices $\ell'\in\pset$ 
satisfying $\ell'\geq\ell$, $t_{\ell_1}>t_{\next(\ell_1)}$. 
Then, 
\[
\sum_{\ell'\in\pset:\ell'\geq\ell, t_{\ell'}>t_{\next(\ell')}}\!\!\!\!\!\!\!\!\!\bigl(6t_{\ell'}-2t_{\next(\ell')}\bigr)\!
\leq \bigl(6t_{\ell_1}-2t_{\next(\ell_1)}\bigr)+\ldots+\bigl(6t_{\ell_a}-2t_{\next(\ell_a)}\bigr)\!
\leq 6t_{\ell_1}+4\bigl(t_{\ell_2}+t_{\ell_3}+\ldots+t_{\ell_a}\bigr)
\]
Recall that the $t_\ell$s are all powers of $2$, or $0$. 
So $t_{\ell_2}\leq t_{\ell_1}/2$, $t_{\ell_3}\leq t_{\ell_2}/2$, and so on. 
So the RHS above is at most $6t_{\ell_1}+4t_{\ell_1}\leq 10t_\ell$.
This implies that 
$\bigl(\load_{\tsg}(i)-10t_\ell\bigr)^+\leq\sum_{\ell'\in\{0\}\cup\pset:\ell'<\ell}
\sum_jp_{ij}\tq_{ij}^{(\ell')}$.
Therefore, 
\begin{equation*}
\begin{split}
h_{10\vec{t}}\bigl(\nw;\load_{\tsg}(i)\bigr) & = 
\sum_{\ell\in\pset}\bigl(\nw_\ell-\nw_{\next(\ell)}\bigr)\bigl(\load_{\tsg}(i)-10t_\ell\bigr)^+ \\
& \leq\sum_{\ell\in\pset}\bigl(\nw_\ell-\nw_{\next(\ell)}\bigr)
\sum_{\ell'\in\{0\}\cup\pset:\ell'<\ell}\sum_jp_{ij}\tq_{ij}^{(\ell')}
= \sum_{\ell'\in\{0\}\cup\pset}\nw_{\next(\ell)}\sum_j p_{ij}\tq_{ij}^{(\ell')}. 
\end{split}
\end{equation*}
It follows that 
$\sum_{i=1}^mh_{10\vec{t}}\bigl(\nw;\load_{\tsg}(i)\bigr)\leq
\sum_{\ell\in\{0\}\cup\pset}\sum_{i,j}\nw_{\next(\ell)}p_{ij}\tq_{ij}^{(\ell)}$.
Using \Cref{iter}, we can bound the RHS by
$$
2\sum_{\ell\in\{0\}\cup\pset}\sum_{i,j}\nw_{\next(\ell)}p_{ij}\bq_{ij}^{(\ell)}
+4\sum_{\ell\in\pset}\nw_{\next(\ell)}t_{\ell}
\leq 2\cdot\olblp(\nw;\bq)+4\sum_{\ell\in\pset}\nw_\ell t_\ell. \qedhere
$$
\end{proof}
\subsection{Improved approximation for Top-$\ell$ and ordered load balancing problems}
\label{sec:loadbal:splcase}
In this section we prove \Cref{thm:top-l:loadbal} and \Cref{thm:single-ordered:loadbal}. 
Recall, in the (single) ordered load balancing problem, we have only one non-increasing weight vector $w$ and we wish to find an assignment $\sigma$ minimizing $\cost(w;\load_{\sg})$.
In the Top-$\ell$ problem, this weight vector is a $0,1$ vector. 

We prove this by rounding~\eqref{ordlblp} for a particular valid threshold vector. As usual, let $\vec{o}$ be the load vector for the optimal assignment.
For the ordered problem, we first sparsify $w$ to get $\tw$ using \Cref{wsparse} with $\delta = \epsilon$. Next, we use \Cref{polyguess} to 
to get threshold vector $\vec{t}$ satisfying (a) $\vo_\ell \leq t_\ell \leq (1+\e)\vo_\ell$ if $\vo_\ell \geq \e\vo_1/n$, and $t_\ell = 0$ otherwise.
This enumeration is what leads to the $(1+\e)$ loss. For the Top-$\ell$ load balancing problem, we can in fact exactly guess the $\vo_\ell$, that is, the $\ell$th largest 
cost in the optimum solution.

Our improved approximation algorithms follow from {\em oblivious, randomized rounding} algorithm for \eqref{ordlblp}.
In fact, the randomized algorithm would be oblivious of the guesses of $t_\ell$'s (the $t_\ell$'s will be used only to solve \eqref{ordlblp}).
However, being randomized, this algorithm doesn't suffice to give good algorithms for the \minmax problem. Indeed, the derandomized version of these algorithms are {\em not} oblivious.
Without further ado, we state and analyze the randomized algorithm in the proof of the following lemma.

\begin{lemma}
\label{lem:obl-rand-round}
There is an algorithm which takes as input $x_{ij} \in [0,1]$ for all $i,j$ pairs, and returns a random assignment $\tsg : J \to [m]$ with the following property.
For any $t$ and any $(y_{ij},z_{ij})$ satisfying (a) $y_{ij} + z_{ij} = x_{ij}$ for all $i,j$,
(b) $\sum_{j} p_{ij}z_{ij} \leq t$ for all $i$, and (c) $p_{ij}y_{ij} \ge (p_{ij} - t)x_{ij}$ for all $i,j$,
we get
$
\Exp[\sum_{i=1}^m (\load_{\tsg}(i) - 2t)^+] \leq 2\sum_{i,j} p_{ij}y_{ij}
$
\end{lemma}

\begin{proof} 
The algorithm is a randomized version of the Shmoys-Tardos algorithm~\cite{ShmoysT93} for the generalized assignment problem.
More precisely, for every machine $i$, we make $n_i = \ceil{\sum_{j\in J} x_{ij}}$ copies. Let $I_c$ be the union of the copies. Now we define a bipartite graph on the vertex set $I_c \cup J$ and define a fractional (sub)-matching $\bx$ on it. Fix a machine $i$ and consider the $n_i$ copies.
Arrange the jobs $J$ in {\em non-increasing} order of $p_{ij}$. We start with the first copy and call it active. Each job in the order tries to send $x_{ij}$ units of mass to the active copy till the total $\bx$-mass faced by the active copy equals $1$. We then move to the next copy and the job sends the remainder unit of its fraction to that copy. We continue till all jobs in $J$ distribute a total of  $\sum_j x_{ij}$ on the $n_i$ copies, and all but perhaps one of the $n_i$ copies face a fractional $x_{ij}$-mass of exactly $1$. In sum, at the end of this procedure for all machine, for every job we have $\sum_{k\in I_c} \bx_{kj} = 1$ while for every machine copy $k\in I_c$, we have $\sum_{j\in J} \bx_{kj} \leq 1$. For each machine $i$, we let $J^{(i)}_r$ be the set of jobs $j$ which have $\bx_{kj}> 0$ for the $r$th copy of machine $i$. A standard result from matching theory~\cite{LovaszP86}  gives us the following claim.

\begin{claim}\label{randmatch}
	There is a distribution $\mathsf{D}$ on {\em matchings} in this bipartite graph such that for any copy $k\in I_c$ and any job $j\in J$, we have
	\[
	\Pr_{M \leftarrow \mathsf{D}} [(k,j) \in M] \leq \bx_{kj} \leq x_{ij}
	\]
\end{claim}

The randomized rounding algorithm for load balancing samples a matching $M$ from $\mathsf{D}$ described in \Cref{randmatch}, and then allocates to machine $i$ all the jobs $j$ such that $(k,j)\in M$ for any copy $k$ of machine $i$. Let $\tsg$ be this random assignment. 

\paragraph{Analysis.}
For each machine $i$, let $Z_i$ denote the {\em random variable} indicating the {\em load} $p_{ij}$ of the job $j\in J^{(i)}_1$ allocated to the {\em first} copy of machine $i$. 

\begin{claim}\label{clm:load}
	$\load_{\tsg}(i) \leq \sum_{j\in J} p_{ij}x_{ij} + Z_i$
\end{claim}
\begin{proof}
	Since the jobs are in descending order, the load of the random job allocated to the $r+1$th copy of machine $i$ is at most 
	$\sum_{j\in J^{(i)}_r} p_{ij}\bx_{ij}$. Thus, the load on machine $i$ due to all but the job allocated to its first copy is at most 
	$\sum_{j\in J} p_{ij}x_{ij}$. The claim follows now from the definition of $Z_i$.
\end{proof}
\noindent
Thus, 
\[
(\load_{\tsg}(i) - 2t)^+ \leq \big(\sum_{j}p_{ij}x_{ij} - t\big)^+ + (Z_i - t)^+ \leq \sum_{j} p_{ij}y_{ij} + (Z_i - t)^+ 
\]
where the last inequality uses assumptions (a) and (b) of the lemma.
The proof of the lemma follows from the following claim.
\begin{claim}\label{clm:exp1}
	For any machine $i$, $\Exp[(Z_i - t)^+] \leq \sum_{j} p_{ij}y_{ij}$
\end{claim}
\begin{proof}
	If $k$ is the first copy of machine $i$, then we get 
	$\Exp[(Z_i - t)^+] = \sum_{j\in J^{(i)}_1} \Pr_{M\leftarrow \mathsf{D}} [(k,j)\in M] \cdot (p_{ij} - t)^+ \leq \sum_j (p_{ij} - t)^+ x_{ij} \leq \sum_{j} p_{ij}y_{ij}$.
	The first inequality uses \Cref{randmatch}, and the second uses assumption (c) of the Lemma.
\end{proof}\end{proof}

\begin{proof}[{\bf Proof of \Cref{thm:single-ordered:loadbal}}]
As described above, using \Cref{wsparse} and \Cref{polyguess}, we have a vector $\vec{t}$ with which we solve \eqref{ordlblp}.
Given the solution $x$, we apply \Cref{lem:obl-rand-round}. Note that for all $\ell$, 
the tuple $(t_\ell, y^{(\ell)}_{ij}, z^{(\ell)}_{ij})$ satisfies the conditions of the lemma.
Part (a) follows from~\eqref{xyzconst}, part (b) follows from adding up~\eqref{mctload} for all $1\le \ell' \le \ell$, and part (c) follows from~\eqref{jobtload}.
Therefore, we get that for each $\ell$, $\Exp[\sum_{i=1}^m (\load_{\tsg}(i) - 2t_\ell)^+] \leq 2\sum_{i,j} p_{ij}y^{(\ell)}_{ij}$. This in turn implies
$\Exp[h_{2\vec{t}}~(\tw;\load_{\tsg})] \leq 2\olblp_{\vec{t}}(\tw;x,y,z) \leq 2\sum_{i=1}^m h_{\vec{t}}(\tw;\vec{o})$, where the last inequality follows from \Cref{lem:lpvalid}. 
Using \Cref{proxyub}, we get that $\Exp[\cost(w;\lvec_{\tsg})] \leq (2 + \epsilon)\iopt$.
\end{proof}

\begin{proof}[{\bf Proof of \Cref{thm:top-l:loadbal}}]
Note that the $\e$-loss over $2$ in the previous theorem came from two sources: one is in moving to the sparsified weight vector, and the other in the guess of $\vo_\ell$'s.
For the Top-$\ell$ version of the problem, the position set $\pset = \{\ell\}$ is the singleton position $\ell$. 
The $0,1$ weight vector, in this case, coincides with the sparsified vector.
Indeed, we can just focus on the simpler~\eqref{smp:ordlblp}.
Furthermore, as we show below, we can guess $\vo_\ell$ ``exactly'' via binary search. 
In particular, for any guess $t$ of $\vo_\ell$, we solve \eqref{smp:ordlblp} of value $\olblp_t$. As per the previous proof, the algorithm described in \Cref{lem:obl-rand-round} gives a randomized algorithm with expected Top-$\ell$ cost $\leq 2(\ell t + \LP_t)$ for {\em any} $t$. Therefore, via binary search, we find the $t$ which minimizes $(\ell t + \olblp_t)$; this minimum value is $\leq \iopt$ since for $t = \vo_\ell$, the value is $\le \iopt$. For this $t$, the randomized algorithm described in the proof of~\Cref{thm:single-ordered:loadbal} returns an assignment with expected cost $\leq 2\iopt$.
\end{proof}

\paragraph{Derandomization.}
Both the above algorithms above can be easily derandomized, but this comes at the cost of obliviousness. We describe the derandomization for the Top-$\ell$ problem and the derandomization of the ordered problem is similar. In particular, for any $t$ we give a deterministic algorithm which returns an assignment $\tsg$ with $\sum_{i=1}^m (\load_{\tsg}(i) - 2t)^+ \leq 2\olblp_t$;
this will imply a deterministic $2$-approximation using~\Cref{proxyub} as it did in the proof of \Cref{thm:top-l:loadbal}.

In the proof of \Cref{lem:obl-rand-round}, when we construct the bipartite graph between jobs and the copies of the machines, introduce a cost $(p_{ij} - t)^+$ on the edges of the form $(k,j)$ where $k$ is the first copy of machine $i$ and $j\in J^{(i)}_1$.
Every other $(k,j)$ edge has cost $0$. Subsequently, find a {\em minimum cost} matching which matches every job, and  every copy of any machine which was also fractionally fully matched. 
The deterministic assignment $\tsg$ is given by this matching as in the proof of the lemma. Also as in the proof, we get that 
for any machine $i$, $(\load_{\tsg}(i) - 2t)^+ \leq \sum_j p_{ij}y_{ij} + (Z_i - t)^+$ where $Z_i$ is the $p_{kj}$ of the job assigned to the first copy.
Since we have found the matching precisely minimizing this cost, the minimum value is at most the expected value (given by any distribution, in particular, the distribution of \Cref{randmatch}), which was shown to be $\le \sum_{i,j} p_{ij}y_{ij}$ in \Cref{clm:exp1}. In sum, we can deterministically find an assignment $\sigma$ with $\sum_{i=1}^m (\load_{\tsg}(i) - 2t)^+ \leq 2\olblp_t$. 
%
%

\section{$k$-Clustering}
\label{clustering}\label{sec:clustering}
In this section, we use our framework to design constant factor approximation algorithms for the minimum-norm $k$-clustering problem.
We are given a metric space $\bigl(\D,\{c_{ij}\}_{i,j\in\D}\bigr)$, and an integer
$k\geq 0$. Let $n=|\D|$. 
For notational similarity with facility-location problems, let $\F:=\D$, denote the
candidate set of facilities.%
\footnote{Our results either directly extend, or can be adapted, to the setting where
	$\F\neq\D$.} 
A feasible solution opens a set $F\sse\F$ of at most $k$ facilities, and assigns each
client $j\in\D$ to a facility $i(j)\in F$. This results in the assignment-cost vector
$\acost:=\{c_{i(j)j}\}_{j\in\D}$. 

In {\em minimum-norm $k$-clustering}, the goal is to minimize the norm of $\acost$ under a 
given monotone, symmetric norm. The {\em ordered $k$-median}\footnote{Ideally, we would have called this the ordered $k$-clustering problem since $k$-median is a special case. We stick to the ordered median name since this is what it is called in the literature.} problem is
the special case where we are given non-increasing weights 
$w_1\geq w_2\geq\ldots\geq w_n\geq 0$, and the goal is to minimize
$\obj(w;\acost)=w^T\vc$. The {\em $\ell$-centrum} problem is the further
special case, where $w_1=1=\ldots=w_\ell$ and the remaining $w_i$s are $0$. That is,
we want to minimize the sum of the $\ell$ largest assignment costs.
%

\begin{theorem} 
	\label{normcl}\label{thm:normcl}
	Given any monotone, symmetric norm $f$ on $\R^m$ with a $\kp$-approximate ball-optimization oracle for $f$ (see~\eqref{balloracle}), 
	and for any $\e>0$, there is a 	$\kp\bigl(408+O(\e)\bigr)$-approximation algorithm for the problem of finding 
	$F\sse\F$ with $|F|\leq k$ such that the resulting assignment-cost vector
	$\acost$ minimizes $f\bigl(\acost\bigr)$. The running time of the algorithm is 
	$\poly\bigl(\text{input size},(\frac{n}{\e})^{O(1/\e)}\bigr)$. 
\end{theorem}

As shown by the reduction in Section~\ref{minnorm}, the key
component needed to tackle 
the norm-minimization problem 
is an algorithm for the {\em \minmax ordered $k$-median problem},
wherein we are given {\em multiple} non-increasing weight vectors
$w^{(1)},\ldots,w^{(N)}\in\R_+^m$, and our goal is to find an assignment $\sg:J\to[m]$
to minimize $\max_{r\in[N]}\obj(w^{(r)};\lvec_\sg)$. 

\begin{theorem} \label{thm:simordcl}\label{simordcl} [Min-max ordered $k$-median]~
	
	Given any non-increasing weight vectors
	$w^{(1)},\ldots,w^{(N)}\in\R_+^n$, we can find a $\bigl(408+O(\e)\bigr)$-approximation algorithm
	for the Min-Max Ordered $k$-median problem of finding
	a $F\sse\F$ with $|F|\leq k$ such that the
	resulting assignment-cost vector $\acost$ minimizes $\max_{r\in[N]}\obj(w^{(r)};\acost)$. 
	The running time of the algorithm is $\poly\bigl(\text{input size},n^{O(1/\e)}\bigr)$.
\end{theorem}

As per the framework described in Section~\ref{sec:minmaxapproach}, we first
write (in Section~\ref{sec:oblclus}) an LP-relaxation for the (single) ordered $k$-median problem. 
We then show a {\em deterministic, weight-oblivious
	rounding} scheme (in Section~\ref{sec:clus:deto}) which implies \Cref{thm:simordcl}. 

We have not optimized the constant in the approximation factor for easier exposition of ideas. 
For the special case of the (single) ordered $k$-median problem we can obtain a much better approximation factor.
Specifically, this improves upon the factors from~\cite{ChakrabartyS18, ByrkaSS18}. Our technique for this, however, is different from 
LP-rounding. Instead we give a combinatorial, primal-dual algorithm for the LP (as in our previous work~\cite{ChakrabartyS18}) and our 
improvement stems from the better notion of proxy costs.

\begin{theorem}	\label{thm:ordered-median}
	There is a polynomial time $(5+\e)$-approximation for the ordered $k$-median problem, for any constant $\e>0$.
\end{theorem}

\subsection{Linear programming relaxation} 
\label{oblclus} \label{sec:oblclus} \label{sec:clus:lp}
We begin by restating some notions from Sections~\ref{sparsify} and~\ref{sec:proxy} in the clustering setting.
As always, we let $\vec{o}$ denote the costs induced by an optimal solution. 
For convenience, we use $\dt=1$ in the sparsification described in Section~\ref{sparsify}.
Therefore, the relevant positions for us is $\pset=\pset_{n,1}:=\{\min\{2^s,n\}: s\geq 0\}$. 
For $\ell\in\pset$, recall that $\next(\ell)$ is the smallest index in $\pset$ larger than 
$\ell$ if $\ell<n$, and is $n+1$ otherwise.
Given a weight vector $w\in\R_+^n$ (with
non-increasing coordinates), we sparsify it to $\tw$, that is, 
%
for every $r\in[n]$, we set $\tw_r=w_{r}$ if $r\in\pset$;
otherwise, if $\ell\in\pset$ is such that $\ell<r<\next(\ell)$, we set
$\tw_r=w_{\next(\ell)}$.  
Recall from \Cref{simpsparse} that for any vector $v\in \R_+^n$, we have $\obj(\tw;v)\leq \obj(w;v)\leq 2\obj(\tw;v)$.

Given any valid {\em threshold vector} $\vec{t}\in\R_+^\pset$ with non-increasing coordinates,  
we have  the proxy function 

\begin{alignat}{1}
\prox_{\vec{t}}(\tw;v)\ & :=\ 
\bigl(\tw_\ell-\tw_{\next(\ell)}\bigr)\ell\cdot t_\ell
+\sum_{j\in\D}h_{\vec{t}}(\tw;v_j), \qquad \text{where} \notag \\
h_{\vec{t}}(\tw;a)\ &:=\sum_{\ell\in\pset}\bigl(\tw_\ell-\tw_{\next(\ell)}\bigr)(a-t_\ell)^+
\label{hclus}
\end{alignat}

From Section~\ref{sec:proxy}, we know that for the right choice of $\vec{t}$, the above proxies well-approximate $\iopt$.
In particular, since $\vo_1$ takes at most $n^2$ values, we may assume that we know $\rho=\vo_1$; so
by \Cref{polyguess}, we may assume that we have 
$\vec{t}$ that satisfies: 
$\vo_\ell\leq t_\ell\leq(1+\e)\vo_\ell$ for all $\ell\in\pset$ with
$\vo_\ell\geq\frac{\e\vo_1}{n}$, and $t_\ell=0$ for all other $\ell\in\pset$. 
In this section, it will be convenient to set $t_\ell=\frac{\e t_1}{n}$ whenever
$t_\ell=0$. Then, we have $\vo_\ell\leq t_\ell\leq (1+\e)\vo_\ell+\frac{\e t_1}{n}$ for
all $\ell\in\pset$, and $\vo_1\leq t_1\leq (1+\e)\vo_1$ (in particular); 
if these conditions hold then we say that $\vec{t}$ {\em well-estimates} $\vo$. 
As per \Cref{proxyub}, we
focus on 
the problem of finding an assignment-cost vector
$\acost$ that (approximately) minimizes $\sum_{j\in\D} h_{\vec{t}}(\tw;\acost_j)$. 
Our LP relaxation below for this is parametrized by the threshold vector $\vec{t}$.\smallskip

We augment the standard $k$-median LP for this 
{\em non-metric} $k$-median 
problem. A key extra feature is the set of valid constraints~\eqref{tnum}. 
These are crucially exploited in the
rounding algorithm. 
In the sequel, we always use $i$ to index $\F$ and $j$ to index $\D$.
\begin{alignat}{3}
\min & \quad & \ocllp(\tw;y)\ :=\ \sum_{j,i} & h_{\vec{t}}(\tw;c_{ij})x_{ij} 
\tag{$\textrm{OCl-P}_{\vec{t}}$} \label{ordcluslp} \\
\text{s.t.} && \sum_i x_{ij} & \geq 1 \qquad && \frall j \label{casgn1} \tag{OCl-1}\\ 
&& 0\leq x_{ij} & \leq y_i && \frall i,j \label{openf1} \tag{OCl-2}\\
&& \sum_i y_i & \leq k. \label{kmed1} \tag{OCl-3}\\
&& \sum_{i:c_{ij}\leq r}y_i & \geq 1 \qquad &&
\forall j, r:\ \exists \ell\in\pset\ \ \text{s.t.}\ \ 
\bigl|\{k\in\D: c_{jk}\leq r-t_\ell\}\bigr|>\ell \label{tnum} \tag{OCl-4}
\end{alignat}

\begin{remark}\label{rem:greedy}
	We note that the fractional setting of the $y$-variables implies the setting of the $x$-variables:
	if $c_{ij}<c_{i'j}$, then we use $i$ fully before
	using $i'$, that is, if $\bx_{i'j}>0$ then $\bx_{ij}=\by_i$. 
	Given $y$, this is the optimal setting of $x$ since the order of the $c_{ij}$'s and $h_{\vec{t}}(c_{ij})$'s are the same.
\end{remark}

\begin{lemma}
	\label{lem:lpvalid:clus}
	Let $\vec{t}$ be threshold vector that well-estimates $\vo$. Then $\ocllp~(\tw;y) \le h_{\vec{t}}~(\tw;\vec{o})$.
\end{lemma}
\begin{proof}
	Consider the optimal solution whose assignment costs are $\vec{o}$. Consider the solution $y^*_i = 1$ 
	for every opened facility, and $y^*_i = 0$ otherwise. $x^*_{ij} = 1$ if client $j$ is assigned facility $i$.
	Note that $\ocllp(\tw;y^*)$ is precisely $h_{\vec{t}}~(\tw;\vec{o})$.
	Constraints \eqref{casgn1}--\eqref{kmed1}, the standard
	$k$-median constraints, are clearly satisfied. 
	
	We now show that \eqref{tnum} are also satisfied by $y^*$.
	whenever $\vec{t}$ well-estimates $\vo$. 
	Consider any $j$, $r$, and index $\ell\in\pset$. 
	Since $t_\ell\geq\vo_\ell$, at most $\ell$ clients have assignment cost larger than
	$t_\ell$ in this optimal solution. 
	If no facility is opened within the ball $\{i:c_{ij}\leq r\}$, 
	then all the clients $k$ with $c_{jk}\leq r-t_\ell$ will incur assignment cost larger than
	$t_\ell$; if there are more than $\ell$ such clients then this cannot happen for this
	optimal solution, so \eqref{tnum} 
	holds for this optimal solution.
\end{proof}
As discussed in Section~\ref{sec:minmaxapproach}, our approach to min-max ordered optimization is via deterministic, weight-oblivious rounding
of an LP for the ordered optimization problem. The theorem below formalizes this for the clustering problem.

\begin{theorem} \label{oblclusthm} 	{\em (Deterministic weight-oblivious rounding for $k$-clustering.)}~
	
	Let $\vec{t}$ be a valid threshold vector that well-estimates $\vo$.
	There is a deterministic, weight-oblivious rounding procedure which given a solution
	$(\bx,\by)$ satisfying \eqref{casgn1}--\eqref{tnum}, produces a set $F\sse\F$ with $|F|\leq k$ and a  
	resulting assignment-cost vector $\acost$ which has the property that for {\em any}
	sparsified weight vector $\nw$, 
	we have  
	$\sum_{j\in\D} h_{44\vec{t}}(\nw;\acost_j)\leq
	44\cdot\ocllp(\nw;\by)+40\sum_{\ell\in\pset}\nw_\ell\next(\ell)t_\ell$.
\end{theorem}

The theorem implies that $(\bx,\by)$ is an optimal solution to \eqref{ordcluslp}, then we
obtain an $O(1)$-approximation for ordered $k$-median. 
We remark that 
Byrka et al.~\cite{ByrkaSS18} show that a
{\em randomized} rounding procedure of Charikar and Li~\cite{CharikarL12} for the standard
$k$-median LP 
has the property that it produces an assignment-cost vector $\acost$ satisfying
$\Exp[\sum_{j\in\D}(\acost_j-19\thr)^+]\leq\sum_{j,i}(c_{ij}-\thr)^+\bx_{ij}$ for every
$\thr\in\R_+$; that is, it gives a threshold-oblivious rounding for $\ell$-centrum. Since
$\obj(w;v)$ is a nonnegative linear combination of $\cost(\ell;v)$ terms, this also gives
a {\em randomized} weight-oblivious rounding for ordered $k$-median. 
However, as noted earlier, this randomized 
guarantee is insufficient for the purposes of utilizing it for \minmax ordered
$k$-median (and consequently min-norm $k$-clustering). 
The deterministic property in \Cref{oblclusthm}
is crucial 
and is a key distinction between our guarantee and 
the one in~\cite{ByrkaSS18}. Indeed, we need to develop various new ideas to obtain our
result. 

We prove \Cref{oblclusthm} in Section~\ref{sec:clus:deto}.
In the remainder of this section, we show how this leads 
to an $O(1)$-approximation for both \minmax ordered $k$-median, and minimum-norm
$k$-clustering. 

\begin{proof}[{\bf Proof of \Cref{thm:simordcl}}]
	We let $\iopt=\max_{r\in[N]}\obj(w^{(r)};\vo)$.
	We sparsify each $w^{(r)}$ to $\tw^{(r)}$ for all $r\in[k]$, where recall that
	we set $\dt=1$ in the procedure of Section~\ref{sparsify}; so every
	$\ell\in\pset=\pset_{n,1}$ is of the form $\min\{2^s,n\}$.
	As described above (before the description of the LP), in polynomial time we have access to 
	a threshold vector $\vec{t}$ which well-estimates the optimal assignment-cost vector $\vo$.
	More precisely, we have a polynomial sized set of guesses which contains a well-estimating vector, and for each such guess 
	we do what we describe next, and return the best solution.
	
	We solve an LP similar to \eqref{ordcluslp} with the objective modified to encode the min-max-ness.
	\begin{alignat}{6}
	\min \qquad\lambda :&& (x,y) ~\textrm{satisfies \eqref{casgn1} - \eqref{tnum}} \notag \\
	&&  \sum_{\ell\in \pset}(\tw^{(r)}_{\ell} - \tw^{(r)}_{\next(\ell)})\ell t_\ell ~~+~~ \ocllp~(\tw^{(r)};y) ~~\leq \lambda && \qquad \forall r \in [N] \notag
	\end{alignat}
	
	%
	Let $(\bx,\by)$ be an optimal solution to the above LP. Let $\acost$ be the
	assignment-cost vector obtained by applying \Cref{oblclusthm} to round
	$(\bx,\by)$. 
	Then, for {\em every} $r\in[N]$, we have
	\begin{alignat}{1}
	\prox_{44\vec{t}}\bigl(&\tw^{(r)};\acost\bigr) = 
	\sum_{\ell\in\pset}\bigl(\tw^{(r)}_\ell-\tw^{(r)}_{\next(\ell)}\bigr)\ell\cdot 44t_\ell
	+\sum_{j\in\D}h_{44\vec{t}}\bigl(\tw^{(r)};\acost_j\bigr) \notag \\
	& \ \ \leq
	\sum_{\ell\in\pset}\bigl(\tw^{(r)}_\ell-\tw^{(r)}_{\next(\ell)}\bigr)\ell\cdot 44t_\ell
	+44\cdot\ocllp\bigl(\tw^{(r)};\by\bigr)
	+40\sum_{\ell\in\pset}\tw^{(r)}_\ell\next(\ell)t_\ell. 
	\label{simclineq1}
	\end{alignat}
	The next claim bounds the third term in \eqref{simclineq1}.
	\begin{claim}
		$\sum_{\ell\in\pset}\tw^{(r)}_\ell\next(\ell)t_\ell\leq (4 + 10\e)\iopt$.
	\end{claim}
	\begin{proof}
		We first show
		$\sum_{\ell\in\pset}\tw^{(r)}_\ell\next(\ell)t_\ell\leq 
		4(1+\e)\obj\bigl(\tw^{(r)};\vo\bigr)+3\e\tw^{(r)}_1t_1$.
		This is because every $\ell\in\pset$ is of the form $\min\{2^s,n\}$; 
		so if $\ell'$ is such that $\next(\ell')=\ell$, we 
		have $\next(\ell)\leq 4(\ell-\ell')$. 
		Furthermore, $\vec{t}$ well-estimates $\vo$. 
		Therefore, for any $\ell$, 
		\begin{alignat}{6}
		\tw^{(r)}_\ell\next(\ell)t_\ell &&~\le~&& 4(1+\e)(\ell-\ell')\tw^{(r)}_\ell\vo_\ell+\frac{\e t_1}{n}\cdot\tw^{(r)}_1\next(\ell) \label{eq:tootoo}
		\end{alignat} 
		When we add for all $\ell\in \pset$, the second terms add up to $\leq 3\e t_1 \tw^{(r)}_1$ since $\sum_{\ell\in \pset} \next(\ell) \leq 3n$ (again, we use every $\ell$ is a power of $2$ except one in $n$). Since $t_1 \leq (1+\e)\vo_1$, we get the second terms add up to $\le 3\e(1+\e) \iopt \le 6\e\iopt$ since $\e\leq 1$.
		To argue about the first terms, note
		\[
		(\ell-\ell')\tw^{(r)}_\ell\vo_\ell = \sum_{j=\ell'+1}^\ell \tw^{(r)}_\ell\vo_\ell \leq \sum_{j=\ell'+1}^\ell \tw^{(r)}_j\vo_j
		\]
		where we have used the non-increasing property of both $\tw^{(r)}$ and $\vo$. Therefore, the first terms of \eqref{eq:tootoo}
		telescope to $\leq 4(1+\e)\obj(\tw^{(r)};\vo)\leq 4(1+\e)\iopt$. 
	\end{proof}
	
	Plugging the above in \eqref{simclineq1} and combining with Lemmas~\ref{proxylb}
	and~\ref{proxyopt}, we obtain that 
	\begin{equation*}
	\begin{split}
	\max_{r\in[N]}\obj\bigl(\tw^{(r)};\acost\bigr) & \leq 
	\max_{r\in[N]}\prox_{44\vec{t}}\bigl(\tw^{(r)};\acost\bigr) \\
	& \leq 44\cdot\max_{r\in[N]}\Bigl(
	\sum_{\ell\in\pset}\bigl(\tw^{(r)}_{\next(\ell)}-\tw^{(r)}_\ell\bigr)\ell\cdot t_\ell
	+\ocllp\bigl(\tw^{(r)};\by\bigr)\Bigr)+(160+400\e)\iopt \\
	& \leq 44\cdot\max_{r\in[N]}\Bigl(
	\sum_{\ell\in\pset}\bigl(\tw^{(r)}_{\next(\ell)}-\tw^{(r)}_\ell\bigr)\ell\cdot t_\ell
	+h_{\vec{t}}\bigl(\tw^{(r)};\vo\bigr)\Bigr)+(160+400\e)\iopt \\ 
	& = 44\cdot\max_{r\in[N]}\prox_{\vec{t}}\bigl(\tw^{(r)};\vo\bigr)+(160+400\e)\iopt \\
	& \leq 44(1+2\e)\cdot\max_{r\in[N]}\obj\bigl(\tw^{(r)};\vo\bigr)+(160+400\e)\iopt 
	\leq \bigl(204+O(\e)\bigr)\cdot\iopt.
	\end{split}
	\end{equation*}
	The first inequality above is due to \Cref{proxylb}; the second follows by expanding
	$\prox$ and using \eqref{simclineq1}. The third inequality follows from \Cref{lem:lpvalid:clus}.
	The next equality is simply the definition of $\prox$; the last two inequalities follow from
	Lemmas~\ref{proxyopt} and~\ref{wsparse} respectively.
	%
	Again applying \Cref{wsparse} (with $\dt=1$) gives that 
	$\max_{r\in[L]}\obj\bigl(w^{(r)};\acost\bigr)\leq 2\bigl(204+O(\e)\bigr)\cdot\iopt$. 
\end{proof}

\begin{proof}[{\bf Proof of \Cref{normcl}}]
	We combine \Cref{normredn} and \Cref{simordcl}. We only need to show
	that we can obtain the estimates $\high,\lb,\ub$ in~\ref{assum1}, \ref{assum2}, 
	and they lead to the stated running time.
	The approximation guarantee obtained is
	$\bigl(408+O(\e)\bigr)\kp(1+3\e)=\kp\bigl(408+O(\e)\bigr)$. 
	
	By scaling, we may assume that $c_{ij}\geq 1$ for every non-zero $c_{ij}$. 
	Let $\vo$ be the sorted cost vector induced by an optimal solution.
	Let $e_i\in\R^m$ denote the vector with $1$ in coordinate $i$, and $0$s
	everywhere else. 
	We can determine in polytime if $\vo_1=0$; if not, 
	we have $\vo_1\geq 1$, and $\iopt\geq f(\vo_1e_1)\geq\lb:=f(e_1)$ since $f$ is monotone. 
	In any solution, the assignment cost of any client $j$ is at most $\max_{i\in\F}c_{ij}$.
	So $\iopt\leq f\bigl(\{\max_i c_{ij}\}_{j\in\D}\bigr)\leq
	\ub:=\sum_jf\bigl((\max_i c_{ij})e_1\bigr)$.
	This also means that $\vo_1\leq\high:=\sum_j\max_{i\in\F}c_{ij}$, since by monotonicity,
	we have $\iopt=f(\vo)\geq f(\vo_1e_1)$. So $\ub/\lb=\high$ and
	$\log\bigl(\frac{n\cdot\ub\cdot\high}{\lb}\bigr)=\poly(\text{input size})$. So the running
	time of the reduction in \Cref{normredn}, and the size of the \minmax
	ordered-$k$-median problem it creates, are 
	$\poly\bigl(\text{input size},(\frac{n}{\e})^{O(1/\e)}\bigr)$, and the entire 
	running time is $\poly\bigl(\text{input size},(\frac{n}{\e})^{O(1/\e)}\bigr)$. 
\end{proof}

\subsection{Deterministic weight oblivious rounding : proof of \Cref{oblclusthm}}\label{sec:clus:deto}
Fix a sparsified vector $\nw$. This is used {\em only} in the analysis.
Define $\bC_j:=\sum_ic_{ij}\bx_{ij}$, and $\ocllp[j]:=\sum_i h_{\vec{t}}(\nw;c_{ij})x_{ij}$
for every client $j$. For a set $S\sse\F$, and a 
vector $v\in\R^F$, we define $v(S):=\sum_{i\in S}v_i$. For any $p\in\F\cup\D$ and
$S\sse\F\cup\D$, define $c(p,S):=\min_{r\in S}c_{pr}$.

\paragraph{Overview.}

We proceed by initially following the template of the $k$-median LP-rounding
algorithm by Charikar et al.~\cite{CharikarGST02,CharikarL12}, with some subtle but important
changes. We cluster clients around nearby centers (which 
are also clients) as in~\cite{CharikarGST02,CharikarL12} to ensure that every non-cluster center $k$
is close to some cluster center $j=\ctr(k)$ (step~\ref{c1}).
Let $D$ be the set of cluster centers. 
For $j\in D$, let $F_j$ be the set of facilities that are nearer to $j$ than to any other
cluster center, $\nbr(j)$ be the cluster-center (other than itself) nearest to $j$, and
let $a_j:=c_{j\nbr(j)}$. We will eventually ensure that we open a set $F$ of facilities
such that $c(j,F)=O(a_j)$ for every $j\in D$. So for a non-cluster center $k$ for which
$a_{\ctr(k)}=O(\bC_k)$ we have $c(k,F)=O(\bC_k)$, and this will also imply
that $h_{\alpha\vec{t}}\bigl(\nw;c(k,F)\bigr)=O(1)\cdot\ocllp[k]$ for some constant $\alpha$ (see \Cref{cmove2}
). This turns out to be true for non-cluster centers $k$ which are ``far away'' from their respective cluster centers $j$.
So we can focus on the point that are ``near'' to their corresponding cluster centers; in the algorithm we use $N_j$ (for ``near'') 
the points near the center $j$. 

Moving each ``near'' non-cluster center $k$ to $\ctr(k)$ 
yields a consolidated instance, where
at each $j\in D$, we have $d_j$ clients (including $j$) co-located at $j$. 
Clearly, $(\bx,\by)$ also induces a fractional $k$-median solution 
to this consolidated instance. 
However, {\em unlike in standard $k$-median}, it
is not in general true that the LP-objective-value
$\sum_{j\in D,i}d_jh_{\vec{t}}(\nw;c_{ij})\bx_{ij}$ of the solution to the consolidated
instance is at most the LP-objective-value of $(\bx,\by)$. The reason is that 
while the clustering ensures that $\bC_j\leq\bC_k$ if $j=\ctr(k)$, this does not imply that 
$\sum_i h_{\vec{t}}(\nw;c_{ij})\bx_{ij}\leq\sum_i h_{\vec{t}}(\nw;c_{ik})\bx_{ik}$. 
Nevertheless, we show that an approximate form of this inequality holds, 
and a good solution to the consolidated instance does translate to a
good solution to the original instance (see \Cref{cmove}).

We now focus on rounding the solution to the consolidated instance. 
As in~\cite{CharikarGST02,CharikarL12}, we can obtain a more-structured fractional
solution to this consolidated instance, 
where every cluster-center $j$ is served to an extent of 
$\ny_j=\by(F_j)\geq 0.5$ by itself, and to an extent of $1-\ny_j$ by $\nbr(j)$.
We now perform another clustering step
(step~\ref{c3}), where we select some $(j,\nbr(j))$ pairs with the property that every
$k\in D$ that is not part of a pair is close to a some $j$ that belongs to a pair, and
$a_j\leq a_k$.
For standard $k$-median, it suffices to ensure that: (1) we open at most $k$ facilities,  
and (2) we open at least one facility in each pair.

However, for the oblivious guarantee, we need to impose more constraints, and this is
where we diverge substantially from~\cite{CharikarGST02,CharikarL12}.
Define $t_0:=\infty$ and $\next(0)=1$. 
Note that 
we want to compare the cost of the rounded solution for the consolidated instance 
to the cost $\sum_{j\in D}d_jh_{\al\vec{t}}(\nw;a_{j})(1-\ny_j)$ of the
above structured fractional solution, where $\al$ is a suitable constant. 
The LP solution can be used to define variables $\nq^{(\ell)}_j$ for all
$\ell\in\{0\}\cup\pset$, 
where $a_{j}\nq_j^{(\ell)}$ is intended to represent (roughly speaking) 
$(1-\ny_j)\times\bigl(\min\{a_j,\al t_\ell\}-\al t_{\next(\ell)}\bigr)^+$,
so that $\sum_{\ell\in\{0\}\cup\pset}\nw_{\next(\ell)}a_{j}\nq_j^{(\ell)}$ is
$O\bigl(h_{\al\vec{t}}(\nw;a_{j})(1-\ny_j)\bigr)$. This latter term can be charged to the LP-cost (see~\Cref{lem:q}).

Now in addition to properties (1), (2), following the template in Section~\ref{simord}, we
also seek to assign each $j\in D$ where a center is not opened to a {\em single} threshold
$t_\ell$ where $t_\ell=\Omega(a_j)$, $t_{\next(\ell)}\leq a_j$, so that:  
(3) for every $\ell\in\{0\}\cup\pset$, the total $d_ja_j$ cost 
summed over all $j\in D$ that are not open and assigned to $t_\ell$ is (roughly speaking)
comparable to $\sum_{j\in D}d_ja_j\nq_j^{(\ell)}$.
We 
apply \Cref{iterrnd} on a suitable system to round $\nq$ to an
integral solution (which specifies both the open facilities and the assignment of clients
to thresholds) satisfying the above properties.
An important property that we need in order to achieve this is, is an upper bound on
$d_j$, and this is the key place where we exploit constraint \eqref{tnum}.
Properties (1)--(3) will imply that, for a suitable constant $\alpha$, the resulting
assignment-cost vector $\acost$ for the consolidated instance satisfies 
$\sum_{j\in D}d_jh_{\al\vec{t}}(\nw;\acost_j)$ is 
$O$(cost of fractional solution for consolidated instance). Finally, \Cref{cmove} (iii) transfers this guarantee to the original
instance.  We now give the details. \smallskip

\noindent
{\bf Algorithm.}
\begin{enumerate}[label=C\arabic*., ref=C\arabic*, topsep=0.5ex, itemsep=0ex,
	labelwidth=\widthof{C3.}, leftmargin=!]
	\item {\bf Clustering I.}\ \label{c1}
	Let $S\assign\D$, and $D\assign\es$. While $S\neq\es$, we do the following.
	We pick $j\in S$ with smallest $\bC_j$. We add $j$ to $D$. 
	For every $k\in S$ (including $j$) such that $c_{jk}\leq 4\max\{\bC_j,\bC_k\}$, we remove
	$k$ from $S$, and set $\ctr(k)=j$. 

	At the end of the above loop, for every $j\in D$, define the following quantities.
	Let $F_j=\{i:c_{ij}=\min_{j'\in\D}c_{ij'}\}$ with ties broken arbitrarily, and
	$\ny_j:=\min\{1,\by(F_j)\}$.
	
	Define $\nbr(j)=\argmin_{k\in D: k\neq j}c_{jk}$ if $\ny_j<1$, again with arbitrary
	tie-breaking, and $\nbr(j)=j$ otherwise.
	Let $a_j:=c_{j\nbr(j)}$ denote the distance of $j$ to $\nbr(j)$.
	We define the ``near'' set $N_j := \{k\in \D: \ctr(k) = j, c_{jk} \leq 3a_j/10\}$, and let $d_j := |N_j|$.
	Let $N := \bigcup_{j\in D} N_j$.
	The consolidated instance consists of the clients in $D$, each of which has demand $d_j$.
	That is, in the consolidated instance, for every $j\in D$, we move each 
	$k\in N$ to $\ctr(k)$, and drop all other clients.

	\item {\bf Clustering II for consolidated instance.}\ \label{c3}
	We create a collection $\C$ of disjoint clusters, where each cluster consists of 
	at most two nodes of $D$, as follows.
	Initialize $S\assign D$, $\C\assign\es$. While $S\neq\es$, pick $j\in S$ with smallest
	$a_j$; break ties in favor of nodes with $\ny_j=1$.  
	Add $\{j,\nbr(j)\}$ to $\C$, and remove every $k\in S$ with
	$\{k,\nbr(k)\}\cap\{j,\nbr(j)\}\!\!\neq\!\!\es$. 
		
	\item {\bf Auxiliary LP, iterative Rounding, and facility opening.}\ \label{c4}
	Recall that $t_0=\infty$ and $\next(0)=1$. For every $j\in D$, do the following. 
	If $a_j\leq 20t_\ell$ for all $\ell\in\{0\}\cup\pset$, set
	$\nq_j^{(\ell)}=(1-\ny_j)\bigl(\min\{a_j,10t_\ell\}-10t_{\next(\ell)}\bigr)^+/a_j$ for all 
	$\ell\in\pset$. Otherwise, let $\bell\in\pset$ be the smallest index such that
	$a_j>20t_\ell$. For every $\ell\in\{0\}\cup\pset$, set $\nq_j^{(\ell)}=0$ if
	$a_j>20t_\ell$, and 
	$\nq_j^{(\ell)}=(1-\ny_j)\bigl(\min\{a_j,10t_\ell\}-10t_{\next(\ell)}\bigr)^+/(a_j-10t_{\bell})$
	otherwise. 

	\noindent
	Next we consider the following auxiliary LP which we round to open our facilities.
	%
	\begin{alignat}{3}
	\min & \quad & \sum_{j\in D}d_ja_{j}q_{j}^{(0)} & \tag{IR2} \label{iterclus} \\
	\text{s.t.} 
	&& \sum_{\ell\in\{0\}\cup\pset}q_{j}^{(\ell)} &\leq 1 \qquad && \forall j 
	\label{tasgn} \\
	&& \sum_{j\in C}\sum_{\ell\in\{0\}\cup\pset}q_j^{(\ell)} &\leq 1 \qquad && \forall C\in\C 
	\label{bundle} \\
	&& \sum_{j\in D}\sum_{\ell\in\{0\}\cup\pset}q_j^{(\ell)} &\geq |D|-k \label{kbnd} \\
	&& \sum_{j\in D}d_ja_jq_{j}^{(\ell)} & \leq \sum_{j\in D}d_ja_j\nq_{j}^{(\ell)} 
	\qquad && \forall \ell\in\pset \label{tbudget} \\
	&& q & \geq 0. \notag
	\end{alignat}

	Later, in~\Cref{lem:q} we show that $\nq$ is a feasible solution the above LP.
%
%
%
	We next use \Cref{iterrnd} to round $\nq$ to an integral point $\tq$ 
	taking $A_1$, $A_2$ to be the constraint matrix of the constraints
	\eqref{tasgn}--\eqref{kbnd}.
	We open centers at $F=\{j\in D:\sum_{\ell\in\{0\}\cup\pset}\tq_j^{(\ell)}=0\}$. This ends the description of our algorithm.
\end{enumerate}
\paragraph{Analysis.}
The analysis proceeds in a few steps. In each step we state the main lemmas and prove them later in~\Cref{sec:prooflemmas}.
The first step is to show that moving to the consolidated instance doesn't cost is much,
We start with a standard claim from~\cite{CharikarGST02} and its implication on $\ny_j$'s.
	\begin{lemma}\label{lem:cgst}
		If $j,k\in D$, then $c_{jk}\geq 4\max\{\bC_j,\bC_k\}$.
	\end{lemma}
\noindent
This implies that for any $j\in D$ and $i\notin F_j$, $c_{ij} > 2\bC_j$, which in turn implies $\ny_j\geq 1/2$.
The next lemma shows that consolidating the clients doesn't cost much; again note that {\em unlike the standard $k$-median case}, 
the LP-cost of a non-cluster point mayn't be less than of the cluster center. Nevertheless, the following lemma shows a charging is possible.
\begin{lemma} \label{cmove}
	If $k\in\D$ and $j=\ctr(k)$, then 
	$\sum_ih_{5\vec{t}}(\nw;c_{ij})\bx_{ij}\leq 5\cdot\ocllp[k]$.
\end{lemma}
In our consolidation step, we dropped the ``far'' away clients. 
The first statement in the following lemma justifies this; as we show later, our algorithm eventually open a subset $F\subseteq \F$ such that for every client $j\in D$, $c(j,F)\leq 2a_j$ (\Cref{lem:googoodoll}). 
%
The second statement shows that if the consolidated instance has a ``good'' solution, then 
the clients in $N$ also have ``small'' connection costs.
\begin{lemma} \label{cmove2}
	Let $F\sse\F$ be such that $c(j,F)\leq 2a_j$ for all $j\in D$.
	Then, for any 
	$k\in\D\sm N$, 
	we have $h_{31\vec{t}}\bigl(\nw;c(k,F)\bigr)\leq 31\cdot\ocllp[k]$.
	Also, for any $\tht\geq 0$, we have 
	$\sum_{k\in N}h_{(\tht+4)\vec{t}}\bigl(\nw;c(k,F)\bigr)\leq
	\sum_{j\in D}d_jh_{\tht\vec{t}}\bigl(\nw;c(j,F)\bigr)+4\cdot\sum_{k\in N}\ocllp[k]$.
\end{lemma}

Thus, we need to bound the connection costs of the consolidated instance. This is the heart of the proof.
First, we show that the $\nq^{(\ell)}_j$ variables defined in~\ref{c4} satisfies two properties. 
The first property is that it is a feasible solution to the auxiliary LP~\eqref{iterclus}. The second property shows 
how the $w$-weighted combination of these variables corresponding to a client $j\in D$, can be upper bounded by its fractional contribution to
the original linear program.

\begin{lemma} \label{qbnd}\label{lem:q}
	The vector $\nq$ satisfies the following two conditions. For any $j\in D$, we have
	\begin{enumerate}[noitemsep]
		\item  $\nq^{(\ell)}_j$'s are a feasible solution to \eqref{iterclus}.
		\item  $\sum_{\ell\in\{0\}\cup\pset}\nw_{\next(\ell)}a_j\nq_j^{(\ell)} \leq 4\cdot\sum_ih_{5\vec{t}}(\nw;c_{ij})\bx_{ij}$ 
	\end{enumerate}
	%
\end{lemma}

Since $\tq$ is obtained by rounding $\nq$ using~\Cref{iterrnd}, 
constraint \eqref{kbnd} ensures that the number of facilities we finally open $|F|\leq k$.
We first establish that every client in $D$ is at bounded distance from $F$ (as promised earlier). For brevity, we use $\acost_j$ to denote $c(j,F)$.
\begin{lemma} \label{clus2}\label{lem:clus2}\label{lem:googoodoll}
	We have $F\cap C\neq\es$ for every $C\in\C$, and hence, $\vec{c}_j \leq 2a_j$ for
	every $j\in D$. 
\end{lemma}

\begin{proof}
	Since $\tq$ in step~\ref{c4} is obtained by rounding $\nq$ using \Cref{iterrnd}, it
	satisfies \eqref{bundle}, and its support is contained in that of $\nq$.
	Since $\tq$ satisfies \eqref{bundle}, if $C\in\C$ is such that $|C|=2$, then it is 
	immediate that $F\cap C\neq\es$.
	If $|C|=1$, say $C=\{k\}$, then we must
	have $\ny_k=1$, and so $\sum_{\ell\in\{0\}\cup\pset}\nq_k^{(\ell)}=0$. 
	Therefore, $\sum_{\ell\in\{0\}\cup\pset}\tq_k^{(\ell)}=0$, and so $k\in F$.
	
	Consider $j\in D$, and suppose $j\notin F$. Then, there is some 
	some $C=\{j',\nbr(j')\}\in\C$ with $a_{j'}\leq a_j$ and
	$\{j,\nbr(j)\}\cap\{j',\nbr(j'\}\neq\es$.  There is some $i\in F\cap C$, and 
	$c_{ij}\leq c_{j\nbr(j)}+c_{j'\nbr(j')}\leq 2a_j$. So $c(j,F)\leq 2a_j$.
\end{proof}
\noindent
The next lemma upper bounds the connection cost $\sum_{j\in D}d_jh_{\tht\vec{t}}\bigl(\nw;\vec{c}_j\bigr)$ for some suitable constant $\theta$, 
by the $w$-weighted cost of the solution $\tq$. Then using~\Cref{lem:q}, as a corollary, this is bounded by the LP-cost.
\begin{lemma}\label{lem:ubd}
	\begin{equation*}
	\sum_{j\in D}d_jh_{40\vec{t}}(\nw;\acost_j)\leq
	2\cdot\sum_{\ell\in\{0\}\cup\pset}\nw_{\next(\ell)}\cdot\sum_{j\in D}d_ja_j\nq_j^{(\ell)}
	+40\cdot\sum_{\ell\in\pset}\nw_{\next(\ell)}\next(\ell)t_\ell. 
	\end{equation*}
\end{lemma}
\begin{corollary}\label{cor:easy}
	\begin{equation*}
	\sum_{j\in D}d_jh_{40\vec{t}}(\nw;\acost_j)\leq
	40\sum_{k\in N} \ocllp[k]
	+40\cdot\sum_{\ell\in\pset}\nw_{\next(\ell)}\next(\ell)t_\ell. 
	\end{equation*}	
\end{corollary}
\begin{proof}
	Using \Cref{qbnd}.(ii) for all $j\in D$, we get
	\begin{alignat}{4}
	\sum_{j\in D} d_j \sum_{\ell\in\{0\}\cup\pset} \tw_{\next(\ell)} a_j \nq_{j}^{(\ell)} &~\leq~&  4\sum_{j\in D} \sum_{k\in N_j} \sum_{i} h_{5\vec{t}}(\tw;c_{ij})\bx_{ij} 
	&~\leq~ & 20 \sum_{j\in D} \sum_{k\in N_j} \ocllp[k] \notag
	\end{alignat}
	where in the first inequality we have used $d_j = |N_j|$ and \Cref{qbnd}.(ii), and in the second we have used \Cref{cmove}.
\end{proof}
\noindent
\begin{proof}[{\bf Proof of \Cref{oblclusthm}.}]
	We need to upper bound $\sum_{j\in\D}h_{44\vec{t}}(\nw;\acost_j)$. We start by splitting the clients in $\D$
	into those in $N$ and those not in $N$, and apply \Cref{cmove2} to get the following.
	\begin{equation*}
	\begin{split}
	\sum_{j\in\D}h_{44\vec{t}}(\nw;\acost_j) & \leq 
	\sum_{k\in N}h_{44\vec{t}}(\nw;\acost_k)+\sum_{k\in\D\sm N}h_{31\vec{t}}(\nw;\acost_k) \\
	& \leq \sum_{j\in \D} d_j h_{40\vec{t}}(\nw;\vec{c}_j) + 4\sum_{k\in N} \ocllp[k]
	+31\cdot\sum_{k\in\D\sm N}\ocllp[k] \notag \\
	& \leq 44\sum_{k\in N} \ocllp[k] + 31\sum_{k\notin N} \ocllp[k] + 40\cdot\sum_{\ell\in\pset}\nw_{\next(\ell)}\next(\ell)t_\ell.
	\end{split}
	\end{equation*}
	where the last inequality follows from Corollary~\ref{cor:easy}. 
\end{proof}

\subsubsection{Proofs of the Lemmas}\label{sec:prooflemmas}
\begin{proof}[{\bf Proof of~\Cref{lem:cgst}}]
	This is standard: suppose that $j$ was added to $D$ before $k$. If
	$c_{jk}<4\max\{\bC_j,\bC_k\}$, then $k$ would have been removed from $S$ at this point,
	and would never have been added to $D$.
\end{proof}
\begin{proof}[{\bf Proof of \Cref{cmove}}]
	For the proof of this lemma, and indeed that of~\Cref{cmove2}, one inequality that we will use repeatedly is that for any
	client $k$, and any $\rho \ge 0$, we have 
	$(\bC_k- \rho)^+\leq\sum_i(c_{ik}-\rho)^+\bx_{ik}$, since $x_{ik}$'s (ranging over $i$) form a probability distribution and the 
	$(z)^+$ function is convex. In particular, this implies
	\begin{equation}
	\sum_{\ell\in\pset}(\nw_\ell-\nw_{\next(\ell)})(\bC_k-t_\ell)^+\leq\ocllp[k].
	\label{bnd1}
	\end{equation}
	
	Now, since $j=\ctr(k)$, we have 
	$c_{jk}\leq 4\bC_k$. 
	From Remark~\ref{rem:greedy}, we get that for any $\rho \geq 0$,
	$\sum_i(c_{ij}-\thr)^+\bx_{ij}\leq\sum_i(c_{ij}-\thr)^+\bx_{ik}$.
	So we have
	\begin{equation*}
	\begin{split}
	\sum_i h_{5\vec{t}}(\nw;c_{ij})\bx_{ij} & =
	\sum_{\ell\in\pset}\bigl(\nw_\ell-\nw_{\next(\ell)}\bigr)\sum_i(c_{ij}-5t_\ell)^+\bx_{ij}
	\leq\sum_{\ell\in\pset}\bigl(\nw_\ell-\nw_{\next(\ell)}\bigr)\sum_i(c_{ij}-5t_\ell)^+\bx_{ik}
	\\ 
	& \leq 
	\sum_{\ell\in\pset}\bigl(\nw_\ell-\nw_{\next(\ell)}\bigr)\sum_i(c_{ik}+4\bC_k-5t_\ell)^+\bx_{ik}
	\\ 
	& \leq\sum_{\ell\in\pset}\bigl(\nw_\ell-\nw_{\next(\ell)}\bigr)\sum_i(c_{ik}-t_\ell)^+\bx_{ik}
	+4\cdot\sum_{\ell\in\pset}\bigl(\nw_\ell-\nw_{\next(\ell)}\bigr)(\bC_k-t_\ell)^+. 
	\end{split}
	\end{equation*}
	The penultimate and final inequalities follow from \Cref{triangle}.
	Using \eqref{bnd1}, the final expression above is at most $5\cdot\ocllp[k]$.
\end{proof}

\begin{proof}[{\bf Proof of \Cref{cmove2}}]
	First consider $k\in\D\sm N$ with $j=\ctr(k)$. 
	By definition, we have
	$\frac{3a_j}{10}\leq c_{jk}\leq 4\bC_k$.
	Since $c(j,F) \leq 2a_j$, we get $c(j,F)\leq \frac{40}{3}\bC_k$.
	In turn, this implies
	and 
	$c(k,F)\leq 4\bC_k+c(j,F)\leq\frac{92}{3}\cdot\bC_k\leq 31\cdot\bC_k$. 
	So 
	$h_{31\vec{t}}\bigl(\nw;c(k,F)\bigr)\leq
	31\cdot\sum_{\ell\in\pset}(\nw_\ell-\nw_{\next(\ell)})(C_k-t_\ell)^+\leq 31\cdot\ocllp[k]$,
	where the first inequality follows from \Cref{triangle} and the last inequality follows from \eqref{bnd1}.
	
	Now consider $j\in D$, and $k\in N$ with $\ctr(k)=j$. 
	Then, $c(k,F)\leq c(j,F)+4\bC_k$, so again utilizing
	\Cref{triangle}, we have
	\begin{alignat}{1}
	h_{(\tht+4)\vec{t}}\bigl(\nw;c(k,F)\bigr)
	& \leq\sum_{\ell\in\pset}\bigl(\nw_\ell-\nw_{\next(\ell)}\bigr)\bigl(c(j,F)+4\bC_k-(\tht+4)t_\ell\bigr)^+
	\notag \\
	& \leq\sum_{\ell\in\pset}\bigl(\nw_\ell-\nw_{\next(\ell)}\bigr)\bigl(c(j,F)-\tht t_\ell\bigr)^+
	+4\sum_{\ell\in\pset}\bigl(\nw_\ell-\nw_{\next(\ell)}\bigr)\bigl(\bC_k-t_\ell\bigr)^+
	\notag \\
	& \leq h_{\tht\vec{t}}\bigl(\nw;c(j,F)\bigr)+4\cdot\ocllp[k].
	\end{alignat}
	%
	Adding up these inequalities for all $k\in N$ with
	$\ctr(k)=j$, and then over all $j\in D$ gives \\
	$\sum_{k\in N}h_{(\tht+4)\vec{t}}\bigl(\nw;c(k,F)\bigr)$ on the LHS and
	$\sum_{j\in D}d_jh_{\tht\vec{t}}\bigl(\nw;c(j,F)\bigr)+\sum_{k\in N}4\cdot\ocllp[k]$.
	
\end{proof}

\begin{proof}[{\bf Proof of \Cref{qbnd}}]
	For every client $j\in D$, we show that $\sum_{\ell\in\{0\}\cup\pset}\nq_j^{(\ell)}= (1-\ny_j)$.
	This would imply $\nq$ satisfies~\eqref{tasgn}-\eqref{kbnd}, since $\ny_j\geq 1/2$, and $y$ satisfies
	\eqref{kmed}. $\nq^{(\ell)}_k$ satisfies \eqref{tbudget} trivially.
	We also we show that for every $j\in D$, 
	$\sum_{\ell\in\{0\}\cup\pset}\nw_{\next(\ell)}a_j\nq_j^{(\ell)} \leq 2h_{10\vec{t}}(\nw;a_j)(1-\ny_j)$.
	Part (b) of the lemma will follow from \Cref{treesol} which is stated and proved below.
	
	Fix $j\in D$. 
	If $a_j\leq 20t_\ell$ for all $\ell\in\{0\}\cup\pset$, then 
	$\sum_{\ell\in\{0\}\cup\pset}\nq_j^{(\ell)}=
	\frac{1-\ny_j}{a_j}\cdot\sum_{\ell\in\{0\}\cup\pset}\bigl(\min\{a_j,10t_\ell\}-10t_{\next(\ell)}\bigr)^+$.
	Noting that
	$\sum_{\ell\in\{0\}\cup\pset}\bigl(\min\{a_j,10t_\ell\}-10t_{\next(\ell)}\bigr)^+=a_j$
	(recall that $t_{n+1}=0$)
	proves part (a) in this case. 
	Also, $\sum_{\ell\in\{0\}\cup\pset}\nw_{\next(\ell)}a_j\nq_j^{(\ell)}=
	(1-\ny_j)\cdot\sum_{\ell\in\{0\}\cup\pset}\nw_{\next(\ell)}\bigl(\min\{a_j,10t_\ell\}-10t_{\next(\ell)}\bigr)^+
	=(1-\ny_j)h_{10\vec{t}}(\nw;a_j)$. The last equality uses the equivalent way of writing $h_{\vec{t}}(\cdot)$ alluded to in~\Cref{sec:minmaxapproach}.
	
	In the other case, let $\bell\in\pset$ be the smallest index such that $a_j>20t_\ell$. 
	We have $a_j\leq 20t_\ell$ for $\ell\in\{0\}\cup\pset$ iff $\ell<\bell$.
	So $\sum_{\ell\in\{0\}\cup\pset}\nq_j^{(\ell)}=
	\frac{1-\ny_j}{a_j-10t_{\bell}}\cdot
	\sum_{\ell\in\{0\}\cup\pset:\ell<\bell}\bigl(\min\{a_j,10t_\ell\}-10t_{\next(\ell)}\bigr)^+$
	and
	$\sum_{\ell\in\{0\}\cup\pset:\ell<\bell}\bigl(\min\{a_j,10t_\ell\}-10t_{\next(\ell)}\bigr)^+=a_j-10t_{\bell}$,
	so part (a) holds in this case as well.
	Since $\frac{a_j}{a_j-10t_{\bell}}\leq 2$, we also have
	$\sum_{\ell\in\{0\}\cup\pset}\nw_{\next(\ell)}a_j\nq_j^{(\ell)}\leq
	2(1-\ny_j)\cdot
	\sum_{\ell\in\{0\}\cup\pset:\ell<\bell}\nw_{\next(\ell)}\bigl(\min\{a_j,10t_\ell\}-10t_{\next(\ell)}\bigr)^+
	\leq 2(1-\ny_j)h_{10\vec{t}}(\nw;a_j)$.
	
	As mentioned earlier, the lemma follows from the following easy claim.
	\begin{claim} \label{treesol}
		We have 
		$h_{10\vec{t}}(\nw;a_j)(1-\ny_j)\leq 2\cdot\sum_ih_{5\vec{t}}(\nw;c_{ij})\bx_{ij}$ for all
		$j\in D$. 
	\end{claim}
	\begin{proofof}{\Cref{treesol}}
		Fix $j\in D$. 
		For every $i\in F_k$, where $k\in D$, $k\neq j$, we have
		$c_{jk}\leq c_{ij}+c_{ik}\leq 2c_{ij}$, and so $a_j=c_{j\nbr(j)}\leq 2c_{ij}$.
		Also $\ny_j\geq\sum_{i\in F_j}\bx_{ij}$, so $1-\ny_j\leq\sum_{i\notin F_j}\bx_{ij}$.
		So $h_{10\vec{t}}(\nw;a_j)(1-\ny_j)\leq
		\sum_{i\notin F_j}h_{10\vec{t}}(\nw;2c_{ij})\bx_{ij}
		\leq 2\cdot\sum_{i\notin F_j}h_{5\vec{t}}(\nw;c_{ij})\bx_{ij}$.
	\end{proofof}
\end{proof}

\begin{proof}[{\bf Proof of \Cref{lem:ubd}}]
	Consider any $\ell\in\pset$, and any $j\in D\sm F$. 
	\Cref{lem:clus2} implies $\vec{c}_j \leq 2a_j$.
	By definition (step~\ref{c4}), 
	we have that $\nq_j^{(\ell)}=0=\tq_j^{(\ell)}$ if $a_j>20t_\ell$.
	So we have $\sum_{\ell'\in\pset:\ell'\geq\ell}\acost_j\tq_j^{(\ell')}\leq 40t_\ell$. 
	Since $j\notin F$, we have $\sum_{\ell'\in\{0\}\cup\pset}\tq_j^{(\ell')}=1$.
	So $(\acost_j-40t_\ell)^+\leq\sum_{\ell'\in\{0\}\cup\pset:\ell'<\ell}\acost_j\tq_j^{(\ell')}$.
	Therefore, 
	\begin{equation*}
	\begin{split}
	h_{40\vec{t}}(\nw;\acost_j) & =
	\sum_{\ell\in\pset}\bigl(\nw_\ell-\nw_{\next(\ell)}\bigr)(\acost_j-40t_\ell)^+ \\
	& \leq \sum_{\ell\in\pset}\bigl(\nw_\ell-\nw_{\next(\ell)}\bigr)
	\sum_{\ell'\in\{0\}\cup\pset:\ell'<\ell}\acost_j\tq_j^{(\ell')}
	=\sum_{\ell'\in\{0\}\cup\pset}\nw_{\next(\ell)}\acost_j\tq_j^{(\ell')}.
	\end{split}
	\end{equation*}
	Note that the above bound also clearly holds if $j\in F$.
	\Cref{clus2} shows that $\acost_j\leq 2a_j$ for all $j\in D$.
	It follows that 
	$\sum_{j\in D}d_jh_{40\vec{t}}(\nw;\acost_j)\leq
	2\cdot\sum_{\ell\in\{0\}\cup\pset}\nw_{\next(\ell)}\cdot\sum_{j\in D}d_ja_j\tq_j^{(\ell)}$.
	
	If $\nq_j^{(\ell)}>0$, we have $\ny_j<1$ and $10t_{\next(\ell)}<a_j\leq 20t_\ell$. 
	We exploit constraint \eqref{tnum} 
	to show that if $\nq_j^{(\ell)}>0$, then $d_j\leq\next(\ell)$. 
	Suppose not. 
	Consider constraint \eqref{tnum} for client
	$j$, $r=\frac{4a_j}{10}$, and consider index $\next(\ell)$. Notice that $k\in N_j$ implies that
	$c_{jk}\leq\frac{4a_j}{10}-t_{\next(\ell)}$. 
	Since $d_j=|N_j|>\next(\ell)$, \eqref{tnum} enforces that 
	$\sum_{i:c_{ij}\leq r}\by_i\geq 1$. But $c_{ij}\leq r$ implies that $i\in F_j$ (otherwise,
	we would have $a_j\leq 2r$); this means that $\by(F_j)\geq 1$, and so $\ny_j=1$, which
	yields a contradiction.
	
	So $\nq_j^{(\ell)}>0$ implies that $d_ja_j\leq 20\next(\ell)t_\ell$.
	By \Cref{iterrnd}, we have that $\sum_{j\in D}d_ja_j\tq_j^{(\ell)}$ is at most 
	$\sum_{j\in D}d_ja_j\nq_j^{(0)}$ if $\ell=0$, and at most
	$\sum_{j\in D}d_ja_j\nq_j^{(\ell)}+20\next(\ell)t_\ell$ otherwise.
	%
	%
	Therefore
	$$
	\sum_{\ell\in\{0\}\cup\pset}\nw_{\next(\ell)}\cdot\sum_{j\in D}d_ja_j\tq_j^{(\ell)}
	\leq \sum_{\ell\in\{0\}\cup\pset}\nw_{\next(\ell)}\cdot\sum_{j\in D}d_ja_j\nq_j^{(\ell)}
	+20\cdot\sum_{\ell\in\pset}\nw_{\next(\ell)}\next(\ell)t_\ell. 
	$$
	Combining everything, we obtain that 
	$$
	\sum_{j\in D}d_jh_{40\vec{t}}(\nw;\acost_j)\leq
	2\cdot\sum_{\ell\in\{0\}\cup\pset}\nw_{\next(\ell)}\cdot\sum_{j\in D}d_ja_j\nq_j^{(\ell)}
	+40\cdot\sum_{\ell\in\pset}\nw_{\next(\ell)}\next(\ell)t_\ell. 
	\qedhere
	$$ 
\end{proof}
\subsection{Improved primal-dual algorithm for ordered $k$-median}
\label{improv} 
We now devise a much-improved $(5+\e)$-approximation algorithm for ordered
$k$-median. 
%
We sparsify the weight vector $w\in\R_+^n$ to $\tw$ taking $\dt=\e$ in
Section~\ref{sparsify}, where $0<\e\leq 1$. 
Let $\pset=\pset_{n,\e}=\bigl\{\min\{\ceil{(1+\e)^s},n\}: s\geq 0\}$.
By \Cref{polyguess}, we may assume that we have
$\vec{t}\in\R^\pset$ such that $\vo_\ell\leq t_\ell\leq(1+\e)\vo_\ell$ for all $\ell\in\pset$ with
$\vo_\ell\geq\frac{\e\vo_1}{n}$, and $t_\ell=0$ for all other $\ell\in\pset$. By
\Cref{proxyub}, we can then focus on the problem of finding an assignment-cost vector 
$\acost$ minimizing $\sum_{i=1}^nh_{\vec{t}}(\tw;\acost_j)$. 
We now consider the standard-$k$-median LP \eqref{primal} (i.e., we will not need  
constraints \eqref{tnum}), and its dual \eqref{dual}. 
Since $\tw$ is fixed throughout, we abbreviate $h_{\vec{t}}(\tw;\bullet)$ to
$h_{\vec{t}}(\bullet)$. 


\vspace*{-3ex}
{\centering
	\begin{minipage}[t]{0.45\textwidth}
		\begin{alignat}{3}
		\min & \quad & \sum_{j,i} h_{\vec{t}}(c_{ij}) & x_{ij} \tag{\thrlp} \label{primal} \\
		\text{s.t.} && \sum_i x_{ij} & \geq 1 \qquad && \frall j \label{casgn} \\ 
		&& 0\leq x_{ij} & \leq y_i && \frall i,j \label{openf} \\
		&& \sum_i y_i & \leq k. \label{kmed} 
		\end{alignat}
	\end{minipage}
	\qquad
	\begin{minipage}[t]{0.45\textwidth}
		\begin{alignat}{3}
		\max & \quad & \sum_{j} \alpha_j & - k\cdot\ld \tag{\thrdual} \label{dual} \\
		\text{s.t.} && \alpha_j & \leq h_{\vec{t}}(c_{ij})+\beta_{ij} \qquad && 
		\forall i,j \label{ftight} \\
		&& \sum_j\beta_{ij} & \leq \ld && \forall i \label{fpay} \\
		&& \al,\ld & \geq 0 \notag.
		\end{alignat}
	\end{minipage}
}

\medskip
Let $\OPT=\OPT_{\vec{t}}$ denote the common optimal value of \eqref{primal} and
\eqref{dual}. So $\OPT\leq\sum_j h_{\vec{t}}(\vo_j)$.
%
%
Let $\lb$ denote a lower bound on $\iopt$ such that $\log\lb$ is polynomially bounded
(e.g., we can take $\lb$ to be $\tw_1\cdot$(estimate of optimal $k$-center objective)).
We will be using the following claim which makes simple observations about the $h_{\vec{t}}~(\cdot)$ function.

\begin{claim} \label{triangle2}
	We have: (i) $h_{\vec{t}}(x)\leq h_{\vec{t}}(y)$ for any $x\leq y$;
	(ii) $h_{\tht_1\vec{t}}(x)\leq h_{\tht_2\vec{t}}(x)$ for any $\tht_1\geq\tht_2$, and any $x$;
	(iii) $h_{(\tht_1+\tht_2)\vec{t}}(x+y)\leq h_{\tht_1\vec{t}}(x)+h_{\tht_2\vec{t}}(y)$ for
	any $\tht_1,\tht_2,x,y$. 
\end{claim}

\begin{proof}
	Part (iii) is the only part that is not obvious. 
	For any $\ell\in\pset$, by part (iii) of \Cref{triangle}, we 
	If $h_{(\tht_1+\tht_2)\vec{t}}(x+y)=0$, then the
	inequality clearly holds; otherwise, 
	$h_{(\tht_1+\tht_2)\vec{t}}(x+y)=
	\sum_{\ell\in\pset}\bigl(\tw_\ell-\tw_{\next(\ell)}\bigr)(x-
	x-\thr_1+y-\thr_2\leq (x-\thr_1)^++(y-\thr_2)^+$.
\end{proof}

Our algorithm is based on the primal-dual schema coupled with Lagrangian relaxation. For
each $\ld\geq 0$, we describe a primal-dual algorithm to open a good set of facilities, and
then vary $\ld$ to obtain a convex combination of at most two solutions, called a 
{\em bi-point solution} that opens $k$ facilities. Finally, we round this bi-point solution.
The primal-dual process for a fixed $\ld\geq 0$ is very similar to the Jain-Vazirani
primal-dual process for $k$-median, which was also used in~\cite{ChakrabartyS18}.

\begin{enumerate}[label=P\arabic*., topsep=0.5ex, itemsep=0ex, labelwidth=\widthof{P2.}, leftmargin=!]
	\item \textbf{Dual-ascent.}\  
	Initialize $\D'=\D$, $\al_j=\beta_{ij}=0$ for all $i,j\in\D$, $T=\es$. 
	The clients in $\D'$ are called {\em active clients}. If $\al_j\geq h_{\vec{t}}(c_{ij})$, we
	say that $j$ reaches $i$. 
	(So if $c_{ij}\leq t_n$, then $j$ reaches $i$ from the very beginning.) 
	
	Repeat the following until all clients become inactive.
	Uniformly raise the $\alpha_j$s of all active clients, and the $\beta_{ij}$s for
	$(i,j)$ such that $i\notin T$, $j$ is active, and can reach $i$ 
	until one of the following events happen.
	\begin{enumerate}[label=$\bullet$, nosep]
		\item Some client $j\in\D$ reaches some $i$ (and previously could not reach $i$): if 
		$i\in F$, we {\em freeze} $j$, and remove $j$ from $\D'$. 
		\item Constraint \eqref{fpay} becomes tight for some $i\notin T$: we add $i$ to $T$;
		for every $j\in\D'$ that can reach $i$, we freeze $j$ and remove $j$ from $\D'$.
	\end{enumerate} 
	
	\item \textbf{Pruning.}\  
	Initialize $F\assign\es$.
	We consider facilities in $T$ in 
	{\em non-decreasing order of when they were added to $T$}. 
	When considering facility $i$, we add $i$ to $F$ if for every 
	$j\in\D$ with $\beta_{ij}>0$, we have $\beta_{i'j}=0$ for all other facilities $i'$
	currently in $F$. 
	
	\item Return $F$ as the set of centers. 
	Let $i(j)$ denote the point nearest to $j$ (in terms of $c_{ij}$) in $F$.
\end{enumerate}

Define $\pay(i):=\{j\in\D: \beta_{ij}>0\}$; for a set $S\sse\F$, define
$\pay(S):=\bigcup_{i\in S}\pay(i)$. The following theorem states the key properties
obtained from the primal-dual algorithm.

\begin{theorem} \label{lmp}
	The solution returned by the primal-dual algorithm satisfies the following.
	\begin{enumerate}[(i), topsep=0.25ex, itemsep=0.25ex, parsep=0ex,
		labelwidth=\widthof{(ii)}, leftmargin=!]
		\item 
		$3\ld|F|+\sum_{j\in\pay(F)}3h_{\vec{t}}(c_{i(j)j})+\sum_{j\notin\pay(F)}h_{3\vec{t}}(c_{i(j)j})\leq 3\sum_j\al_j$ 
		
		\item For any $j\in\D$, there is a facility $i\in F$ such that 
		$h_{3\vec{t}}(c_{ij})\leq 3h_{\vec{t}}(c_{ij})\leq 3\al_j$, and $\al_j\geq\al_k$ for all
		$k\in\pay(i)$. 
	\end{enumerate}
\end{theorem}

\begin{proof}
	Part (i) follows from the analysis in~\cite{ChakrabartyS18}, which we can simplify
	slightly using part (ii). 
	For every $i\in F$ and $j\in\pay(i)$, we have $i(j)=i$. So
	$\sum_{j\in\pay(F)}3\al_j=3\sum_{j\in\pay(F)}\bigl(\beta_{i(j)j}+h_{\vec{t}}(c_{i(j)j})\bigr)
	=3\ld|F|+\sum_{j\in\pay(F)}3h_{\vec{t}}(c_{i(j)j})$.
	Consider a client $j\notin\pay(F)$. By part (ii), which we prove below, there is some
	$i\in F$ such that $h_{3\vec{t}}(c_{ij})\leq 3\al_j$, and so 
	$h_{3\vec{t}}(c_{i(j)j})\leq 3\al_j$ (by \Cref{triangle2}, (i)). This completes the
	proof of part (i). 
	
	Part (ii) is new and follows from our pruning step. Fix
	$j\in\D$. Consider the facility $i'\in T$ that caused $j$ to freeze. If $i'\in F$, we can
	take $i=i'$ and we are done. Otherwise, there must be some facility $i\in T$ that was 
	{\em added before $i'$ to $T$}  
	such that $\pay(i)\cap\pay(i')\neq\es$. 
	Let $\tm_{i'}$ and $\tm_i$ denote the times when $i'$ and $i$ were added to $T$. Then,
	$\al_j\geq\tm_{i'}\geq\tm_i$, and for any client $k\in\pay(i)$, we have
	$\al_k\leq\tm_i$. 
\end{proof}

For $\ld=\ub:=(n+1)h_{\brho}(\max_{i,j}c_{ij})$, the
primal-dual algorithm opens only one facility. 
We now perform binary search in $[0,\ub]$ to find the ``right'' $\ld$. 
If during the binary search,
we find some $\ld$ such that the above primal-dual algorithm returns $F$ with $|F|=k$,
then part (i) of \Cref{lmp} shows that 
$\sum_j h_{3\vec{t}}(c_{i(j)j})\leq 3\OPT\leq 3\sum_j h_{\vec{t}}(\vo)$, 
and so by \Cref{proxyub}, we have 
$\obj\bigl(w;\{c_{i(j)j}\}_j\bigr)\leq 3(1+\e)(1+2\e)\iopt$.

Otherwise, we find two sufficiently close values $\ld_1<\ld_2$, primal solutions
$(F_1,i_1:\D\to F_1)$, $(F_2,i_2:\D\to F_2)$, and dual solutions
$(\al^1,\beta^1)$, $(\al^2,\beta^2)$, obtained for $\ld=\ld_1$ and 
$\ld=\ld_2$ respectively, such that $|F_1|>k>|F_2|$. 
We show here how to utilize $F_1$ and $F_2$ to obtain a simpler $9$-approximation,
and defer the proof of the following improved guarantee to Appendix~\ref{append-5apx}.

\begin{theorem} \label{5apxthm}
	Let $0<\e\leq 1$.
	If we continue the binary search until $\ld_2-\ld_1<\frac{\e\lb}{n^22^n}$, then there is a
	way of opening $k$ facilities from $F_1\cup F_2$ so that the resulting solution
	has $\obj(w;\bullet)$-cost at most 
	$\bigl(5+O(\e)\bigr)\iopt$.
\end{theorem}


\vspace*{-2ex}
\paragraph{Obtaining a $\bigl(9+O(\e)\bigr)$-approximation.} 
We continue the binary search until $\ld_2-\ld_1\leq\e\lb/n$ (assuming we do not find
$\ld$ for which $|F|=k$).
Let $a,b\geq 0$ be such that $ak_1+bk_2=k$, $a+b=1$.
A convex combination of $F_1$ and $F_2$ yields a feasible bi-point solution 
that we need to round to a feasible solution.  
Let $d_{1,j}=h_{3\vec{t}}(c_{i_1(j)j})$ and $d_{2,j}=h_{3\vec{t}}(c_{i_2(j)j})$. 
Let $C_1:=\sum_j d_{1,j}$ and $C_2:=\sum_j d_{2,j}$.
Then, 
\begin{equation*}
\begin{split}
aC_1+bC_2 
& \leq 3a\Bigl(\sum_j\al_{1,j}-k_1\ld_1\Bigr)+3b\Bigl(\sum_j\al_{2,j}-k_2\ld_2\Bigr) \\
& \leq 3a\Bigl(\sum_j\al_{1,j}-k\ld_2\Bigr)+3b\Bigl(\sum_j\al_{2,j}-k\ld_2\Bigr)+3ak_1(\ld_2-\ld_1)
\leq 3\OPT+3\e\lb.
\end{split}
\end{equation*}
If $b\geq 0.5$, then $C_2=\sum_jh_{3\vec{t}}(c_{i_2(j)j})\leq 6\sum_j h_{\vec{t}}(\vo_j)+6\e\lb$, 
so $F_2$ yields a solution of $\obj(\tw;\bullet)$-cost at most
$6(1+\e)(1+2\e)\iopt+6\e\lb(1+\e)$.
So suppose 
$a\geq 0.5$. The procedure for rounding the bi-point solution is similar to that in the
Jain-Vazirani algorithm for $k$-median, except that we derandomize their
randomized-rounding step by solving an LP.

\begin{enumerate}[label=B\arabic*., topsep=0.5ex, itemsep=0ex, labelwidth=\widthof{B2.},
	leftmargin=!] 
	\item For every $i\in F_2$, let $\sg(i)\in F_1$ denote the facility in $F_1$ closest to
	$i$ (under the $c_{ij}$ distances). If $|\sg(F_2)|<k_2$, add facilities from $F_1$ to it
	to obtain $\bF_1\sse F_1$ such that $\sg(F_2)\sse\bF_1$ and $|\bF_1|=k_2$.
	
	\item {\bf Opening facilities.}\ We will open either all facilities in $\bF_1$, or all
	facilities in $F_2$. Additionally, we will open $k-k_2$ facilities from $F_1\sm\bF_1$.
	We formulate the following LP to determine how to do this. Variable $\tht$ indicates
	if we open the facilities in $\bF_1$, and variables $z_i$ for every $i\in F_1\sm\bF_1$
	indicate if we open facility $i$. 
	\begin{alignat}{2}
	\min & \quad & \sum_{j:i_1(j)\in\bF_1}\bigl(\tht d_{1,j} & +(1-\tht)d_{2,j}\bigr)
	+\sum_{k:i_1(k)\notin\bF_1}\bigl(z_{i_1(k)}d_{1,k}+(1-z_{i_1(k)})(2d_{2,k}+d_{1,k})\bigr) 
	\tag{R-P} \label{roundlp} \\
	\text{s.t.} && \sum_{i\in F_1\sm\bF_1}z_i & \leq k-k_2, \qquad 
	\tht\in[0,1], \quad z_i\in[0,1] \ \ \forall i\in F_1\sm\bF_1.
	\end{alignat}
	The above LP is integral, and we open the facilities specified by an integral optimal
	solution (as discussed above), and assign each client to the nearest open facility.
\end{enumerate}

\medskip
We prove that: (1) \eqref{roundlp} has a fractional solution of objective value at most
$2(aC_1+bC_2)$, and (2) any integral solution $(\ttht,\tz)$ to \eqref{roundlp} yields a
feasible solution with assignment-cost vector $\acost$ such that 
$\sum_j h_{9\vec{t}}(\acost_j)$ 
is at most the objective value of $(\ttht,\tz)$. 
Together with the bound on $aC_1+bC_2$, using \Cref{proxyub}, these imply that the
solution returned has $\obj(w;\bullet)$-cost at most 
$(1+\e)(1+2\e)\cdot 9\cdot\iopt+6\e\lb(1+\e)\leq\bigl(9+O(\e)\bigr)\iopt$.

For the former, consider the solution where we set $\tht=a$, $z_i=a$ for all $i\in F_1\sm\bF_1$. 
We have $\sum_{i\in F_1\sm\bF_1}z_i=a(k_1-k_2)=k-k_2$. 
Every client $j$ with $i_1(j)\in\bF_1$ contributes
$ad_{1,j}+bd_{2,j}$ to the objective value of \eqref{roundlp}, which is also its
contribution to $aC_1+bC_2$. Consider a client $k$ with $i_1(k)\notin\bF_1$.
Its contribution to the objective value of \eqref{roundlp} is  
$ad_{1,k}+b(2d_{2,k}+d_{1,k})\leq d_{1,k}+2bd_{2,k}$, which is at most twice
its contribution to $aC_1+bC_2$ (since $a\geq 0.5$).

For the latter, suppose we have an integral solution $(\ttht,\tz)$ to \eqref{roundlp}. 
Let $\acost_j$ denote the assignment cost of client $j$ under the resulting solution.
For every $j$ with $i_1(j)\in\bF_1$, either $i_1(j)$ or $i_2(j)$ is opened, so 
$h_{9\vec{t}}(\acost_j)\leq h_{3\vec{t}}(\acost_j)\leq \ttht d_{1,j}+(1-\ttht)d_{2,j}$. 
Now consider $k$ with $i_1(k)\notin\bF_1$. If $\tz_{i_1(k)}=1$, then 
$h_{3\vec{t}}(\acost_k)\leq d_{1,k}$. Otherwise, $\acost_k$ is at most
$c_{i_2(k)k}+c_{i_2(k)\sg(i_2(k))}\leq c_{i_2(k)k}+(c_{i_2(k)k}+c_{i_1(k)k})$ since
$\sg(i_2(k))$ is the facility in $F_1$ closest to $i_2(k)$.
Applying \Cref{triangle2}, we then have 
$$
h_{9\vec{t}}(\acost_k)\leq h_{6\vec{t}}(2c_{i_2(k)k})+h_{3\vec{t}}(c_{i_1(k)k})=2d_{2,k}+d_{1,k}
$$
which is the contribution of $k$ to the objective value of $(\ttht,\tz)$.
Therefore, $\sum_j h_{9\vec{t}}(\acost_j)$ is at most the objective value of
$(\ttht,\tz)$. 

\section{Multi-budgeted ordered optimization and simultaneous optimization}  
\label{extn} \label{fairness}

In this section we show how the deterministic, weight-oblivious rounding can be used to obtain results for multi-budgeted ordered optimization, which in turn, using the results of Goel and Meyerson~\cite{GoelM06}, implies {\em constant-factor approximations} to the best simultaneous optimization factor possible for any instance of the unrelated load-balancing and $k$-clustering problem.
We begin by formally defining these problems. 

\begin{definition}[Multi-budgeted ordered optimization]
	Given a optimization problem where a solution induces a cost vector $\vec{v}$, given $N$ non-negative, non-increasing weights functions 
	$w^{(1)},\ldots,w^{(N)}$, and $N$ budgets $B_1,\ldots,B_N \in \R_+$, the multi-budgeted ordered optimization problems asks whether there exists a solution
	inducing a cost vector $\vec{v}$ such that $\obj(w^{(r)};\vec{v}) \leq B_r$ for all $1\le r\le n$.
	
	A $\rho$-approximation algorithm for this problem would either assert no such solution exists, or furnish a solution inducing a cost vector $\vec{v}$
	such that $\obj(w^{(r)};\vec{v}) \leq \rho \cdot B_r$ for all $1\le r\le n$.
\end{definition}

\begin{theorem}\label{thm:constant-multibud}
	There are $O(1)$-factor approximation algorithms for the multi-budgeted ordered (unrelated machines) load-balancing problem and for the multi-budgeted 
	ordered $k$-clustering problem.
\end{theorem}

The following is a slight modification of the definition given in~\cite{GoelM06} where they used general monotone, symmetric {\em convex} functions but their notion of approximation scaled the 
cost-vector by a factor and applied the function on it. As discussed earlier, the definition below implies the same for the original~\cite{GoelM06} notion.

\begin{definition}[Optimal simultaneous optimization factor]
	Given an {\em instance} $\cI$ of an optimization problem, an simultaneous $\alpha$-approximate solution induces a cost vector $\vec{v}$ such that
	$g(\vec{v}) \leq \alpha \iopt(g)$ where $\iopt(g) = \min_{\vec{w}} g(\vec{w})$ where $\vec{w}$ ranges over cost vectors induced by all feasible solutions.
	Let $\alpha^*_\cI$ be the smallest $\alpha$ for which an simultaneous $\alpha$-approximate solution exists for the instance $I$.
	This is the best simultaneous optimization factor for this instance.
	
	A $\rho$-approximation to the best simultaneous optimization factor takes an instance $\cI$ and returns a solution $\vec{v}$ such that 
	$g(\vec{v}) \leq \rho\alpha_\cI \iopt(g)$ for any monotone, symmetric norm $g$.
\end{definition}

The following theorem establishes the connections between the two problems via the work of Goel and Meyerson~\cite{GoelM06}.
\begin{theorem}
	A $\rho$-approximation algorithm for the multi-budgeted ordered optimization problem implies a $\rho(1+\e)$-approximation
	to the best simultaneous optmization factor for any instance.
\end{theorem}
\begin{proof}
	Using the terminology of Goel and Meyerson~\cite{GoelM06}, a vector $v \in \R^n_+$ is $\alpha$-\emph{submajorized} by $w \in \R^n_+$ if and only if for all $1\leq \ell\leq n$, 
	$\textrm{Top-}\ell(v) \leq \alpha\cdot \textrm{Top-}\ell(w)$. That is, for any $\ell$, the sum of the $\ell$ largest entries of $v$ are at most $\alpha$ times the sum of the $\ell$ largest entries of $w$. A cost vector $v$ is {\em globally $\alpha$-balanced} if it is $\alpha$-submajorized by any other feasible cost-vector $w$. Modifying the theory of majorization by Hardy, Littlewood, and Polya~\cite{HardyLP34}, Goel and Meyerson~\cite{GoelM06} establish the following.
	\begin{customthm}{GM}[Theorem 2.3,~\cite{GoelM06} (Paraphrased)]\label{thm:gm}
		A solution inducing a cost vector $v$ is simultaneous $\alpha$-approximate if and only if $v$ is globally $\alpha$-balanced.
	\end{customthm}
	
	Fix an instance $\cI$ of an optimization problem. 
	For any $1\leq \ell\le n$, let $\iopt_\ell := \min_{w} \textrm{Top-}\ell(w)$ where the minimization is over feasible cost vectors for this instance $\cI$.
	Let $\alpha^*_\cI$ be the smallest $\alpha$ for which an $\alpha$-simultaneous approximate solution exists for the instance $I$.
	By~\Cref{thm:gm}, this means that there is a solution inducing a cost vector $\vec{v^*}$ such that for all $1\leq \ell\leq d$, we have 
	$\textrm{Top-}\ell(v^*) \leq \alpha^*_\cI\cdot \iopt_\ell$.
	
	Now suppose we knew $\iopt_\ell$ for all $\ell$. Then we can use the $\rho$-approximate multi-budgeted optimization algorithm to obtain a $\rho$-approximate instance optimal solution.
	There are $n$ weight vectors where $w^{(\ell)}$ has $\ell$ ones and rest zeros.
	Via binary search, we find the smallest $A$ such that setting budgets $B_\ell := A\cdot \iopt_\ell$ and running our $\rho$-approximation algorithm, we get a feasible solution $v$
	Clearly, $A \leq \alpha^*_\cI$ since $v^*$ is the certificate for it; and $v$ is globally $\rho A$-balanced. This implies $v$ is a $\rho$-approximate instance-optimal solution to the simultaneous optimization problem.
	
	However, we don't know $\iopt_\ell$. But once again we can use the sparsification tricks done throughout the paper. First we observe that we need know only estimates of $\iopt_\ell$, and that only for the $\ell \in \pset_{n,\e} := \{\ceil{(1+\e)^s},n\}$. The latter is because for any $\ell < i < \next(\ell)$, we have $\iopt_\ell \leq \iopt_i \leq (1+\e)\iopt_\ell$. The first inequality follows from definition and 
	the second inequality follows since in the solution inducing the $\iopt_\ell$ solution, the contribution of the coordinates from $\ell$ to $i$ is at most $\e\iopt_\ell$. So
	any vector $v$ which satisfies ${\textrm Top-}\ell(v)\leq \alpha\textrm{Top-}\ell(w)$ for all $w$ only for $\ell\in \pset$, is in fact also a global $\alpha(1+\e)$-balanced vector.
	Therefore, it suffices therefore to know $\iopt_\ell$ only for the $\ell$s in $\pset$. Furthermore, with another $(1+\e)$-loss, we need only know a 
	non-increasing (valid) threshold vector  $\vec{t}$ such that $\iopt_\ell \leq \vec{t}_\ell \leq (1+\e)\iopt_\ell$ for $\ell\in\pset$. By Claim~\ref{nondec}, there are only polynomially many guesses, and for each we perform the binary search procedure described above (but only for $|\pset|$ many weight vectors.)
\end{proof}
As a corollary, using~\Cref{thm:constant-multibud}, we get
\begin{theorem}\label{thm:simul}
	There is a constant factor approximation algorithm to the best simultaneous optimization factor of any instance of the 
	unrelated machines load balancing and the $k$-clustering problem.
\end{theorem}

We now prove \Cref{thm:constant-multibud}.
\begin{proof}[{\bf Proof of \Cref{thm:constant-multibud}}]
	The theorem is a corollary of \Cref{obllbthm} and \Cref{oblclusthm}.
	We show the proof for load balancing and the proof for clustering is analogous and is omitted from the extended abstract.
	First we sparsify each weight to $\tw$ using Lemma~\ref{wsparse}. Suppose there is indeed an assignment $\vec{o}$ which matches all the budgets.
	Using the enumeration procedure in 
	Lemma~\ref{polyguess} with $\e = 1$ and finding a $\rho$ that is a power of $2$ such that
	$\vo_1\leq\rho\leq 2\vo_1$, we assume that we have obtained  a valid threshold vector
	$\vec{t}$ where all $t_\ell$s are powers of $2$ or $0$, and which satisfies the conditions:
	$\vo_\ell\leq t_\ell\leq 2\vo_\ell$ if $\vo_\ell\geq\frac{\vo_1}{m}$, and $t_\ell=0$
	otherwise. 
	
	For each such guess, we try to find a feasible solution to the LP (which is very similar to \eqref{lp:comb})
	\begin{alignat}{6}
	&& (x,y,z) ~\textrm{satisfies \eqref{jasgn} - \eqref{jobtload}} \label{lp:comb11} \\
	&&  \sum_{\ell\in \pset}(\tw^{(r)}_{\ell} - \tw^{(r)}_{\next(\ell)})\ell t_\ell ~~+~~ \olblp_{\vec{t}}(\nw^{(r)};x,y,z) ~~\leq 3B_r && \qquad \forall r \in [N] \label{eq:multi11}
	\end{alignat}
	From the proof of Claim~\ref{clm:jujubes}, we know that if there is an assignment $\sigma^*$ matching all the budgets, then for some $\vec{t}$ the above LP is feasible.
	So, if all the LPs return infeasible, we can answer infeasible. Otherwise, we get a solution $(\bx,\by,\bz)$ satisfying \eqref{jasgn} - \eqref{jobtload}, and the threshold vector 
	$\vec{t}$ satisfies the powers of $2$ condition. Now we apply \Cref{obllbthm}. We get an assignment $\tsg$, and as in the proof of \Cref{thm:simordlb}, we get 
	for all $r\in [N]$, $\obj(w^{(r)};\lvec_{\tsg}) \leq 38(1+\dt)B_r$.
\end{proof}

\section{Acknowledgements}

This work started when both authors were visiting the Simons Institute for the Theory of Computing, Berkeley,
in their Fall 2017 program of ``Bridging Discrete and Continuous Optimization.'' We gratefully acknowledge their
support and hospitality. DC also thanks David Eisenstat for pointing him to the stochastic fan-out model and its
connections with ordered optimization. \newline

\newpage

\appendix
\noindent
{\Large \bf Appendices}

\section{Refined sparsification: proof of ~\Cref{wsparse}} \label{append-wsparse}

\def\prev{\mathsf{prev}}
Recall, $\pset_{n,\dt}:=\bigl\{\min\{\ceil{(1+\dt)^s},n\}: s\geq 0\bigr\}$ and we abbreviate $\pset_{n,\dt}$ to $\pset$ in the remainder of this section, and
whenever $n,\dt$ are clear from the context. 
For $\ell\in\pset$, $\ell<n$, define $\next(\ell)$ to be the smallest index in $\pset$ larger
than $\ell$. 
Similarly we define $\prev(\ell)$.
For every index $i\in[n]$, we set $\tw_i=w_{i}$ if $i\in\pset$;
otherwise, if $\ell\in\pset$ is such that $\ell<i<\next(\ell)$ (note that $\ell<n$), set
$\tw_i=w_{\next(\ell)}=\tw_{\next(\ell)}$. For notational convenience, 
we often extend $\pset$ to $\{0\}\cup\pset\cup\{n+1\}$; in that case, $\next(0) := 1$ and
$\next(n):= n+1$. The weights are extended as $w_0 := \tw_0 := \infty$ and $w_{n+1} := \tw_{n+1} = 0$.
\Cref{wsparse} states that for any $\vv\in \R_+^n$, 
we have $\obj(\tw;\vv)\leq\obj(w;\vv)\leq(1+\dt)\obj(\tw;\vv)$. The simple direction $\obj(\tw;v)\leq\obj(w;v)$ follows since since $\tw\leq w$.

\def\last{\mathsf{last}}
 To prove the other direction, fix a cost vector $\vv$.
 Let us make a few notation-simplifying definitions. 
 For every $i\in [n]$, let $\alpha_i := w_i\vv_i$ and let $\beta_i := \tw_i\vv_i$.
 Thus, both $\alpha$'s and $\beta$'s are non-increasing, and $\alpha_\ell := \beta_\ell$, for all $\ell\in\pset$.
 For each $\ell \in \pset$, we define the set $J_\ell := \{1,\ldots,\ell-1\}$, and so $J_1 := \emptyset$. 
 Also note, $J_{n+1} := [n]$.
 The proof follows from this simple observation about ceilings.
\begin{claim}\label{clm:ceilingfact}
For any $\ell\in \pset$, $|J_{\next(\ell)}| \leq (1+\dt)\ell$. 
 \end{claim}
\begin{proof}
We need to show that $\next(\ell) - 1\leq (1+\dt)\ell$ since the LHS is the size of $J_{\next(\ell)}$.
First observe that the claim trivially holds for $\ell = n$. 
So we may assume $\ell := \ceil{(1+\dt)^s}$ for some $s\ge 0$.
We will use the following observation, $\ceil{(1+\dt)z} < 1 + (1+\dt)\ceil{z}$ for any non-negative $z$.
This follows since the ceiling of a number is at most one more than it, and the ceiling monotonically increases value.
Now, note that $\next(\ell) \le \ceil{(1+\dt)^t}$ where $t$ is the smallest integer $>s$ such that $\ceil{(1+\dt)^t} \neq \ell$.
(The inequality may occur if $\next(\ell)=n$ instead). Now apply the observation with $z:=(1+\dt)^{t-1}$; $\ceil{z} = \ell$ by definition.
So we get, $\next(\ell) \le 1 + (1+\dt)\ell$.
\end{proof}
\noindent
The proof of~\Cref{wsparse} now follows easily. First note,
\[
\cost(w;\vv) = \sum_{i=1}^n \alpha_i \le \sum_{\ell\in\pset} |J_{\next(\ell)} \setminus J_{\ell}|\alpha_\ell  
= \sum_{\ell\in\pset} \alpha_\ell \left(|J_{\next(\ell)}| - |J_{\ell}|\right)
= \sum_{\ell\in\pset} |J_{\next(\ell)}|\left(\alpha_\ell - \alpha_{\next(\ell)} \right),
\]
where the inequality above follows since $\alpha$'s are non-increasing. Now we use the fact that $\alpha_\ell = \beta_\ell$ for $\ell\in \pset$, 
and~\Cref{clm:ceilingfact}, to get
\[
\cost(w;\vv) \leq (1+\dt) \sum_{\ell\in\pset} \ell \left(\beta_\ell - \beta_{\next(\ell)} \right) = (1+\dt) \sum_{\ell\in\pset} \beta_{\next(\ell)} \left(\next(\ell) - \ell\right)
\]
Using the fact that $\beta$'s are non-increasing, we get that the last summand in the RHS is at most the sum of all the $\beta_i$'s which is 
$\cost(\tw;\vv)$. Together, we get $\cost(w;\vv)\leq(1+\dt)\cost(\tw;\vv)$.

\section{Proofs from~\Cref{minnorm}}\label{appen-norms}

\begin{proof}[{\bf Proof of~\Cref{normsubprop}}]~
	
	\noindent 
	{\it Part (i).\ } 
	Since $f$ is a norm, we have $f(2x) = 2f(x)$ and $f(x/2) = f(x)/2$.
	Since $\sgr$ is a subgradient at $x$, we have 
	$f(x) = f(2x) - f(x) \ge \sgr^\top x$ implying, $\sgr^\top x \leq f(x)$.
	On the other hand, $-f(x)/2 = f(x/2) - f(x) \ge \sgr^\top (-x/2)$, implying $\sgr^\top x \geq f(x)$.
	Hence $f(x)=\sgr^Tx$. 
	
	For any $y\in\R^n$, since
	$f(y)-f(x)\geq\sgr^T(y-x)$, using $f(x)=\sgr^\top x$ we get that $f(y)\geq\sgr^T y$.
	Also, for any $\ld\geq 0$, $f(y)-f(\ld x)=f(y)-f(x)+f(x)(1-\ld)\geq\sgr^T(y-x)+\sgr^Tx(1-\ld)=\sgr^T(y-\ld x)$. 
	Therefore $\sgr$ is a subgradient of $f$ at $\ld x$. \medskip
	
	\noindent 
	{\it Part (ii).\ } 
	Let $\hsgr$ be a subgradient of $f$ at $x$. Define
	$\sgr=(\hsgr)^+:=\{(\hsgr_i)^+\}_{i\in[n]}$. 
	We first claim that if $x_i>0$, then $\hsgr_i\geq 0$. Suppose not. Let $x^{(-i)}$ denote
	the vector with $x^{(-i)}_j=x_j$ for all $j\neq i$, and $x^{(-i)}_i=0$.
	So $f(x)\geq f\bigl(x^{(-i)}\bigr)$ by monotonicity. But
	$f\bigl(x^{(-i)}\bigr)-f(x)\geq\hsgr_i(-x_i)>0$, which yields a contradiction.
	
	Consider $y\in\R^n$. Let $I=\{i\in[n]:\hsgr_i\geq 0\}$.
	Define $y'_i=y_i$ if $i\in I$ and $0$ otherwise. By the above, we know that $x_i=0$ for  
	$i\notin I$. 
	By monotonicity, we have $f(y)\geq f(y')$, so 
	$$
	f(y)-f(x)\geq f(y')-f(x)\geq\sum_{i\in[n]}\hsgr_i(y'_i-x_i)=\sum_{i\in I}\hsgr_i(y'_i-x_i)
	=\sum_{i\in I}\sgr_i(y'_i-x_i)=\sum_{i\in[n]}\sgr_i(y_i-x_i).
	$$
	The last equality follows since $\sgr_i=0$ for all $i\notin I$, and $y_i=y'_i$ for all
	$i\in I$. \medskip
	
	\noindent 
	{\it Part (iii).\ }
	Suppose there are $i,j\in[n]$ such that $\sgr_i<\sgr_j$ but $x_i>x_j$. Let $x'$ be the
	vector obtained from $x$ by swapping $x_j$ and $x_i$: i.e., $x'_k=x_k$ for all
	$k\in[n]\sm\{i,j\}$, $x'_i=x_j$, $x'_j=x_i$. Then, 
	$f(x')-f(x)\geq\sgr^T(x'-x)=(\sgr_i-\sgr_j)(x_j-x_i)>0$, but $f(x')=f(x)$ due to symmetry,
	which gives a contradiction. This also implies that if $\sgr_i>\sgr_j$, then 
	$x_i\geq x_j$. It follows that there is a common permutation $\kp:[n]\to[n]$ such that
	$\sgr_{\kp(1)}\geq\ldots\geq\sgr_{\kp(n)}$ and $x_{\kp(1)}\geq\ldots\geq x_{\kp(n)}$. 
	Hence, 
	$f(x)=\sgr^Tx=\sum_{i\in[n]}\sgr_{\kp(i)}x_{\kp(i)}=\sgr^{\down}\cdot x^{\down}=\obj(\sgr^{\down};x)$. 
	
	We have $f\bigl(x^{(\pi)}\bigr)=f(x)=\sgr^Tx=\sgr^{(\pi)}\cdot x^{(\pi)}$. For any
	$y\in\R_n$, we have
	$f(y)=f\bigl(y^{(\pi^{-1})}\bigr)\geq\sgr^Ty^{(\pi^{-1})}=\sgr^{(\pi)}\cdot y$, so
	$f(y)-f(x)\geq\sgr^{(\pi)}\cdot(y-x)$. This shows that $\sgr^{(\pi)}$ is a subgradient of
	$f$ at $x^{(\pi)}$.
\end{proof}

\begin{proof}[{\bf Proof of~\Cref{nondec}}]
	Any non-decreasing sequence $a_1\geq a_2\geq\ldots\geq a_k$, where $a_i\in\{0\}\cup [N]$ for all
	$i\in[k]$, can be mapped bijectively to the set of $k+1$ integers 
	$N-a_1,a_1-a_2,\ldots,a_{k-1}-a_k,a_k$ from $\{0\}\cup [N]$ that add up to $N$. The number
	of such sequences of $k+1$ integers is equal to the coefficient of $x^N$ in the generating
	function $(1+x+\ldots+x^N)^k$. This is equal to the coefficient of $x^N$ in $(1-x)^{-k}$,
	which is $\binom{N+k-1}{N}$ using the binomial expansion. 
	Let $M=\max\{N,k-1\}$. We have 
	$\binom{N+k-1}{N}=\binom{N+k-1}{M}\leq\bigl(\frac{e(N+k-1)}{M}\bigr)^M\leq(2e)^M$.
\end{proof}

\begin{proof}[{\bf Proof of~\Cref{thm:minnorm}}]
	We first bound the number of oracle calls to $\A$. 
	By~\Cref{nondec}, since the enumeration of $u_1,\ldots,u_{\ell^*}$
	involved in $\W'$ requires guessing a non-increasing sequence of $O(\log n/\e)$ exponents
	from a range of size $O\bigl(\frac{1}{\e}\log(\frac{n}{\e})\bigr)$, we have
	\[|\W'|=O\bigl(|\pset|\cdot\frac{1}{\e}\log(\frac{n\cdot\ub\cdot\high}{\lb})(\frac{n}{\e})^{O(1/\e)}\bigr)
	=O\bigl(\frac{\log n}{\e^2}\log(\frac{n\cdot\ub\cdot\high}{\lb})(\frac{n}{\e})^{O(1/\e)}\bigr).\]
	The latter quantity is thus a bound on the number of calls to $\A$ and $|\W|$. 
	
	We now prove parts (i) and (ii), from which the final guarantee will follow easily.
	For part (i), consider any $w\in\W$. If $w=\frac{\lb}{n\cdot\high}\bon^n$, then 
	$\obj(w;\vv)\leq\frac{\lb}{n\cdot\high}\sum_{i\in[n]}\vv_i$. Otherwise, we know that
	$\bopt(w)\leq\kp(1+\e)$. So $w^T\vv^{\down}/f(\vv^{\down})\leq\bopt(w)\leq\kp(1+\e)$. Hence,
	$w^T\vv^{\down}\leq\kp(1+\e)f(\vv^{\down})$, or equivalently, $\obj(w;\vv)\leq\kp(1+\e)f(\vv)$.
	
	For part (ii), it suffices to show, due to~\Cref{normeqv}, that 
	$\obj(\sgr;\vv)$ is at most the stated bound, for every $\sgr\in\Scol$. So fix
	$\sgr\in\Scol$. If $\sgr_1<\frac{\lb}{n\cdot\high}$, then 
	$\sgr<\frac{\lb}{n\cdot\high}\cdot\bon^n$, 
	so $\obj(\sgr;\vv)<\obj\bigl(\frac{\lb}{n\cdot\high}\cdot\bon^n;\,\vv\bigr)
	<(1-\e)^{-1}\max_{w\in\W}\obj(w;\vv)$.
	So suppose otherwise. 
	
	Let $\ell^*\in\pset$ be the largest index for which $\sgr_\ell\geq\frac{\e\sgr_1}{n}$. 
	Since $\sgr_1\in\bigl[\frac{\lb}{n\cdot\high},\ub\bigr]$, there is some $\tw_1$
	that is a power of $(1+\e)$ such that $\sgr_1\leq\tw_1\leq(1+\e)\sgr_1$. For every
	$\ell\in\pset$, $\ell\leq\ell^*$, we have $\sgr_\ell\geq\frac{\e\tw_1}{n(1+\e)}$. Hence, 
	there are non-increasing $\tw_\ell$ values that are powers of $(1+\e)$ satisfying 
	$\tw_\ell\leq\sgr_\ell\leq(1+\e)\tw_\ell$ for all $\ell\in\pset$, $\ell\leq\ell^*$.
	Thus, there is some $\tw\in\W'$ such that $\sgr_\ell\leq\tw_\ell\leq(1+\e)\sgr_\ell$ for  
	all $\ell\geq\ell^*$ in $\pset$, and $\tw_\ell=0$ for all other $\ell\in\pset$.
	
	Next, we claim that $\tw\in\W$. For every $x\in\R_+^n$, we have
	$\obj(\sgr;x)\leq\sum_{i:\sgr_i\geq\e\sgr_1/n}\tw_ix^{\down}_i+\tfrac{\e\sgr_1}{n}\cdot nx^{\down}_1$, 
	so $(1-\e)\obj(\sgr;x)\leq\obj(\tw;x)$. Also, since $\tw\leq(1+\e)\sgr$, we have
	$\obj(\tw;x)\leq(1+\e)\obj(\sgr;x)$. 
	Since $\sgr\in\Scol$, we have 
	$\max_{x\in\ball_+(f)}\sgr^Tx=\max_{x\in\ball_+(f)}\obj(\sgr;x)=1$. 
	It follows that
	$$
	\max_{x\in\ball_+(f)}\tw^Tx=\max_{x\in\ball_+(f)}\obj(\tw;x)
	\in\Bigl[(1-\e)\max_{x\in\ball_+(f)}\obj(\sgr;x),\,(1+\e)\max_{x\in\ball_+(f)}\obj(\sgr;x)\Bigr]
	=[1-\e,1+\e].
	$$
	Therefore, the point $\hx\in\ball_+(f)$ returned by $\A$ on $\tw$ satisfies
	$\tw^T\hx\in\bigl[(1-\e)/\kp,1+\e\bigr]$ showing that $\tw\in\W$.
	
	Finally, since $\obj(\sgr;\vv)\leq\obj(\tw;\vv)/(1-\e)$, this implies that
	$\obj(\sgr;\vv)\leq(1-\e)^{-1}\max_{w\in\W}\obj(w;\vv)$.
	
	The final approximation guarantee of the theorem now follows easily. The optimum of the
	\minmax ordered-optimization problem is at most $\max_{w\in\W}\obj(w;\vec{o})$, which
	by part (i) is at most
	$\max\bigl\{\kp(1+\e)\iopt,\frac{\lb}{n\cdot\high}\cdot n\cdot\vo_1\}\leq\kp(1+\e)\iopt$.
	Therefore, $\max_{w\in\W}\obj(w;\tv)\leq\gm\kp(1+\e)\iopt$.
	By part (ii), this implies that 
	$f(\tv)\leq\gm\kp\cdot\frac{1+\e}{1-\e}\cdot\iopt\leq\gm\kp(1+3\e)\iopt$.
\end{proof}

\section{Proofs from~\Cref{ordproxy}} \label{append-ordproxy}

\begin{proof}[{\bf Proof of~\Cref{proxynear}}]
	Consider the difference $\prox_{\vec{t}}(\tw;v)-\prox_{\vec{\tdt}}(\tw;v)$. Since
	$\vec{t}\leq\vec{\tdt}$, only the second term 
	in \eqref{proxexpr} has a nonnegative contribution to this difference, and only 
	indices $\ell\in\pset$ for which $t_\ell\leq\tdt_\ell$ contribute non-negatively The total
	contribution from such indices is at most 
	$\sum_{\ell\in\pset:t_\ell\leq\tdt_\ell}(\tw_\ell-\tw_{\next(\ell)})\sum_{i=1}^n\Dt\leq n\tw_1\Dt$.
	Similarly, only the first (i.e., constant) term in \eqref{proxexpr} has a nonnegative
	contribution to the difference $\prox_{\vec{\tdt}}(\tw;v)-\prox_{\vec{t}}(\tw;v)$, and this
	contribution is at most 
	$\sum_{\ell\in\pset:t_\ell\geq\tdt_\ell}(\tw_\ell-\tw_{\next(\ell)})\ell\Dt\leq n\tw_1\Dt$.
\end{proof}

\begin{proof}[{\bf Proof of~\Cref{proxyub}}]
	The second inequality follows immediately from the first one and ~\Cref{wsparse}, so 
	we focus on showing the first inequality. 
	By ~\Cref{proxylb}, we have 
	\begin{equation*}
	\begin{split}
	\obj(\tw;v) & \leq
	\sum_{\ell\in\pset}\bigl(\tw_\ell-\tw_{\next(\ell)}\bigr)\ell\cdot\tht t_\ell
	+\sum_{i=1}^n h_{\tht\vec{t}}(\tw;v_i) \\
	& \leq\sum_{\ell\in\pset}\bigl(\tw_\ell-\tw_{\next(\ell)}\bigr)\tht\ell t_\ell
	+\gm\sum_{i=1}^nh_{\vec{t}}(\tw;\vo_i)+M \\
	& \leq\max\{\tht,\gm\}\prox_{\vec{t}}(\tw;\vo)+M
	\leq\max\{\tht,\gm\}(1+2\e)\obj(\tw;\vec{o})+M.
	\end{split}
	\end{equation*}
	The last inequality follows from ~\Cref{proxyopt}.
\end{proof}

\section{Iterative rounding of linear system: proof of~\Cref{iterrnd}} \label{append-iterrnd}
We first prove some properties of an extreme point of \eqref{iterlp}.
We call the constraints $Bq\leq d$, budget constraints. Let $N$ be the number of budget
constraints. 

\begin{lemma} \label{extreme}
	Let $q'$ be an extreme point of \eqref{iterlp}.
	Then either $q'$ is integral, or there is some tight budget constraint $(*)$ with
	support $S$ such that $\sum_{j\in S:q'_j>0}(1-q'_j)\leq k$. 
\end{lemma}

\begin{proof}
	Let $T$ denote the support of $q'$.
	It is well known (see, e.g.,~\cite{Schrijver98}) that 
	then there is an invertible submatrix $A'$ of the constraint-matrix of \eqref{iterlp},
	whose columns correspond to he support $T$, and rows correspond to $|T|$
	linearly-independent constraints that are tight at $q'$.
	So if $q''$ denotes the vector comprising the $q_j$ variables for
	$j\in T$, and $g$ denotes the right-hand-sides of these tight constraints, then
	$q'$ is the unique solution to the system $A'q''=g$.  
	
	If $A'$ does not consist of any budget constraints, 
	then the supports of the rows of $A'$ from a laminar family, and it is well known that such a
	matrix is totally unimodular (TU).
	So since $\bigl(\begin{smallmatrix} b_1\\ b_2\end{smallmatrix}\bigr)$ is integral, $q'$ is
	integral. So if $q'$ is not integral, then $A'$ contains at least one budget constraint. 
	
	Let $\Lc$ denote the laminar family formed by the supports of the rows of $A'$
	corresponding to the $A_1q\leq b_1$, $A_2q\geq b_2$ constraints.
	Consider the following token-assignment scheme. Every $j\in T$ supplies $q'_j$ tokens to
	the row of $A'$ corresponding to the smallest set of $\Lc$ containing $j$ (if such a row
	exists), 
	and $(1-q'_{j})/k$ tokens to the at most $k$ budget constraints of $A'$ where it
	appears. Thus, every $j\in T$ supplies at most one token unit overall, and the total
	supply of tokens is at most $|T|$. 
	
	Notice that every row $i$ of $A'$ corresponding to a constraint from $A_1q\leq b_1$ or 
	$A_2q\geq b_2$ consumes at least $1$ token unit: 
	let $L\in\Lc$ is the support of row $i$, and $L'\subsetneq L$ be the largest  
	set of $\Lc\cup\{\es\}$ strictly contained in $L$. 
	If $L'\neq\es$, let
	$i'$ be the row of $A'$ corresponding to set $L'$.
	Row $i$ consumes $\sum_{j\in L\sm L'}q'_j$ tokens, which is equal to $(A'q')_i-(A'q')_{i'}$
	if $L'\neq\es$, and equal to $(A'q')_i$ otherwise. This quantity is an integer, and
	strictly positive (since all $q'_j$s are positive), so is at least $1$.
	Suppose for a contradiction that, for every row $i$ corresponding to a budget constraint of $A'$, 
	$\sum_{j\in T: A'_{ij}>0}(1-q'_j)>k$. 
	Then every constraint of $A'$ consumes at least $1$ token unit, and at least one
	constraint consumes more than $1$ token unit. This yields a contradiction since the total
	consumption of tokens is larger than (number of constraints of $A'$) = $|T|$.
	
	Hence, if $q'$ is not integral, there must be some tight budget constraint $(*)$ (in fact,
	a budget constraint of $A'$) with support $S$ such that 
	$\sum_{j\in S:q'_j>0}(1-q'_j)\leq k$.  
\end{proof}

The iterative-rounding algorithm for rounding $\hq$ is as follows.
We initialize $q=\hq$, and our current system of constraints to the constraints of
\eqref{iterlp}.   
We repeat the following until we obtain an integral solution.

\begin{enumerate}[label=I\arabic*., ref=I\arabic*, topsep=0ex, itemsep=0ex,
	labelwidth=\widthof{I2.}, leftmargin=!]
	\item \label{iter1}
	Move from $q$ to an extreme-point $q'$ of the current system of constraints 
	no greater objective value (under \eqref{iterlp}) 
	whose support is contained in the support of $q$. 
	%
	If $q'$ is not integral, by~\Cref{extreme} there is some tight budget constraint
	$(*)$ with support $S$ such that $\sum_{j\in S: q'_j>0}(1-q'_j)\leq k$. 
	
	\item Set $q\assign q'$. If $q$ is not integral then update the
	system of constraints by dropping $(*)$ (and go to step~\ref{iter}); otherwise, return
	$\tq:=q$. 
\end{enumerate}

We prove that the above process terminates, and the point $\tq$
returned satisfied the stated properties. In each iteration, we drop a budget constraint,
and there are $N$ budget constraints, so we terminate in at most $N$ iterations.
By definition, we terminate with an integral point. We never increase the objective value,
and always stay within the support of $\hq$, so properties (a) and (b) hold.
We never drop a constraint from $A_1q\leq b_1$, $A_2q\geq b_2$ from our system, so the
final point $\tq$ satisfies these constraints. Since $q_j\leq 1$ is an implicit constraint
implied by these constraints (and $\tq$ is integral), this implies that $\tq\in\{0,1\}^n$.

Finally, we prove part (d). Consider a budget constraint $(Bq)_i\leq d_i$. 
If we never drop this budget constraint during iterative rounding, then $\tq$ satisfies
this constraint. Otherwise, consider the iteration when we drop this constraint and the
extreme point $q'$ obtained in \ref{iter1} just before we drop this constraint. Then, if
$S$ is the support of this budget constraint, it must be that 
$(Bq')_i\leq d_i$ and $\sum_{j\in S:q'_j>0}(1-q'_j)\leq k$. Also, the support of $\tq$ is
contained in the support of $x'$.
Therefore,
\begin{equation*}
\begin{split}
\sum_j B_{ij}\tq_{j} & \leq\sum_{j\in S:q'_j>0}B_{ij}
=\sum_{j\in S:q'_j>0}B_{ij}q'_j+\sum_{j\in S:q'_{j}>0}B_{ij}(1-q'_j) \\
& \leq (Bq')_i+k\bigl(\max_{j\in S:q'_j>0} B_{ij}\bigr)
\leq (Bq')_i+k\bigl(\max_{j:\hq_j>0} B_{ij}\bigr)=d_i+k\bigl(\max_{j:\hq_j>0}B_{ij}\bigr).
\end{split}
\end{equation*}
The last inequality follows since $q'_j>0$ implies that $\hq_j>0$.
\hfill \qed

%
%

\section{Improved $\bigl(5+O(\e)\bigr)$-approximation for ordered $k$-median}
\label{append-5apx} 
In this section, we prove~\Cref{5apxthm}.
Recall that we continue the binary search until $\ld_2-\ld_1<\frac{\e\lb}{n^22^n}$. 
For $r=1,2$, and $i\in\F$, define $\pay^r(i):=\{j\in\D: \beta^r_{ij}>0\}$; for a set
$S\sse\F$, define $\pay^r(S):=\bigcup_{i\in S}\pay^r(i)$.
A continuity argument from~\cite{CharikarG99} shows the following; we defer the proof to
the end of this section. 

\begin{lemma}[\cite{CharikarG99}] \label{cont}
	We have $\|\al^1-\al^2\|_\infty\leq 2^n(\ld_2-\ld_1)\leq\frac{\e\lb}{n^2}$. 
	Hence, for any $i\in F_1\cup F_2$, and any $r\in\{1,2\}$, we have 
	$\sum_j\beta^r_{ij}\geq\ld_2-\frac{\e\lb}{n}$.
\end{lemma}

For every $i\in\F$, $j\in\D$, define $\al_j:=\max\{\al^1_j,\al^2_j\}$, and
$\beta_{ij}:=\max\{\beta^1_{ij},\beta^2_{ij}\}$; note that 
$\beta_{ij}=\bigl(\al_j-h_{\vec{t}}(c_{ij})\bigr)^+$.

To obtain the improvement, we utilize insights from the $4$-approximation algorithm for
$k$-median in~\cite{CharikarG99}. The idea is to first augment $F_1$ using facilities from
$F_2$ (that are approximately paid for by $(\al^1,\beta^1)$, and then open facilities in a
similar manner as before. The augmentation step will ensure that for
every client $j$, there is some facility $i$ that is opened with 
$h_{5\vec{t}}(c_{ij})\leq 5\al_j$, and this leads to the $5$-approximation guarantee.

\begin{enumerate}[label=D\arabic*., ref=D\arabic*, topsep=0.5ex, itemsep=0ex,
	labelwidth=\widthof{D3.}, leftmargin=!] 
	\item {\bf Augmenting \boldmath $F_1$.}\ Augment $F_1$ to a maximal set 
	$F'_1\supseteq F_1$ by adding facilities from $F_2$ while preserving the following
	property: for every $j\in\D$, there is at most one $i\in F'_1$ with $\beta^1_{ij}>0$. 
	For every $j\in\D$, redefine $i_1(j)$ to be the facility in $F'_1$ that is closest (in
	terms of $c_{ij}$) to $j$. \label{d1}

	\item Let $k'_1=|F'_1|$, $k_2=|F_2|$.
	For every $i\in F_2$, let $\sg(i)\in F'_1$ denote the facility in $F'_1$ closest to
	$i$ (which will be $i$ if $i\in F'_1$). Let $\bF_1\sse F'_1$ be an arbitrary set such that
	$\sg(F_2)\sse\bF_1$ and $|\bF_1|=k_2$.
	
	\item {\bf Opening facilities.}\ As before, we will open either all facilities in $\bF_1$
	or all facilities in $F_2$, and we will also open $k-k_2$ facilities from $F'_1\sm\bF_1$. 
	To do this, we utilize an LP with the same variables and constraints as \eqref{roundlp}:
	variable $\tht$ to indicate if we open the facilities in $\bF_1$, and variables $z_i$ for
	every $i\in F'_1\sm\bF_1$ to indicate if we open facility $i$. 
	But we use a different objective function.
	For each client $j$, we define an expression 
	$A_j\bigl(\tht,z:=\{z_i\}_{i\in F'_1\sm\bF_1}\bigr)$ that will serve as an upper
	bound on $h_{5\vec{t}}(\text{assignment cost of $j$})$ when $\tht$ and 
	$z$ are integral, and our LP will seek to minimize $\sum_j A_j(\tht,z)$. 
	Define
	$$
	A_j(\tht,z) = 
	\begin{cases}
	\tht h_{\vec{t}}(c_{i_1(j)j})+(1-\tht)h_{\vec{t}}(c_{i_2(j)j}) & 
	i_1(j)\in\bF_1,\ j\in\pay^1(F'_1)\cap\pay^2(F_2); \\
	
	h_{\vec{t}}(c_{i_1(j)j})+(1-z_{i_1(j)})\cdot 2h_{\vec{t}}(c_{i_2(j)j}) & 
	i_1(j)\notin\bF_1,\ j\in\pay^1(F'_1)\cap\pay^2(F_2); \\
	
	(1-\tht)h_{\vec{t}}(c_{i_2(j)j})+\tht\cdot 5\al_j & 
	j\in\pay^2(F_2)\sm\pay^1(F'_1); \\
	
	(1-\tht)h_{3\vec{t}}(c_{i_2(j)j})+\tht\cdot 5\al_j &
	j\notin\pay^1(F'_1)\cup\pay^2(F'_2); \\
	
	\tht\cdot h_{\vec{t}}(c_{i_1(j)j})+(1-\tht)\cdot 5\al_j &
	i_1(j)\in\bF_1,\ j\in\pay^1(F'_1)\sm\pay^2(F'_2); \\
	
	z_{i_1(j)}\cdot h_{\vec{t}}(c_{i_1(j)j})+(1-z_{i_1(j)})\cdot 5\al_j &
	i_1(j)\notin\bF_1,\ j\in\pay^1(F'_1)\sm\pay^2(F'_2);
	\end{cases}
	$$
	We solve the following LP:
	\begin{equation}
	\min \ \ \sum_j A_j(\tht,z) \qquad \text{s.t.} \qquad
	\sum_{i\in F_1\sm\bF_1}z_i \leq k-k_2, \quad 
	\tht\in[0,1], \ \ z_i\in[0,1] \ \ \forall i\in F'_1\sm\bF_1. \tag {O-P} \label{newlp}
	\end{equation}
	The above LP is integral, and we open the facilities specified by an integral optimal
	solution (as discussed above), and assign each client to the nearest open facility.
\end{enumerate}

\vspace*{-1ex}
\paragraph{Analysis.} 
The road map of the analysis is as follows.
Recall that $\al_j=\max\{\al^1_j,\al^2_j\}$ and 
$\beta_{ij}=\max\{\beta^1_{ij},\beta^2_{ij}\}$. 
We first show that by combining ~\Cref{cont} and ~\Cref{lmp}, we can infer two
things (see ~\Cref{newlmp}): (1) for both the $F'_1$ and $F_2$ solutions, 
$\sum_j 3\al_j$ can be used to pay for the $\ld_2$-cost of all open facilities and 
$\sum_j h_{3\vec{t}}(\text{assignment cost of $j$})$; (2) for every client $j$, due to
our augmentation step~\ref{c1}, we have facilities $i\in F_2$, $i'\in F'_1$ such that
$i$ is close to $j$, and $i'$ is close to $i$. 

Next, we show that the optimal value of \eqref{newlp} is (roughly) at most
$5\OPT$ (~\Cref{newlpbnd}). Finally, we show that if we have an integral
solution $(\ttht,\tz)$ to \eqref{newlp}, then this yields a solution 
$\sum_j h_{5\vec{t}}(\text{assignment cost of $j$})$ is (roughly) bounded by
$\sum_jA_j(\ttht,\tz)$ (~\Cref{intbnd}). 
Here, we use property (2) above to argue that for every client $j$, there is some facility
$i$ opened in our final solution with $h_{5\vec{t}}(c_{ij})$ bounded by (roughly) $5\al_j$. 
Combining Lemmas~\ref{newlpbnd} and~\ref{intbnd} yields ~\Cref{5apxthm}.

\begin{lemma} \label{newlmp}
	The following hold.
	\begin{enumerate}[(i), topsep=0.25ex, itemsep=0.25ex, parsep=0ex,
		labelwidth=\widthof{(iii)}, leftmargin=!]
		\item $3\ld_2|F'_1|+\sum_{j\in\pay^1(F'_1)}3h_{\vec{t}}(c_{i_1(j)j})
		+\sum_{j\notin\pay^1(F'_1)}h_{3\vec{t}}(c_{i_1(j)j})\leq 3\sum_j\al_j+3\e\lb$. 
		
		\item $3\ld_2|F_2|+\sum_{j\in\pay^2(F_2)}3h_{\vec{t}}(c_{i_2(j)j})
		+\sum_{j\notin\pay^2(F_2)}h_{3\vec{t}}(c_{i_2(j)j})\leq 3\sum_j\al_j$. 
		
		\item For any $j\in\D$, there are facilities $i\in F_2$ and $i'\in F_1$ such that 
		$h_{3\vec{t}}(c_{ij})\leq 3\al_j$, and $h_{2\vec{t}}(c_{ii'})\leq 2\al_j+\frac{2\e\lb}{n^2}$.
	\end{enumerate}
\end{lemma}

\begin{proof}
	Part (ii) follows immediately from part (i) of ~\Cref{lmp}.
	
	Consider part (i). Since $F'_1\sse F_1\cup F_2$, by ~\Cref{cont}, for every 
	$i\in F'_1$, we have that $\sum_j\beta^1_{ij}\geq\ld_2-\frac{\e\lb}{n}$.
	When adding facilities to $F_1$ in step~\ref{c1}, we ensure that the sets
	$\{\pay^1(i)\}_{i\in F'_1}$ remain pairwise disjoint. For every client $j$, we know that
	if $\beta^1_{ij}>0$ for some $i\in F'_1$, then $i_1(j)=i$; we also know from part
	(ii) of ~\Cref{lmp} that $h_{3\vec{t}}(c_{i_1(j)j})\leq 3\al^1_j$. So
	\begin{equation*}
	\begin{split}
	\sum_j 3\al_j\geq \sum_j 3\al^1_j
	& \geq \sum_{i\in F'_1}\sum_{j\in\pay^1(i)} 3\bigl(\beta^1_{ij}+h_{\vec{t}}(c_{ij})\bigr)
	+\sum_{j\notin\pay^1(F'_1)} h_{3\vec{t}}(c_{i_1(j)j}) \\
	& 3\ld_2|F'_1|-\frac{3|F'_1|\e\lb}{n}+\sum_{j\in\pay^1(F'_1)}3h_{\vec{t}}(c_{i(j)j})
	+\sum_{j\notin\pay^1(F'_1)} h_{3\vec{t}}(c_{i_1(j)j}).
	\end{split}
	\end{equation*}
	
	To prove part (iii), consider any client $j$. By ~\Cref{lmp} (ii), we know that
	there is some $i\in F_2$ such that $h_{3\vec{t}}(c_{ij})\leq 3\al^2_j\leq 3\al_j$, and
	$\al^2_j\geq\al^2_k$ for all $k\in\pay^2(i)$. If $i\in F'_1$, then taking $i'=i$ finishes
	the proof. Otherwise, since $i$ was not added to $F'_1$ in step~\ref{c1}, it must be that
	there is some client $k$ and some facility $i'\in F_1$ such that
	$\beta^1_{ik},\beta^1_{i'k}>0$. So we have 
	$$
	h_{2\vec{t}}(c_{ii'})\leq h_{\vec{t}}(c_{ik})+h_{\vec{t}}(c_{i'k})\leq 2\al^1_k
	\leq 2\al^2_k+\frac{2\e\lb}{n^2}\leq 2\al^2_j+\frac{2\e\lb}{n^2}
	\leq 2\al_j+\frac{2\e\lb}{n^2}. \qedhere
	$$
\end{proof}

\begin{lemma} \label{newlpbnd}
	The optimal value of \eqref{newlp} is at most $5\OPT+5\e\lb\bigl(1+\frac{1}{n^2}\bigr)$.
\end{lemma}

\begin{proof}
	Let $a, b\geq 0$ be such that $ak'_1+bk_2=k$ and $a+b=1$. 
	Define $\charge_j=a\cdot 5\beta^1_{i_1(j)j}+b\cdot 5\beta^2_{i_2(j)j}$.
	Then, we have
	$$
	\sum_j\charge_j=a\sum_{j\in\pay^1(F'_1)}5\beta^1_{i_1(j)j}+b\sum_{j\in\pay^2(F_2)}5\beta^2_{i_2(j)j}
	\geq a\Bigl(5\ld_2k'_1-5\e\lb\Bigr)+b\cdot 5\ld_2k_2=5k\ld_2-5\e\lb
	$$
	where the inequality follows from ~\Cref{cont}.
	Set $\tht=a$ and $z_i=a$ for all $i\in F'_1\sm\bF_1$. 
	We show that $\charge_j+A_j(\tht,z:=\{z_i\}_{i\in F'_1\sm\bF_1})\leq 5\al_j$
	for every client $j$. This will complete the proof since this implies that
	$$
	5k\ld_2+\sum_j A_j(\tht,z)\leq 5\al_j+5\e\lb
	\leq 5\al^2_j+5\e\lb\Bigl(1+\tfrac{1}{n^2}\Bigr),
	$$
	and $\sum_j\al^2_j-k\ld_2\leq\OPT$ since $(\al^2,\beta^2,\ld_2)$ is a
	feasible solution to \eqref{dual}. 
	
	To prove the claim, consider any client $j$. 
	Recall that $a\geq 0.5$. Observe that:
	\begin{enumerate}[label=$\bullet$, topsep=0ex, itemsep=0ex, parsep=0ex]
		\item if $j\in\pay^1(F'_1)$, then $h_{\vec{t}}(c_{i_1(j)j})+\beta^1_{i_1(j)j}=\al^1_j$;
		\item if $j\in\pay^2(F_2)$, then $h_{\vec{t}}(c_{i_2(j)j})+\beta^2_{i_2(j)j}=\al^2_j$, and
		otherwise, we have $h_{3\vec{t}}(c_{i_2(j)j})\leq 3\al^2_j$.
	\end{enumerate}
	By considering each case in the definition of $A_j$, and plugging in the above
	bounds, we obtain the claimed bound on $\charge_j+A_j(\tht,z)$.
	%
\end{proof}

\begin{lemma} \label{intbnd}
	Let $(\ttht,\tz)$ be an integral solution to \eqref{newlp}. 
	Let $X_j$ denote its assignment cost under the resulting solution. 
	We have $\sum_j h_{5\vec{t}}(X_j)\leq \sum_j A_j(\ttht,\tz)+\frac{2\e\lb}{n}$.
\end{lemma}

\begin{proof}
	Consider any client $j$. We abbreviate $A_j(\ttht,\tz)$ to $A_j$. 
	We show that $h_{5\vec{t}}(X_j)\leq A_j+\frac{2\e\lb}{n^2}$, which will prove the lemma. 
	We first note the following. By ~\Cref{newlmp} (iii), there are facilities $i\in F_2$, 
	$i'\in F_1$ such that $h_{3\vec{t}}(c_{ij})\leq 3\al_j$ and 
	$h_{2\vec{t}}(c_{ii'})\leq 2\al_j+\frac{2\e\lb}{n^2}$. If $\ttht=1$, then we know that
	$\sg(i)$ is open. 
	Hence,
	$$
	h_{5\vec{t}}(X_j)\leq h_{5\vec{t}}(c_{\sg(i)j})\leq h_{3\vec{t}}(c_{ij})+h_{2\vec{t}}(c_{i\sg(i)})
	\leq h_{3\vec{t}}(c_{ij})+h_{2\vec{t}}(c_{ii'})\leq 5\al_j+\frac{2\e\lb}{n^2}.
	$$
	Consider each case in the definition of $A_j$. 
	%
	\begin{enumerate}[label=$\bullet$, itemsep=0.25ex, parsep=0ex,
		labelwidth=\widthof{$\bullet$}, leftmargin=!] 
		\item $i_1(j)\in\bF_1,\ j\in\pay^1(F'_1)\cap\pay^2(F_2)$.
		If $\ttht=1$, then $i_1(j)$ is open, and if $\ttht=0$, then $i_2(j)$ is open, so
		$h_{\vec{t}}(X_j)\leq A_j$.
		
		\item $i_1(j)\notin\bF_1,\ j\in\pay^1(F'_1)\cap\pay^2(F_2)$.
		If $\tz_{i_1(j)}=1$, then the bound clearly holds. 
		Otherwise, either $i_2(j)$ is open, or $i:=\sg(i_2(j))$ is open. We have 
		$c_{ii_2(j)}\leq c_{i_1(j)i_2(j)}\leq c_{i_1(j)j}+c_{i_2(j)j}$, and so 
		$X_j\leq 2c_{i_2(j)j}+c_{i_1(j)j}$ holds in both cases. 
		So $h_{3\vec{t}}(X_j)\leq h_{2\vec{t}}(2c_{i_2(j)j})+h_{\vec{t}}(c_{i_1(j)j})=A_j$.
		
		\item $j\in\pay^2(F_2)\sm\pay^1(F'_1)$. If $\ttht=0$, clearly $h_{\vec{t}}(X_j)\leq A_j$.
		Otherwise, as shown above, we have 
		$h_{5\vec{t}}(X_j)\leq 5\al_j+\frac{2\e\lb}{n^2}=A_j+\frac{2\e\lb}{n^2}$.
		
		\item $j\notin\pay^1(F'_1)\cup\pay^2(F'_2)$. If $\ttht=0$, then $i_2(j)$ is open and
		$h_{3\vec{t}}(X_j)\leq A_j$. Otherwise, as above, we have 
		$h_{5\vec{t}}(X_j)\leq A_j+\frac{2\e\lb}{n^2}$. 
		
		\item $i_1(j)\in\bF_1,\ j\in\pay^1(F'_1)\sm\pay^2(F'_2)$. If $\ttht=1$, then clearly 
		$h_{\vec{t}}(X_j)\leq A_j$. 
		Otherwise, $i_2(j)$ is open, and $h_{3\vec{t}}(X_j)\leq 3\al_j\leq A_j$. 
		
		\item $i_1(j)\notin\bF_1,\ j\in\pay^1(F'_1)\sm\pay^2(F'_2)$. If $\tz_{i_1(j)}=1$, then
		clearly $h_{\vec{t}}(X_j)\leq A_j$. Otherwise, if $\ttht=0$, then $i_2(j)$ is open, and
		$h_{3\vec{t}}(X_j)\leq h_{3\vec{t}}(c_{i_2(j)j})\leq 3\al_j\leq A_j$.
		If $\ttht=1$, then as shown at the beginning, we have 
		$h_{5\vec{t}}(X_j)\leq 5\al_j+\frac{2\e\lb}{n^2}=A_j+\frac{2\e\lb}{n^2}$. 
	\end{enumerate}
	
	\vspace*{-3ex}
\end{proof}

\begin{proofof}{{\bf Finishing up the proof of ~\Cref{5apxthm}}}
	Let $X_j$ be the assignment cost of client $j$ in the solution returned.
	Combining Lemmas~\ref{newlpbnd} and~\ref{intbnd}, we obtain that 
	$\sum_j h_{5\vec{t}}(X_j)\leq 5\OPT+\e\lb\bigl(5+\frac{2}{n}+\frac{5}{n^2}\bigr)$.
	Since $\OPT\leq\sum_j h_{\vec{t}}(\vo_j)$, combining this with ~\Cref{proxyub}
	shows that $\obj(w;\bullet)$-cost of the solution returned is at most 
	$5(1+\e)(1+2\e)\iopt+(1+\e)\e\lb\bigl(5+\frac{2}{n}+\frac{5}{n^2}\bigr)$.
\end{proofof} 

\begin{proofof}{~\Cref{cont}}
	We mimic the proof in~\cite{CharikarG99}. We use $x_-$ to denote a quantity
	infinitesimally smaller than $x$. Let $\dt=\ld_2-\ld_1$.
	Sort the clients in increasing order of their $\al^0_j:=\min\{\al^1_j,\al^2_j\}$ value. So
	$\al^0_1\leq\ldots\leq\al^0_n$. We prove that $|\al^1_j-\al^2_j|\leq 2^{j-1}\dt$ for all
	$j=1,\ldots,n$, which implies the lemma.
	
	We proceed by induction on $j$. Consider running the dual-ascent phase of the primal-dual
	algorithm for $\ld=\ld_1$ and $\ld=\ld_2$ in parallel. For the base case, suppose that
	$\al^0_1=\al^r_1$, where $r\in\{1,2\}$. Consider the time point $\tm=\al^0_1$ in the two
	executions. By definition, at time $\tm_-$, all clients are active in the two executions.
	So at time $\tm$, we have $\al^1_j=\al^2_j=t$ for all $j$, and so 
	$\beta^1_{ij}=\beta^2_{ij}$ for all $i, j$.  
	Client $1$ froze in execution $r$ at time $t$, because at that time it can reach some
	facility $f$ 
	for which constraint \eqref{fpay} became tight at time $\tm$; we say that $f$ got paid
	for at time $t$ (in the execution $r$). Let $\br=2-r$. 
	We have $\sum_j\beta^{\br}_{fj}=\sum_j\beta^r_{fj}$ at time $\tm$, so $\sum_j\beta^{\br}_{fj}$
	can increase by at most $\dt$ beyond time $t$. Hence, $\al^{\br}_1$ can increase by at most
	$\dt$ beyond time $\tm$ (since any increase in $\al^{\br}_1$ translates to the same
	increase in $\beta^{\br}_{f1}$ as $\al^{\br}_1\geq h_{\vec{t}}(c_{f1})$ at time $\tm$).
	
	Suppose we have shown that $|\al^1_j-\al^2_j|\leq 2^{j-1}\dt$ for all $j=1,\ldots,\ell-1$
	(where $\ell\geq 2$). Now consider client $\ell$. The induction step follows from a
	similar argument. Consider time point $\tm=\al^0_\ell$ in
	both executions. By definition, all clients $j\geq\ell$ are active at time $\tm_-$ in the
	two executions. So at time $\tm$, we have $\al^1_j=\al^2_j=\tm$ for all $j\geq\ell$. 
	Suppose $\al^0_\ell=\al^r_\ell$,
	where $r\in\{1,2\}$, and let $\br=2-r$. In execution $r$, client $\ell$ froze at time $t$
	due to some facility $f$, where either: (1) $f$ was paid for by time $\tm$, and $\ell$
	reached $f$ at time $\tm$; or (2) $f$ got paid for at time $\tm$, and $\ell$ reached
	$f$ at or before time $\tm$.
	At time $\tm$, we have $\beta^{\br}_{fj}\geq\beta^r_{fj}-2^{j-1}\dt$ for all $j<\ell$ by the
	induction hypothesis, and $\beta^1_{fj}=\beta^2_{fj}$ for all $j\geq\ell$. 
	Therefore, the contribution $\sum_j\beta^{\br}_{fj}$ from clients to the LHS of \eqref{fpay}
	at time $t$ is at least $\ld_1-\sum_{j=1}^{\ell-1}2^{j-1}\dt$. So this contribution can
	increase by at most $\dt+\sum_{j=1}^{\ell-1}2^{j-1}\dt=2^{\ell-1}\dt$ beyond time $\tm$ in
	execution $\br$. So since $\al^{\br}_\ell=\al^r_\ell\geq h_{\vec{t}}(c_{f\ell})$ at time $\tm$,
	it follows that $\al^{\br}_\ell$ can increase by at most $2^{\ell-1}\dt$ beyond time $\tm$.
\end{proofof}

\bibliography{stoc-centridian}
\bibliographystyle{plain}

\end{document}